\documentclass[]{article}

\usepackage{amsmath}                    
\usepackage[bbgreekl]{mathbbol}
\DeclareSymbolFontAlphabet{\mathbbl}{bbold}

\usepackage{amsfonts,amssymb,amsthm,mathtools}    
\DeclareSymbolFontAlphabet{\mathbbm}{bbold}
\DeclareSymbolFontAlphabet{\mathbb}{AMSb}%
\usepackage{tensor,mathrsfs}
\usepackage[breaklinks=true,backref=none]{hyperref}  
\usepackage[utf8]{inputenc} 
\usepackage[T1]{fontenc}
\usepackage{geometry} 				
\usepackage{lmodern}
\usepackage{enumerate}
\usepackage{titlesec}
\usepackage{bm}
\usepackage[dvipsnames]{xcolor}            
\usepackage[customcolors]{hf-tikz}
\usepackage{array}
\usepackage{titling}
\usepackage{authblk}
\usepackage[numbers,sort&compress]{natbib}
\usepackage{scalerel}

\graphicspath{{./font/}{./}}

\newcommand{\dd}{\mathchoice
	{\mathbbm{d}\rrule{.087ex}{1.605ex}\hspace*{0.15ex}} 
	{\mathbbm{d}\rrule{.087ex}{1.605ex}\hspace*{0.15ex}} 
	{\mathbbm{d}\rrule{.08ex}{1.125ex}\hspace*{0.15ex}}  
	{\mathbbm{d}\rrule{.06ex}{.8ex}\hspace*{0.15ex}}     
}

\newlength{\alturaL}
\newcommand{\LL}{\includegraphics[height=1.1\alturaL]{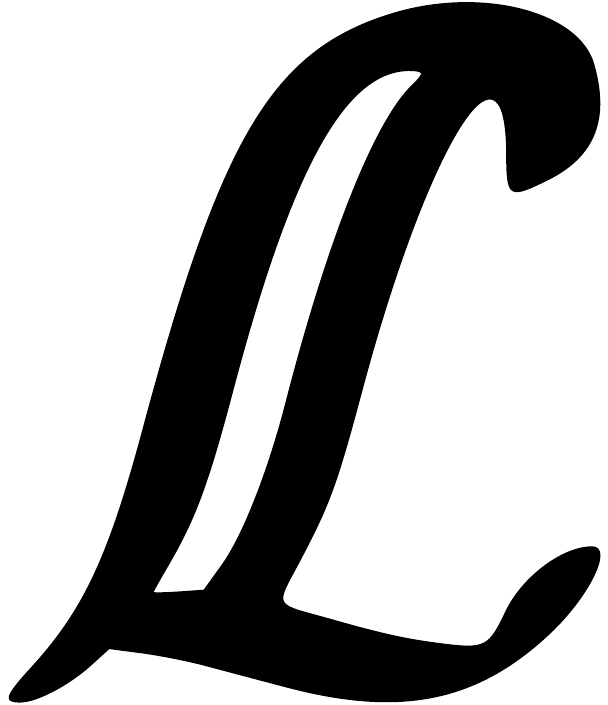}}

\newlength{\alturaO}
\newcommand{\OOmega}{\includegraphics[height=\alturaO]{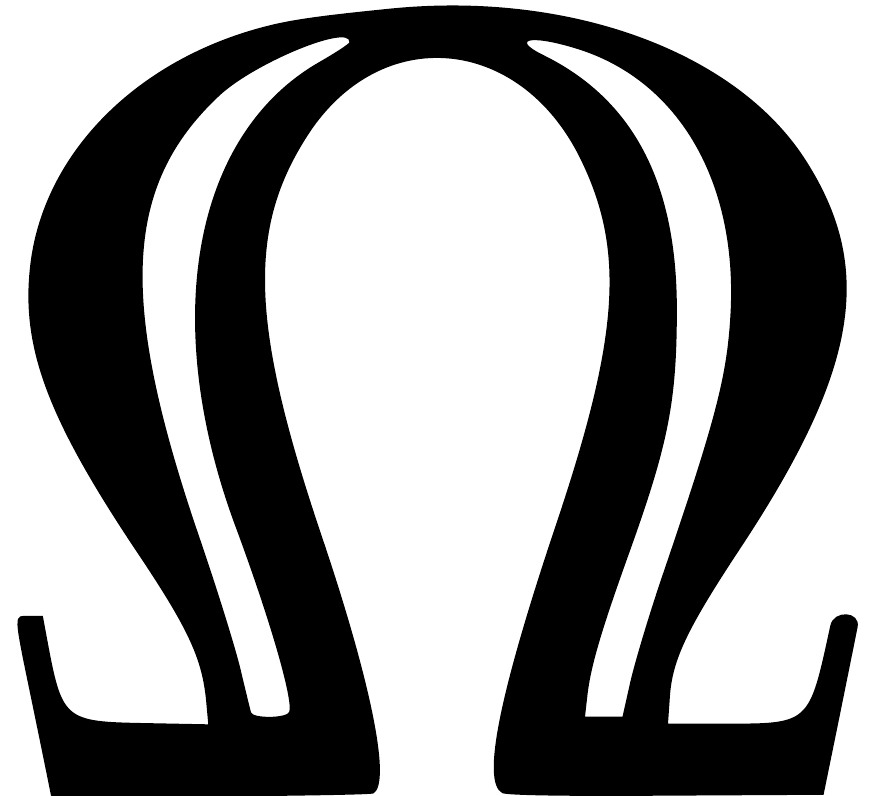}}

\newlength{\alturaI}
\newcommand{\ii}{\raisebox{-.04ex}{\includegraphics[height=1.1\alturaI]{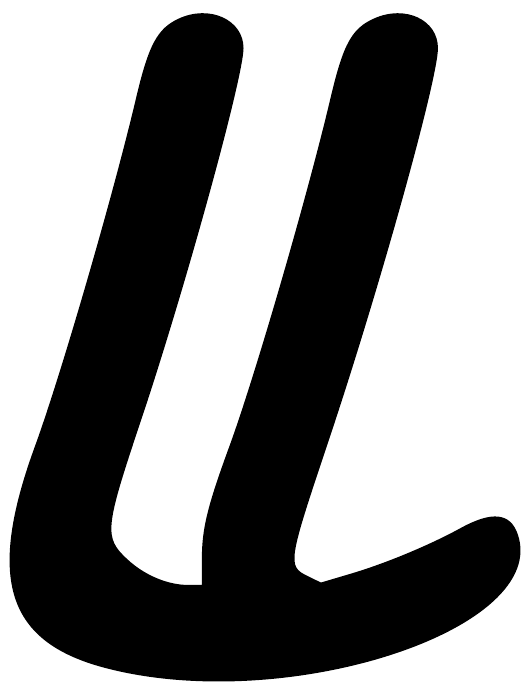}}}

\newlength{\alturaJ}
\newcommand{\jj}{\raisebox{-.37ex}{\includegraphics[height=.9\alturaJ]{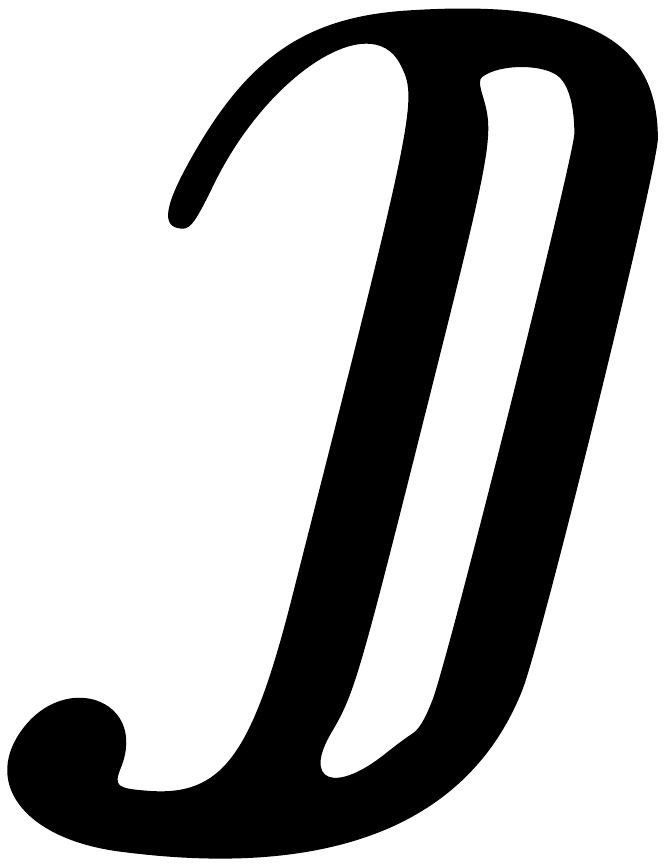}}}

\newlength{\alturaX}
\newcommand{\campos}{\raisebox{-.075ex}{\includegraphics[height=1.05\alturaX]{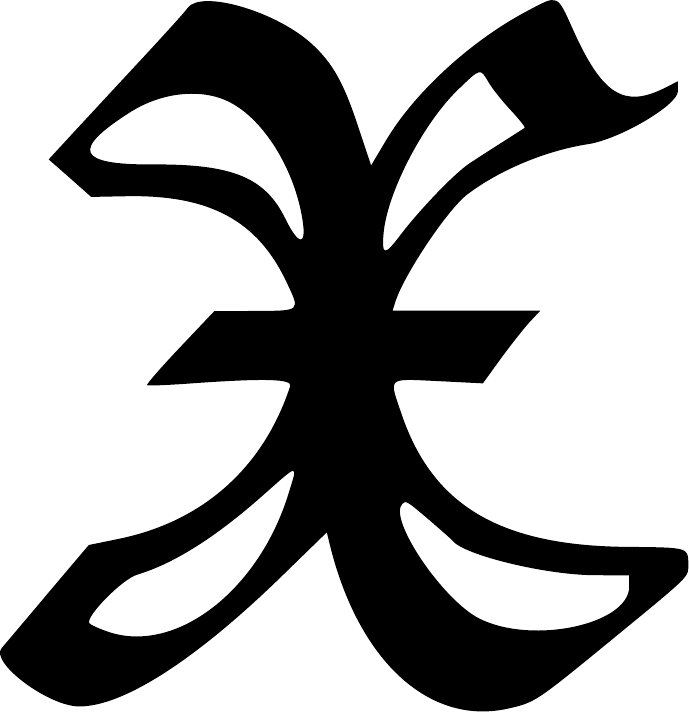}}}

\newcommand{\definicion}{\hfsetfillcolor{blue!5}\hfsetbordercolor{blue}}
\newcommand{\degeneracion}{\hfsetfillcolor{red!5}\hfsetbordercolor{red!50!black}}
\newcommand{\comentario}{\hfsetfillcolor{green!5}\hfsetbordercolor{green!50!black}}

\usepackage{xpatch}
\makeatletter
\AtBeginDocument{\xpatchcmd{\@thm}{\thm@headpunct{.}}{\thm@headpunct{}}{}{}}
\makeatother

\usepackage[nottoc]{tocbibind}          
\usepackage{pstricks}
\usepackage{tikz}
\newcommand*\circled[1]{\tikz[baseline=(char.base)]{
		\node[shape=circle,draw,inner sep=2pt] (char) {#1};}}
\usetikzlibrary{cd}
\usepackage[title]{appendix}

\usepackage{mathtools}
\usepackage{enumitem}
\usepackage{verbatim}
\usepackage{titletoc}
\titlecontents{section}[3em]{\addvspace{10pt}\bfseries}
{\contentslabel{2.5em}}{\hspace{-2.5em}}{\hfill\contentspage}
\dottedcontents{subsection}[6.3em]{}{3.3em}{0.7pc}
\dottedcontents{subsubsection}[10.5em]{}{4.2em}{0.7pc}

\titlespacing*{\subsubsection}{0pt}{2.2ex plus 1ex minus .2ex}{0.1ex plus .2ex}
\titlespacing*{\subsection}{0pt}{1.7ex plus 1ex minus .2ex}{0.1ex plus .2ex}
\titlespacing*{\section}{0pt}{.8ex plus 1ex minus .2ex}{0.1ex plus .2ex}
\titleformat*{\subsubsection}{\Large\bfseries\sffamily}
\titleformat*{\subsection}{\LARGE\bfseries\sffamily}
\titleformat*{\section}{\centering\huge\bfseries\sffamily}

\newcommand{\corurl}{BrickRed}  \newcommand{\corcite}{red}
\newcommand{\corlink}{blue}    \newcommand{\corfile}{black}

\hypersetup{
	linktocpage,
	colorlinks,
	urlcolor=\corurl,	
	citecolor=\corcite,
	linkcolor=\corlink,
	filecolor=\corfile,
	pdfnewwindow,
	pdftitle={Covariant space methods},
	pdfauthor={Juan Margalef Bentabol},
	pdfsubject={CPS}}

\allowdisplaybreaks  
\makeatletter
\newcommand{\pushright}[1]{\ifmeasuring@#1\else\omit\hfill$\displaystyle#1$\fi\ignorespaces}
\newcommand{\pushleft}[1]{\ifmeasuring@#1\else\omit$\displaystyle#1$\hfill\fi\ignorespaces}
\makeatother

\newcommand{\updown}[3]{\overset{#1}{\underset{#2}{#3}}}   

\makeatletter
\newcommand{\raisemath}[1]{\mathpalette{\raisem@th{#1}}}
\newcommand{\raisem@th}[3]{\raisebox{#1}{$#2#3$}}
\makeatother
\newcommand{\peqsub}[2]{#1_{\raisemath{.2pt}{\hspace*{-0.1ex}\scriptscriptstyle #2}\hspace*{-0.2ex}}}
\newcommand{\peqsubfino}[3]{#1_{\raisemath{#3}{\hspace*{-0.1ex}\scriptscriptstyle #2}\hspace*{-0.2ex}}}

\makeatletter
\DeclareRobustCommand{\rvdots}{%
	\vbox{
		\baselineskip4\p@\lineskiplimit\z@
		\kern-\p@
		\hbox{.}\hbox{.}\hbox{.}
}}
\makeatother

\newcommand{\N}{\mathbb{N}}     
\newcommand{\R}{\mathbb{R}}     
\renewcommand{\S}{\mathbb{S}}   
\newcommand{\HH}{\mathbbm{H}}     

\def\QED{{\boldmath$\rule{0.5em}{0.5em}$}}                                
\def\markatright#1{\leavevmode\unskip\nobreak\quad\hspace*{\fill}{#1}}    
\def\qed{\markatright{\QED}}                                              

\newtheorem{theorem}{Theorem}[section]
\newcommand{\lateral}{\peqsub{\partial}{L}}

\makeatletter
\let\c@equation\c@theorem
\makeatother
\usepackage{chngcntr}
\counterwithin{equation}{section}
\newtheorem{definition}[theorem]{Definition}       
                   \newtheorem{proposition}[theorem]{Proposition}
\newtheorem{remark}[theorem]{Remark}
\newtheorem{lemma}[theorem]{Lemma}
\newtheorem*{lemma*}{Lemma}                        \newtheorem{corollary}[theorem]{Corollary}

\renewenvironment{proof}[1][Proof]{\textbf{#1.} }{\qed\\}     

\theoremstyle{definition}

\usepackage{graphicx}
\makeatletter
\newcommand\rrule[3][0pt]{%
	\ifdim#2>#3\math@hrule[#1]{#2}{#3}\else\math@vrule[#1]{#2}{#3}\fi}
\newcommand\math@hrule[3][0pt]{%
	\gdef\mystery@factor{0.07}%
	\@tempdima=#3%
	\rule[#1]{0pt}{#3}
	\raisebox{.5\@tempdima+#1}{%
		\makebox[#2][l]{\kern-.5\@tempdima\@@mathrule{#2}{#3}}}%
}
\newcommand\math@vrule[3][0pt]{%
	\gdef\mystery@factor{0.0}%
	\@tempdima=#2%
	\rule[#1]{0pt}{#3}
	\raisebox{-.0\@tempdima+#1}{%
		\kern0.5\@tempdima%
		\rotatebox{90}{\kern-0.5\@tempdima\makebox[#3][l]{\@@mathrule{#3}{#2}}}%
		\kern0.5\@tempdima}%
}
\def\@@mathrule#1#2{%
	\@tempdimb=#2%
	\@tempdima=\dimexpr#1-\mystery@factor\@tempdimb
	\pdfliteral{%
		q []0 d %
		1 J 
		\strip@pt\@tempdimb\space w \strip@pt\@tempdimb\space 0 m %
		\strip@pt\@tempdima\space 0 l S Q }}
\makeatother

\makeatletter
\DeclareFontFamily{OMX}{MnSymbolE}{}
\DeclareSymbolFont{MnLargeSymbols}{OMX}{MnSymbolE}{m}{n}
\SetSymbolFont{MnLargeSymbols}{bold}{OMX}{MnSymbolE}{b}{n}
\DeclareFontShape{OMX}{MnSymbolE}{m}{n}{
	<-6>  MnSymbolE5
	<6-7>  MnSymbolE6
	<7-8>  MnSymbolE7
	<8-9>  MnSymbolE8
	<9-10> MnSymbolE9
	<10-12> MnSymbolE10
	<12->   MnSymbolE12
}{}
\DeclareFontShape{OMX}{MnSymbolE}{b}{n}{
	<-6>  MnSymbolE-Bold5
	<6-7>  MnSymbolE-Bold6
	<7-8>  MnSymbolE-Bold7
	<8-9>  MnSymbolE-Bold8
	<9-10> MnSymbolE-Bold9
	<10-12> MnSymbolE-Bold10
	<12->   MnSymbolE-Bold12
}{}

\let\llangle\@undefined
\let\rrangle\@undefined
\DeclareMathDelimiter{\llangle}{\mathopen}%
{MnLargeSymbols}{'164}{MnLargeSymbols}{'164}
\DeclareMathDelimiter{\rrangle}{\mathclose}%
{MnLargeSymbols}{'171}{MnLargeSymbols}{'171}
\makeatother

\newcommand{\smallcirc}{\mathbin{\text{\raisebox{0.185ex}{\scalebox{0.7}{$\circ$}}}}}

\newcommand{\F}{\mathcal{F}}
\newcommand{\E}{\hspace*{-0.1ex}E}
\renewcommand{\d}{\mathrm{d}}
\newcommand{\DD}{\mathbb{D}}
\renewcommand{\SS}{\mathbb{S}}
\renewcommand{\L}{\mathcal{L}}
\newcommand{\XX}{\mathbb{X}}
\newcommand{\QQ}{\mathbb{Q}}
\newcommand{\YY}{\mathbb{Y}}
\newcommand{\VV}{\mathbb{V}}
\newcommand{\WW}{\mathbb{W}}
\newcommand{\ZZ}{\mathbb{Z}}
\newcommand{\FF}{\mathbb{F}}

\newcommand{\Sym}{\mathbbm{Sym}}
\newcommand{\Sol}{\mathrm{Sol}}

\newcommand{\Lag}{\mathrm{Lag}}
\newcommand{\CDS}{\mathrm{CDS}}

\newcommand{\prolong}{\mathrm{prol}}
\newcommand{\vol}{\mathrm{vol}}

\def\equivintt{{\setbox0\hbox{\ensuremath{\mathrel{\equiv}}}\rlap{\hbox to \wd0{\hss\ensuremath\int\hss}}\box0}}

\newcommand{\equivint}{\mathrel{\equivintt}}

\def\uptolidss{{\setbox0\hbox{\ensuremath{\mathrel{=}}}\rlap{\hbox to \wd0{\hss\raisebox{1.2ex}{\ensuremath{\scriptscriptstyle\mathcal{L}}}\hss}}\box0}}

\newcommand{\uptolids}{\mathrel{\uptolidss}}

\def\wwedgee{{\setbox0\hbox{\ensuremath{\mathrel{\wedge}}}\rlap{\hbox to \wd0{\hss\,\ensuremath\wedge\hss}}\box0}}

\newcommand{\wwedge}{\mathrel{\wwedgee}}

\title{Geometric formulation of the Covariant Phase Space methods with boundaries}
\author[1,3]{Juan Margalef-Bentabol}
\author[2,3]{Eduardo J.S.~Villaseñor}
\affil[1]{Institute for Gravitation and the Cosmos and Physics Department. Penn State
	University. PA 16802, USA.
	\vspace*{1ex} \mbox{}}
\affil[2]{Departamento de Matemáticas, Universidad Carlos III de Madrid. Avda. de la
	Universidad 30, 28911 Leganés, Spain.
	\vspace*{1ex} \mbox{}}
\affil[3]{Grupo de Teorías de Campos y Física Estadística. Instituto Gregorio Millán (UC3M).
	Unidad Asociada al Instituto de Estructura de la Materia, CSIC, Madrid, Spain.}
\date{}

\begin{document}
	\thispagestyle{empty}
	\renewcommand{\thepage}{Cover}
	
	\mbox{}\vspace*{6ex}
	
	\centerline{\begin{minipage}{1.1\linewidth}
			\begin{center}
				\Huge{\textsf{\textbf{\thetitle}}}
			\end{center}
	\end{minipage}}
	
	\mbox{}\vspace*{1.5ex}
	
	\begin{center}
		
		\mbox{}
		
		\Large{Juan Margalef-Bentabol${}^{1,3}$\qquad Eduardo J.S.~Villaseñor${}^{2,3}$}\\\vspace*{2.5ex}
		
		\small{${}^1$ Institute for Gravitation and the Cosmos and Physics Department. Penn State
			University. PA 16802, USA}\\[1.5ex]
		\small{${}^2$ Departamento de Matemáticas, Universidad Carlos III de Madrid. Avda. de la Universidad 30, 28911 Leganés, Spain.}\\[1.5ex]
		\small{${}^3$ Grupo de Teorías de Campos y Física Estadística. Instituto Gregorio Millán (UC3M). Unidad Asociada al Instituto de Estructura de la Materia, CSIC, Madrid, Spain}
	\end{center}
	
	\mbox{}
	
	\begin{abstract}\noindent\normalsize
		We analyze in full-detail the geometric structure of the covariant phase space (CPS) of any local field theory defined over a space-time with boundary. To this end, we introduce a new frame: the ``relative bicomplex framework''. It is the result of merging an extended version of the ``relative framework'' (initially developed in the context of algebraic topology by R.~Bott and L.W.~Tu in the 1980s to deal with boundaries) and the variational bicomplex framework (the differential geometric arena for the variational calculus). The relative bicomplex framework is the natural one to deal with field theories with boundary contributions, including corner contributions. In fact, we prove a formal equivalence between the relative version of a theory with boundary and the non-relative version of the same theory with no boundary. With these tools at hand, we endow the space of solutions of the theory with a (pre)symplectic structure canonically associated with the action and which, in general, has boundary contributions. We also study the symmetries of the theory and construct, for a large group of them, their Noether currents, and charges. Moreover, we completely characterize the arbitrariness (or lack thereof for fiber bundles with contractible fibers) of these constructions. This clarifies many misconceptions about the role of the boundary terms in the CPS description of a field theory. Finally, we provide what we call the CPS-algorithm to construct the aforementioned (pre)symplectic structure and apply it to some relevant examples.
	\end{abstract}

	\setcounter{tocdepth}{3}\setcounter{secnumdepth}{3}
	
	\newpage
	\thispagestyle{empty}\renewcommand{\thepage}{Table of contents}
	
	\tableofcontents\newpage
	
	\renewcommand{\thepage}{\arabic{page}}\setcounter{page}{1}
	
	\section{Introduction}
	The dynamical evolution of a given field theory on a globally hyperbolic $n$-space-time $M= [t_0,t_f]\times \Sigma$ is governed by its field equations together with some boundary conditions. If the theory is well-posed, we can evolve uniquely the initial data defined over a Cauchy surface, at least for a small interval of the evolution parameter. In general, a lot of interesting theories, such as general relativity or Yang-Mills, are only well-posed up to a gauge transformation, meaning that the evolution exists but is non-unique.\vspace*{2ex}
	
	If we consider the second-order Lagrangian framework, we have to define a space of field configurations $\mathcal{Q}$ over a Cauchy surface $\Sigma\subset M$ and consider its tangent bundle $T\mathcal{Q}$, which has no additional canonical structure. However, the cotangent bundle has a canonical symplectic structure which plays an essential role in the definition of the Hamiltonian framework \cite{tesis,margalef2016hamiltonian,chernoff2006properties}. This approach is excellent to understand the dynamical evolution of a system for a given time but it is not as well suited to understand some non-local concepts, such as the entropy of a black hole in general relativity \cite{wald1993black}. This can be achieved by considering the whole solution $\phi$ over $M$ instead of a curve $\{\varphi_\tau\}_\tau$ of fields $\mathcal{Q}$ over $\Sigma$.\vspace*{2ex}
	
	Let $\Sol(M)$ be the space of solutions of our theory and $\CDS(\Sigma)$ the Cauchy data set, formed by all admissible initial data that we can put over a fixed Cauchy surface. We assume, as it is the case in the Hamiltonian framework, that $\CDS(\Sigma)$ is a submanifold of $T^*\!\mathcal{Q}$ with the inclusion $\imath$. We can define a ``polarization'' map $\mathcal{P}:\Sol(M)\to\CDS(\Sigma)$ of the form $\mathcal{P}(\phi)=(q(\phi),p(\phi))\in T^*\!\mathcal{Q}$. Here, $q(\phi)\in\mathcal{Q}$ represents the initial position and $p(\phi)\in T^*_{q(\phi)}\mathcal{Q}$ the initial momentum. If we evolve them with the equations of motion of the theory, we recover, at least for a small evolution parameter, the solution $\phi$. This map is bijective if the theory is well-posed. If we have some degeneracy (gauge freedom), then the map is only surjective. In any case, we have the following diagram
	\[\Big(T^*\!\mathcal{Q},\Omega\Big)\qquad\longrightarrow\qquad\Big(\mathrm{CDS}(\Sigma),\imath^*\Omega\Big)\qquad\longrightarrow\qquad\Big(\Sol(M),\omega:=\mathcal{P}^*\imath^*\Omega\Big)\]
	The (pre)symplectic structure $\omega$ is fundamental in the study of some global issues but it is not easy to obtain. Besides, it requires a choice of polarization $\mathcal{P}$. It seems then desirable to have alternative methods to define a (pre)symplectic structure $\OOmega$ over $\Sol(M)$ and ways to check if it is equivalent to $\omega$. This is precisely the main goal of this work. Namely, in this paper, we develop the geometrical foundations to construct a (pre)symplectic structure over $\Sol(M)$ adapted to the physical problem. After that, we study some symmetries of the theory and find an interesting subset of them which turn out to be Hamiltonian. All this is done within the context of the so-called Covariant Phase Space (CPS) methods. Finally, we study several examples where $\OOmega$ is indeed equivalent to $\omega$. However, we will also show one particular example where both structures differ. It is important to mention at this point that more work is needed to understand fully this issue.\vspace*{2ex}
	
	The most important contribution of this paper is the introduction of what we refer to as the ``relative bicomplex framework''. In particular, we develop the ``relative framework'', which is the natural one to deal with boundaries, and merge it with the bicomplex framework, which is the natural one to deal with fields. In fact, we prove that all the results that hold in the non-relative framework without boundary, also hold in the relative framework with boundary. Moreover, it can easily be extended to include boundary corners. This should clarify many misconceptions involving the boundary and boundary terms in the CPS methods. Moreover, we completely characterize the arbitrariness of the aforementioned constructions. These results are also useful when no boundaries are present. Finally, we provide a simple four-steps algorithm that can be implemented for any local action theory to obtain the (pre)symplectic structure over $\Sol(M)$.
	
	\subsection{State of the art}
	A careful historical review can be found in the introductions of \cite{khavkine2014covariant,reyes2004covariant,ashtekar1991covariant}. Here we only consider some of the highlights of this area to shed some light about the motivation of this and similar works.\vspace*{2ex}
	
	The idea to consider the geometric structure of the space of fields can be traced all the way back to J.L.~Lagrange \cite{de1808memoire}. Since then, it has subsequently often reappeared. Just to mention a few, S.~Lie, J.G.~Darboux, E.~Cartan, or E.~Nother were interested in these kinds of problems. The more concrete idea of considering the space of solutions instead of the space of initial Cauchy data begins to appear in the physics literature circa 1960. To the best of our knowledge, the first explicit mention is due to Bergman and Kommar \cite{bergmann1962recent}, where they consider the \emph{frozen phase space}. After that, it is easier to find results in the mathematical journals (related to the inverse variational problems). For instance, in the 60s \cite{tonti1969variational}, in the 70s \cite{takens1977symmetries,takens1979global,krupka2015introduction}, and in the 80s \cite{vinogradov1984b,vinogradov1984c,tulczyjew1980euler,olver1995equivalence,anderson1989variational,tsujishita1982,marsden1986covariant}. Following these and similar references, one can see that some results have been (re)obtained over and over using different notations and frameworks. For some modern reviews and extended results see \cite{vitagliano2009secondary,vinogradov2006domains,moreno2013natural,vitolo1999different,bocharov1999symmetries}.\vspace*{2ex}
	
	Meanwhile, in the physics community, the CPS methods regained interest in the 80s with the paper of \v{C}.~Crnkovi\'{c} and E.~Witten \cite{crnkovic1987covariant}. Almost simultaneously, G.~Zuckerman \cite{zuckerman1987action} published a paper (halfway between the physics and mathematics literature) dealing with a very similar problem. After that, we can find many papers discussing the covariant phase space methods \cite{ashtekar1991covariant,crnkovic1987symplectic,crnkovic1988symplectic,lee1990local,wald1993black,iyer1994some,barnich2008surface,corichi2014hamiltonian,corichi2016actions,anderson1996classification}.\vspace*{2ex}
	
	The study of these methods for manifolds with boundaries is scarcer. The first one we are aware of was due to A.M.~Vinogradov \cite{vinogradov1984b,vinogradov1984c}, although it is not studied in much detail. In the physics literature, some boundary conditions at infinity were considered in \cite{wald1993black,ashtekar1991covariant,asktekarIsolated}, although the methods used are a bit \emph{ad-hoc} because, as we will see in section \ref{section: symplectic structure}, the symplectic structure used there is not the most natural one. After that, several attempts to consider exact counter-terms in the symplectic structure (that upon integration become boundary terms) were considered for example in \cite{compere2008setting,andrade2015stability} to obtain a more sensible symplectic structure. The first serious attempt to write a review about these results can be found in \cite{harlow2019covariant}, although it does not include almost any of the aforementioned mathematical references. This led the authors to (re)obtain some results that are obvious in a more geometric language e.g.~the inclusion of their $C$ term is completely natural as we will see in section \ref{section: variations}. Moreover, they fail to see that their (unnecessary) argument to justify the inclusion of the $C$ term also applies to their definition of symmetries, as we show in equation \eqref{eq: (L_X L,L_x l)}. The most relevant literature that deals with the geometry of field theories in manifolds with boundaries can be found in \cite{barbero2019generalizations,diaz2019dirac,margalef2016hamiltonianb,tesis,vinogradov2006domains,Barbero_G_2014}.
	
	\subsection{Structure of the paper}
	After this introduction, we proceed in section \ref{section: mathematical background} with a quick review of the mathematical results and notation needed for the rest of the paper. In particular, we develop the new ``relative bicomplex framework'' based in the relative pairs \cite{bot1982differential} and the variational bicomplex \cite{anderson1989variational}. The central part of this work is section \ref{section: space of fields}, where we study in detail the space of fields, the relevant objects that can be defined there, and the ambiguities that may arise. We also prove that the definition of the objects with geometric meaning, under very mild hypotheses, is not ambiguous. As a byproduct, in section \ref{section: summary of algorithm} we provide what we call the CPS-algorithm to define the canonical presymplectic structure over the space of solutions. Finally, section \ref{section: examples} is devoted to applying the CPS-algorithm to several relevant examples. Some additional results and material for the interested reader are included in the appendices. It is worth mentioning that this work is thought to have three different self-contained ways of reading it:
	\begin{enumerate}
		\item Sections \ref{section: summary of algorithm} and \ref{section: examples}, to get an idea of how the CPS-algorithm works in concrete examples.
		\item Sections \ref{section: mathematical background} and \ref{section: space of fields}, to understand the origin of the algorithm introduced in \ref{section: summary of algorithm} and used in \ref{section: examples}.
		\item Motivated readers are encouraged to read also appendix \ref{Appendix: jets}, where we formalize these concepts in the more natural (although arguably more cumbersome) language of jets.
	\end{enumerate}
	
	As a visual help, some of the formulas are framed with different boxes to highlight their origin:
	
	\begin{center}
		\definicion
		\tikzmarkin{ejemplo1prima}(.25,-0.2)(-0.03,0.45)\tikzmarkin{ejemplo1}(.2,-0.2)(-0.1,0.45)Definitions\tikzmarkend{ejemplo1}\tikzmarkend{ejemplo1prima}\qquad\qquad
		\degeneracion\tikzmarkin{ejemplo2}(.2,-0.2)(-0.1,0.45)Degeneracies\tikzmarkend{ejemplo2}\qquad\qquad
		\comentario\tikzmarkin{ejemplo3prima}(.2,-0.25)(.02,0.5)\tikzmarkin{ejemplo3}(.2,-0.2)(-0.1,0.45)Relevant results\tikzmarkend{ejemplo3}\tikzmarkend{ejemplo3prima}
	\end{center}
	
	\section{Mathematical background}\label{section: mathematical background}
	\subsection{Motivation}
	The space of fields $\F$ over $M$ is an infinite-dimensional manifold defined as the space of sections of a particular bundle. The topology and differential structure of $\F$ are usually taken for granted despite the subtleties inherent to any infinite-dimensional manifold. We somewhat agree with the common belief that quite often it is not necessary to be completely rigorous about it and one can proceed in analogy with the finite-dimensional case. However, we want to \emph{stress} two key facts:
	\begin{itemize}
		\item This is not always the case, especially in the presence of boundaries.
		\item Even accepting a reasonable lack of mathematical rigor, there is plenty of space to improve the formalism without making it more complicated, even in the presence of boundaries.
	\end{itemize}
	This section is precisely devoted to the latter point: introduce the required mathematical background to understand the rest of the paper using a language similar to the one most commonly used in the literature (treating $\F$ as an infinite-dimensional differential manifold where all the usual finite-dimensional manipulations are valid), but using a globally geometric approach. Regarding the first point, we have included in appendix \ref{Appendix: jets} a very basic introduction to the formalism of jet bundles. This language is the natural one to prove important results and perform computations. In this same appendix, we explain how to connect both formalisms. 
	
	\subsection[Differential geometry on \texorpdfstring{$M$}{M}]{Differential geometry on \texorpdfstring{$\boldsymbol{M}$}{M}}
	\subsubsection*{Tensor fields}
	Over an $n$-manifold $M$, we have its spaces of tensor fields i.e.~sections of the bundles $(T M)^{\otimes r}\otimes(T^*\!M)^{\otimes s}$. Taking the sections with $(r,s)=(1,0)$ leads to the space of vector fields $\mathfrak{X}(M)$, while $(r,s)=(0,1)$ leads to the space of $1$-form fields $\Omega^1(M)$. More generally, we have the graded algebra of form fields $\Omega(M)$ with the wedge product $\wedge$. The exterior derivative $\peqsub{\d}{M}$ over $\Omega(M)$ defines the complex
	\[0\longrightarrow\Omega^0(M)\overset{\peqsub{\d}{M}}{\longrightarrow}\Omega^1(M)\overset{\peqsub{\d}{M}}{\longrightarrow}\cdots\overset{\peqsub{\d}{M}}{\longrightarrow}\Omega^{n-1}(M)\overset{\peqsub{\d}{M}}{\longrightarrow}\Omega^n(M)\longrightarrow0\]Other important operations that we will use extensively are the Lie derivative $\L_V$ of any tensor field and the interior product $\iota_V$ of any form field, both of them with respect to a vector field $V\in\mathfrak{X}(M)$.\vspace*{2ex}
	
	For non-compact manifolds, we often have to restrict ourselves to forms with compact support $\Omega^k_c(M)$, in order to guarantee their integrability. To ease the notation, and because it will be more cumbersome in the next few sections, we will not mention anything about the integrability. We assume in the following that all the objects involved can be integrated whenever necessary.
	
	\subsubsection*{Closed and exact forms}
	If $\alpha\in\Omega^k(M)$, we say that its degree, denoted by $|\alpha|$, is $k$. A form $\alpha\in\Omega(M)$ is closed if $\peqsub{\d}{M}\alpha=0$ (it belongs to the kernel of $\peqsub{\d}{M}$), while it is exact if $\alpha=\peqsub{\d}{M}\beta$ for some $\beta\in\Omega(M)$ (it belongs to the image of $\peqsub{\d}{M}$). As $\peqsub{\d}{M}^2=0$, every exact form is closed. The converse depends on the de Rham cohomology
	\[H^k(M)=\frac{\mathrm{ker}\,(\peqsub{\d}{M})_k}{\mathrm{Im}\,(\peqsub{\d}{M})_{k-1}}\qquad \qquad(\peqsub{\d}{M})_k:\Omega^k(M)\to\Omega^{k+1}(M)\] 
	
	\subsubsection*{Local description}
	If $U\subset M$ is  a local patch with coordinates $\{x^i\}$, then we have (using  multi-index notation)
	\[\alpha=\sum_{|I|=|\alpha|}\alpha_I\peqsub{\d}{M} x^I\qquad\qquad\peqsub{\d}{M}\alpha=\sum_{j=1}^n\sum_{|I|=|\alpha|}\frac{\partial\alpha_I}{\partial x^k}\peqsub{\d}{M} x^k\wedge\peqsub{\d}{M} x^{I}\]

	\subsection[Differential geometry on \texorpdfstring{$(M,\partial M)$}{(M,∂ M)}]{Differential geometry on \texorpdfstring{$\boldsymbol{(M,\partial M)}$}{(M,∂ M)}}\label{section: geometry Mxpartial M}
	The best way to deal with forms in manifolds with boundaries is to use what we call the relative framework. Consider the pair $(M,N)$ where $N$ is a submanifold $\jmath:N\hookrightarrow M$ of codimension $1$. Define 
	\[\definicion
	\tikzmarkin{relativeformprima}(0.25,-0.2)(-0.25,0.45)
	\tikzmarkin{relativeform}(0.2,-0.2)(-0.2,0.45)
	\Omega^k(M,N):=\Omega^k(M)\oplus\Omega^{k-1}(N)
	\tikzmarkend{relativeform}
	\tikzmarkend{relativeformprima}
	\]
	and endow it with the following operations
	\[\begin{array}{lll}
	\multicolumn{3}{l}{\displaystyle\bullet\ (\alpha,\beta)\underline{\wedge}(\gamma,\delta):=\left(\alpha\wedge \gamma,\frac{(-1)^{|\alpha|}}{2}(\jmath^*\alpha)\wedge \delta+\frac{1}{2}\beta\wedge(\jmath^*\gamma)\right)}\\[2ex]
	\bullet\ \underline{\d}(\alpha,\beta):=(\peqsub{\d}{M}\alpha,\jmath^*\alpha-\peqsub{\d}{N}\beta)&\qquad\quad&\bullet\ \underline{\iota}_V(\alpha,\beta):=(\iota_V\alpha,-\iota_{\overline{V}}\beta)\\[1ex]
	\bullet\ \underline{\L}_V(\alpha,\beta):=(\L_V\alpha,\L_{\overline{V}}\beta)&&\bullet\ \underline{F}^*(\alpha,\beta):=(F^*\alpha,(F|_P)^*\beta)
	\end{array}\]
	where $V\in\mathfrak{X}(M)$ satisfies $\overline{V}:=V|_N\in\mathfrak{X}(N)$ and $F:(S,P)\to(M,N)$ is a map of pairs. The usual properties hold (see page \pageref{eq: d^2=0}). In particular, $\underline{\d}^2=0$, so we can define the cohomology of $(\Omega^k(M,N),\underline{\d})$. This is known as the \textbf{relative cohomology} and is denoted as $H^k(M,N)$ \cite[p.78]{bot1982differential}. Thus, by definition, $[(\alpha_1,\beta_1)]=[(\alpha_2,\beta_2)]$ if and only if
	\[\alpha_2=\alpha_1+\peqsub{\d}{M} Y\qquad\qquad \beta_2=\beta_1+\jmath^*Y-\peqsub{\d}{N}Z\]
	for some $(Y,Z)\in\Omega^{k-1}(M,N)$. We can also define the integral of top-forms as
	\[\definicion
	\tikzmarkin{relativeintegralprima}(0.25,-0.45)(-0.25,0.6)
	\tikzmarkin{relativeintegral}(0.2,-0.45)(-0.2,0.6)
	\int_{(M,N)}(\alpha,\beta):=\int_M\alpha-\int_N \beta
	\tikzmarkend{relativeintegral}
	\tikzmarkend{relativeintegralprima}
	\]
	From now on, we assume that $N\subset \partial M$ and, for simplicity, we will omit the subindex of the exterior derivatives as the manifolds where they act will always be clear from the context. This whole framework is still valid when $\partial M=\varnothing$.\vspace*{2ex}
	
	If we define the relative boundary $\underline{\partial}(M,N):=(\partial M\setminus N,\partial N)$, which satisfies $\underline{\partial}^2(M,N)=\varnothing$, and the relative inclusion $\underline{\jmath}:\underline{\partial}(M,N)\hookrightarrow(M,N)$, we have a relative version of the Stokes' theorem 
	\begin{align}\label{eq: stokes pair}
	\comentario\tikzmarkin{relativeStokesprima}(0.2,-0.5)(-0.2,0.65)
	\tikzmarkin{relativeStokes}(0.2,-0.45)(-0.2,0.6)
	\int_{(M,N)}\underline{\d}(\alpha,\beta)=\int_{\underline{\partial}(M,N)}\underline{\jmath}^*(\alpha,\beta)
	\tikzmarkend{relativeStokes}
	\tikzmarkend{relativeStokesprima}
	\end{align}
	where we have taken into account that the orientation of $\partial N\subset N$ is opposite to the orientation of $\partial N\subset \partial M$ (see remark \ref{remark: stokes minus} below). Of course, we can take $N=\partial M$ in which case $\underline{\partial}(M,N)=\varnothing$ and the last integral is zero. This result conforms to a general trend: results that hold for $M$ with boundary $\partial M$ also hold for $(M,N)$ with boundary $\underline{\partial}(M,N)$. In particular, results that hold for a manifold \emph{without} boundary, hold for $(M,\partial M)$ because $\underline{\partial}(M,\partial M)=\varnothing$.
	
	\subsection[Differential geometry on \texorpdfstring{$\F$}{F}]{Differential geometry on \texorpdfstring{$\boldsymbol{\F}$}{F}}\label{section: diff geom F}
	\subsubsection*{Tensor fields}
	Let $\F$ be the space of sections of some bundle $E\overset{\pi}{\to} M$\!. We assume that the infinite-dimensional manifold $\F$ has a differential structure and that we also have tensor fields over $\F$ (see appendix \ref{Appendix: jets} for a formal definition of these structures). In particular, we have vector fields $\XX\in\campos(\F)$ and $k$-form fields $\alpha\in\OOmega^k(\F)$ with the wedge product $\wwedge$. This graded algebra with the exterior derivative $\dd$ defines the complex
	\begin{equation}\label{eq: dd complex}
	0\longrightarrow\OOmega^0(\F)\overset{\dd}{\longrightarrow}\OOmega^1(\F)\overset{\dd}{\longrightarrow}\cdots\overset{\dd}{\longrightarrow}\OOmega^{k}(\F)\overset{\dd}{\longrightarrow}\OOmega^{k+1}(\F)\overset{\dd}{\longrightarrow}\cdots
	\end{equation}
	Other important operations that we will use extensively are the Lie derivative $\LL_\VV$ of any tensor field and the interior product $\ii_\VV$ of any form field, both of them with respect to a vector field $\VV\in\campos(\F)$.
	
	\subsubsection*{Closed and exact forms}
	If $\alpha\in\OOmega^k(M)$ we say that its degree, denoted by $\|\alpha\|$, is $k$. A form $\alpha\in\OOmega(\F)$ is closed if $\dd\alpha=0$ and exact if $\dd\alpha=\beta$ for some $\beta\in\OOmega(\F)$. As $\dd^2=0$, every exact form is closed. The converse depends on the cohomology but we will refrain from defining it here because it is not as straightforward as in the finite-dimensional case. Besides, it will not be necessary for our purposes.
	
	\subsubsection*{``Local'' description}
	In this case, there is no exact analog to the coordinates $\{x^i\}$ and their exterior derivatives $\d x^i$. However, we can \emph{formally} introduce a similar concept. First, recall that $\phi=(\phi^1,\cdots,\phi^m)\in\F$ is a section of a bundle $E\to M$\!. Each component $\phi^I$ is a section of a subbundle $E^I\subset E$ so $\F$ can be thought of as a Cartesian product of spaces of fields $\F^I$. Define the $I$-th evaluation function $\mathrm{Eval}^I_p:\F\to E^I$ at $p\in M$ as 
	\[\mathrm{Eval}^I_p(\phi)=\phi^I(p)\in E^I_p\]
	Although in general it is quite ill-behaved, in the following we proceed as if it were smooth in order to connect with the standard physics literature (see section \eqref{section: connection previous formalism} for the proper definition in the language of jets). Thus, we can compute its exterior derivative $\dd \mathrm{Eval}^I_p:T\F\to TE^I$, a $1$-form field of $\F$ with coefficients on $TE^I$. If we take a vector $\VV_{\!\phi}=(\VV^{\scalebox{.6}{$(1)$}}_{\!\phi},\ldots,\VV^{\scalebox{.6}{$(m)$}}_{\!\phi})\in T_\phi\F$, where each component $\VV_{\!\phi}^{\scalebox{.6}{$(I)$}}$ is given by a curve $\{\phi^I_\tau\}_\tau\subset\F^I$ with $\phi^I_0=\phi^I$ and $\partial_\tau|_0\phi^I_\tau=\VV^{\scalebox{.6}{$(I)$}}_{\!\phi}$, then 
	\[\dd_\phi\mathrm{Eval}^I_p(\VV_{\!\phi})=\left.\frac{\d}{\d\tau}\right|_{\tau=0}\mathrm{Eval}^I_p(\phi_\tau)=\left.\frac{\d}{\d\tau}\right|_{\tau=0}\phi^I_\tau(p)=\VV^{\scalebox{.6}{$(I)$}}_{\!\phi}(p)\]
	Thus, if we formally remove the dependence on $p$, we can define $\dd\phi^I$ as
	\[\definicion
	\tikzmarkin{ddphi}(0.25,-0.2)(-0.25,0.45)
	\tikzmarkin{ddphiprima}(0.2,-0.2)(-0.2,0.45)
	\dd\phi^I(\VV_{\!\phi}):=\dd_\phi\mathrm{Eval}^I\!(\VV_{\!\phi})=\VV^{\scalebox{.6}{$(I)$}}_{\!\phi}
	\tikzmarkend{ddphi}
	\tikzmarkend{ddphiprima}
	\]
	
	\subsection[Differential geometry on \texorpdfstring{$M\times\F$}{M x F}]{Differential geometry on \texorpdfstring{$\boldsymbol{M\times\F}$}{M x F}}\label{section: geometry MxF}
	\subsubsection*{Tensor fields}
	The differential structure of a product manifold is defined using the differential structures of the factors. One can then just consider tensor fields over $M\times\F$ in the usual fashion. However, here we have a feature that complicates everything: the second factor $\F$ is defined in terms of the first one $M$, they are entangled! In particular, the dependence on the points of $M$ may appear both explicitly and implicitly through elements of $\F$. The key here will be to only consider tensor fields with a local dependence on the points of $M$ in both senses. Doing that will allow us to somewhat disentangle the factors.\vspace*{2ex}
	
	For simplicity, we focus our attention on the differential forms because they have more structure and are more interesting, although similar arguments apply to general tensor fields. First we define the variational bicomplex
	\begin{equation}\label{eq: bicomplex}
	\begin{tikzcd}[row sep = scriptsize]
	&\rvdots & \rvdots& &\rvdots\\
	0\arrow[r]  & \OOmega^{(0,k)}(M\times\F)\arrow[r,"{\d}"]\arrow[u,"{\dd}"]  & \OOmega^{(1,k)}(M\times\F)\arrow[r,"{\d}"]\arrow[u,"{\dd}"]  &\cdots\arrow[r,"{\d}"] & \OOmega^{(n,k)}(M\times\F)\arrow[r]\arrow[u,"{\dd}"] &0\\
	& \rvdots\arrow[u,"{\dd}"] & \rvdots\arrow[u,"{\dd}"]  & & \rvdots\arrow[u,"{\dd}"] &\\
	0\arrow[r]  & \OOmega^{(0,1)}(M\times\F)\arrow[r,"{\d}"]\arrow[u,"{\dd}"]  & \OOmega^{(1,1)}(M\times\F)\arrow[r,"{\d}"]\arrow[u,"{\dd}"]  &\cdots\arrow[r,"{\d}"] & \OOmega^{(n,1)}(M\times\F)\arrow[r]\arrow[u,"{\dd}"] &0\\
	0\arrow[r]  & \OOmega^{(0,0)}(M\times\F)\arrow[r,"{\d}"]\arrow[u,"{\dd}"]  & \OOmega^{(1,0)}(M\times\F)\arrow[r,"{\d}"]\arrow[u,"{\dd}"]  &\cdots\arrow[r,"{\d}"] & \OOmega^{(n,0)}(M\times\F)\arrow[r]\arrow[u,"{\dd}"]  &0\\
	& 0\arrow[u]  & 0\arrow[u]  & & 0\arrow[u]\arrow[u]  &
	\end{tikzcd}
	\end{equation}
	$\OOmega^{(r,s)}(M\times\F)$ is the space of forms of degree $r$ in $M$ (horizontal part) and $s$ in $\F$ (vertical part). These spaces with $\wwedge$ define a bigraded algebra with two exterior derivatives: the horizontal $\d$, increasing $r$, and the vertical $\dd$, increasing $s$. The wedge product $\wwedge$ restricted to $(k,0)$-forms coincides with $\wedge$. We will often abuse notation and use the latter. If we replace $(M,\d)$ by the relative pair $((M,N),\underline{\d})$ of section \ref{section: geometry Mxpartial M} and use the relative analogues, we can define the \textbf{relative bicomplex framework}
	\[\definicion
	\tikzmarkin{rrelativeformprima}(0.15,-0.28)(-0.17,0.48)
	\tikzmarkin{rrelativeform}(0.1,-0.28)(-0.12,0.48)
	\displaystyle\OOmega^{(r,s)}\Big((M,N)\times\F\Big):=\OOmega^{(r,s)}(M\times\F)\oplus\OOmega^{(r-1,s)}(N\times\F)
	\tikzmarkend{rrelativeform}
	\tikzmarkend{rrelativeformprima}\]
	\[\begin{array}{lcl}
	\multicolumn{3}{l}{\displaystyle\bullet\ (\alpha,\beta)\underline{\wwedge}(\gamma,\delta):=\left(\alpha\wwedge \gamma,\frac{(-1)^{|\alpha|}}{2}(\jmath^*\alpha)\wwedge \delta+\frac{1}{2}\beta\wwedge(\jmath^*\gamma)\right)}\\[2ex]
	\displaystyle\bullet\ \underline{\dd}(\alpha,\beta):=(\dd\alpha,\dd\beta)&\quad\quad&\displaystyle\bullet\ \underline{\ii}_\VV(\alpha,\beta):=(\ii_\VV\alpha,\ii_\VV\beta)\\[1.2ex]
	\displaystyle\bullet\ \underline{\LL}_\VV(\alpha,\beta):=(\LL_\VV\alpha,\LL_\VV\beta)&&\displaystyle\bullet\ \underline{\FF}^*(\alpha,\beta):=(\FF^*\!\alpha,\FF^*\!\beta)
	\end{array}\]
	
	\subsubsection*{From \texorpdfstring{$\boldsymbol{M\times\F}$}{MxF} to \texorpdfstring{$\boldsymbol{M}$}{M}}
	An element $\alpha\in\OOmega^{(r,s)}(M\times\F)$ can be ``projected'' over $M$ if we fix some $\phi\in\F$ and $\VV^{1}_{\!\phi},\dots,\VV^{s}_{\!\phi}\in T_\phi\F$. Let us define $\alpha(\phi;\VV^{1}_{\!\phi},\dots,\VV^{s}_{\!\phi})\in\Omega^r(M)$, with the shorthand notation $\alpha(\phi;\VV^{i}_{\!\phi})$, as
	\begin{align}\label{eq: projection form}
	[\alpha(\phi;\VV^{i}_{\!\phi})]_p(V_p^{1},\dots,V_p^{r}):=\alpha_{(p,\phi)}(V_p^{1},\dots,V_p^{r},\VV^{1}_{\!\phi},\dots,\VV^{s}_{\!\phi})\qquad\qquad V_p^1,\dots,V_p^r\in T_pM
	\end{align}
	This simply separates the arguments in horizontal and vertical \cite{zuckerman1987action}. However, the base points are not separated in the sense that the dependence is on the whole field $\phi:M\to E$ instead of just $\phi(p)$. From now on, we restrict ourselves to $\OOmega_{\mathrm{loc}}^{(r,s)}(M\times\F)$. This is defined as the forms $\alpha$ such that $\alpha(\phi;\VV^i_\phi)$ only depends on $(p,\phi(p))$, on $\VV_{\!\phi}^i(p)$, and on finitely many of the derivatives of $\VV_{\!\phi}^i$ at $p$. 
	Finally, it is not hard to prove that if we project $\d\alpha\in\Omega^{(r+1,s)}(M\times\F)$, we obtain the exterior derivative of the projection:
	\[(\d\alpha)(\phi;\VV^{i}_{\!\phi})=\d\Big(\alpha(\phi;\VV^{i}_{\!\phi})\Big)\]
	This same construction can be generalized to any other tensor fields in $M$ and on $\F$ with the same concept of bigradation and the same idea of local tensor fields.
	
	\subsubsection*{Background and dynamical fields}
	Among the fields we are considering, some of them might be considered as background objects. It is convenient to separate them and denote the background fields as $\tilde{\phi}\in\widetilde{\F}$ and reserve the notation $\phi\in\F$ for the dynamical ones. Of course, the local dependence applies to both types of fields. In the following, although we should write $\F\times\widetilde{\F}$, we will ease the notation and write simply $\F$ understanding that there might be additional dependence on some fixed fields. \vspace*{2ex}
	
	With this notation at hand, a Lagrangian is an element $L\in\OOmega^{(n,0)}_{\mathrm{loc}}(M\times\F)$ such that $L(\phi|\tilde{\phi})\in\Omega^n(M)$ is local in the usual sense. For instance, consider the following Lagrangian
	\[L(\phi|g)=g^{-1}(\d\phi,\d\phi)\vol_g\in\Omega^n(M)\]
	where $\phi\in\F=\Omega^0(M)$ is the dynamical field and $g\in\widetilde{\F}=\mathrm{Met}(M)$ is the fixed field.

	\subsubsection*{Properties}
	Here we summarize the degrees of certain operators over $(M,N)\times\F$
	\[\begin{array}{|cc|cc|cc|}\hline
	\underline{\d}   &  \underline{\dd}  & \underline{\iota}_X & \underline{\ii}_\XX & \underline{\L}_X  & \underline{\LL}_\XX\\ \hline
	\Big((1,1),0\Big) & \Big((0,0),1\Big) &  \Big((-1,-1),0\Big)  &  \Big((0,0),-1\Big) & \Big((0,0),0\Big) & \Big((0,0),0\Big)\rule{0ex}{3.1ex}\\[1ex]\hline
	\end{array}\]
	The usual properties of differential geometry hold over $(M,N)\times\F$ 
	\begin{align}
	&\bullet\ \underline{\d}^2=0&&\underline{\dd}^2=0\label{eq: d^2=0}\\[1ex]
	&\bullet\ (\alpha,\beta)\underline{\wedge}(\gamma,\delta)=(-1)^{|\alpha||\gamma|}(\gamma,\delta)\underline{\wedge}(\alpha,\beta)&&(\alpha,\beta)\underline{\wwedge}(\gamma,\delta)=(-1)^{|\alpha||\gamma|+\|\alpha\|\|\gamma\|}(\gamma,\delta)\underline{\wwedge}(\alpha,\beta)\label{eq: supersym}\\[1ex]
	&\bullet\ \underline{\L}_V=\underline{\iota}_V\underline{\d}+\underline{\d}\,\underline{\iota}_V&&\underline{\LL}_{\VV}=\underline{\ii}_{\VV}\underline{\dd}+\underline{\dd}\,\underline{\ii}_{\VV}\label{eq: Cartan}\\[1ex]
	&\bullet\ \begin{array}{l}
	\hspace*{-1ex}\underline{\d}\Big((\alpha,\beta)\underline{\wedge}(\gamma,\delta)\Big)=\underline{\d}(\alpha,\beta)\underline{\wedge}(\gamma,\delta)+(-1)^{|\alpha|}(\alpha,\beta)\underline{\wedge}\,\underline{\d}(\gamma,\delta)\\[1ex]
	\hspace*{-1ex}\underline{\dd}\Big((\alpha,\beta)\underline{\wwedge}(\gamma,\delta)\Big)=\underline{\dd}(\alpha,\beta)\underline{\wwedge}(\gamma,\delta)+(-1)^{\|\alpha\|}(\alpha,\beta)\underline{\wwedge}\,\underline{\dd}(\gamma,\delta)\label{eq: Leibniz}
	\end{array}\span\span\\[.5ex]
	&\bullet\ \begin{array}{l}
	\hspace*{-1ex}\underline{\iota}_V\Big((\alpha,\beta)\underline{\wedge}(\gamma,\delta)\Big)=\underline{\iota}_V(\alpha,\beta)\underline{\wedge}(\gamma,\delta)+(-1)^{|\alpha|}(\alpha,\beta)\underline{\wedge}\,\underline{\iota}_V(\gamma,\delta)\\[1ex]
	\hspace*{-1ex}\underline{\ii}_\VV\Big((\alpha,\beta)\underline{\wwedge}(\gamma,\delta)\Big)=\underline{\ii}_\VV(\alpha,\beta)\underline{\wwedge}(\gamma,\delta)+(-1)^{\|\alpha\|}(\alpha,\beta)\underline{\wwedge}\,\underline{\ii}_\VV(\gamma,\delta)\label{eq: Leibniz imath}
	\end{array}\span\span\\[1ex]
	&\bullet\ \underline{\L}_V\underline{\d}=\underline{\d}\,\underline{\L}_V&&\underline{\LL}_{\VV}\underline{\dd}=\underline{\dd}\,\underline{\LL}_{\VV}\label{eq: LLdd=ddLL}\\[1ex]
	&\bullet\ \underline{F}^*\underline{\iota}_{F_{\!*}\!V}(\alpha,\beta)=\underline{\iota}_V\underline{F}^*(\alpha,\beta)&&\underline{\mathbb{F}}^*\underline{\ii}_{\mathbb{F}_{\!*}\!\VV}(\alpha,\beta)=\underline{\ii}_\VV\underline{\mathbb{F}}^*(\alpha,\beta)\label{eq: pullback interior product}
	\end{align} 
	Moreover, the operations in $\F$ and in $M$ commute. For instance
	\begin{equation}\label{eq: d dd=dd d}
	\underline{\d}\,\underline{\dd}=\underline{\dd}\,\underline{\d}\qquad \qquad\qquad\qquad \underline{\ii}_\VV\underline{\d}=\underline{\d}\,\underline{\ii}_\VV\qquad \qquad\qquad\qquad \underline{\iota}_V\underline{\dd}=\underline{\dd}\,\underline{\iota}_V
	\end{equation}
	In most of the mathematical literature those objects anti-commute (see section \ref{section: connection previous formalism}). This can be reconcile considering $(-1)^r\dd$ instead of $\dd$ and analogously for the other operators.

	\subsubsection*{Closed and exact forms}
	We know that a sufficient condition for a closed $k$-form of $M$ to be also exact is that $H^k(M)=0$. Of course, this is not a necessary condition. In fact, we have the following result:
	\begin{theorem}[Horizontal exactness theorem]\label{theorem: Wald exact - close}\mbox{}\\
		Let $r<n$ and $(\alpha,\beta)\in\OOmega^{(r,s)}_{\mathrm{\mathrm{loc}}}((M,\partial M)\times\F)$. If $(\alpha,\beta)$ is $\underline{\d}$-closed and one of the following two conditions holds
		\[(\alpha,\beta)(0;\VV_{\!\phi}^I)=0\qquad\text{or}\qquad s>0\]
		then there exists $(\gamma,\delta)\in\OOmega^{(r-1,s)}_{\mathrm{loc}}((M,\partial M)\times\F)$ such that $\underline{\d}(\gamma,\delta)=(\alpha,\beta)$.
	\end{theorem}
	\begin{proof}\mbox{}\\
		From the definition of $\underline{\d}$, we have on one hand $\d\alpha=0$ over $M\times\F$. From \cite{wald1990identically} or \cite[th.4.4]{takens1979global}, it follows (their proofs can be adapted to the case of a manifold with boundary) that $\alpha=\d \gamma$ for some $\gamma\in\OOmega^{(r-1,s)}_{\mathrm{loc}}(M\times\F)$. On the other hand, now we have $0=\jmath^*\alpha-\d\beta=\d(\jmath^*\gamma-\beta)$, so applying \cite{wald1990identically} again leads to the result.
	\end{proof}
	
	Notice that if $s=0$, the non-zero elements of $H^r(M)$, those which are independent of the fields, are non-exact. This is the reason why we have to include the vanishing condition for $\phi=0$.
	
	\subsection[\texorpdfstring{$M$}{M} as a space-time]{\texorpdfstring{$\boldsymbol{M}$}{M} as a space-time}\label{subsection: M space-time}
	The $n$-manifold $M$ is chosen so that it represents a physically reasonable space-time. As usual, we take it connected and oriented. Moreover, we require that $M$ admits a foliation by Cauchy hypersurfaces (although most of the results only require an arbitrary foliation). Without loss of generality, $M=I\times \Sigma$ for some interval $I=[t_i,t_f]$ and some $(n-1)$-manifold $\Sigma$ with boundary $\partial\Sigma$ (possibly empty). If we denote $\Sigma_i=\{t_i\}\times\Sigma$ and $\Sigma_f=\{t_f\}\times\Sigma$, we can split $\partial M$ in three distinguished parts
	\[\partial M=\Sigma_i\cup \peqsub{\partial}{L} M\cup \Sigma_f\]
	the lids $\Sigma_i,\Sigma_f$ and the ``lateral boundary'' $\lateral M:=I\times \partial \Sigma$. Notice that $M=I \times \Sigma$ is a manifold with corners $\partial\Sigma_i\cup\partial\Sigma_f$ which is equal, as a set, to $\partial(\lateral M)$. This is not a problem as most results of differential geometry, like Stokes' theorem, hold when the corners are simple enough (as they are here). The following diagram will be useful to keep track of the possible embeddings and the induced geometric objects, especially in section \ref{section: examples}, where we include some examples with detailed computations.
	\begin{equation}\label{eq: diagrama}
	\begin{tikzcd}[row sep = scriptsize]
	\Big(\Sigma,\gamma,D,\{a,b,\ldots\}\Big)\arrow[rrr,"{(\imath,n^\alpha)}",hook]  & & &
	\Big(M,g,\nabla,\{\alpha,\beta,\ldots\}\Big)\\\\
	\Big(\partial\Sigma,\overline{\gamma},\overline{D},\{\overline{a},\overline{b},\ldots\}\Big) 
	\arrow[uu, "{(\overline{\jmath},\mu^a)}",hook]
	\arrow[rrr,"{(\overline{\imath},\overline{m}^{\overline{\alpha}})}"', hook] & & &
	\Big(\partial M,\overline{g},\overline{\nabla},\{\overline{\alpha},\overline{\beta},\ldots\}\Big)
	\arrow[uu,"{(\jmath,\nu^\alpha)}"', hook]
	\end{tikzcd}
	\end{equation}
	The $4$-tuples consist of the manifold, the (pullback) metric, its associated Levi-Civita connection (only for non-degenerate metrics), and the abstract indices associated with the manifold. The labels of the arrows represent the embedding and the normal vector field (normalized in the non-degenerate case). For instance $(\imath,n^\alpha)$ represents the embedding $\imath$ and the vector field $n^\alpha$ defined over $\imath(\Sigma)$ and $g$-normal to it. The normal vector field is chosen future pointing for the Cauchy embeddings (vertical arrows in the diagram) and outward to the boundary (horizontal arrows). Notice that in the bottom lid $\Sigma_i$, there is a discrepancy as it can be seen as embedded through $\imath_i$ or as part of $\partial M$\!. Most often it appears as the boundary of $M$\!, so we choose the outwards pointing convention (past pointing). As a reminder to the reader, most of the objects living completely at the boundary are denoted with an overline (like $\overline{g}$, which in index notation is $\overline{g}_{\overline{\alpha}\overline{\beta}}$). Some of the objects included in the digram will not be used in this paper but we include them for completion. Finally, notice that the embeddings $\jmath$ and $\overline{\jmath}$ are fixed, since the boundaries are fixed. However $\imath$ and $\overline{\imath}$, which embed $(\Sigma,\partial\Sigma)$ in $(M,\partial M)$, can be chosen among all the embeddings satisfying $\imath(\partial\Sigma)\subset\lateral M$\!.\vspace*{2ex}
	
	As $(M,g)$ is oriented, we define the metric volume form $\vol_g$ that assigns $1$ to every positive orthonormal basis. We now orient $\Sigma$ and $\lateral M$ (in the non-degenerate case, the degenerate case needs a bit more work) with $\vol_{\gamma}$ and $\vol_{\overline{g}}$ respectively given by
	\begin{align}\label{eq: orientation sigma}
	\jmath^*(\iota_{\vec{U}}\vol_g)=\nu_\alpha U^\alpha\vol_{\overline{g}}\qquad\qquad\qquad\imath^*(\iota_{\vec{U}}\vol_g)=-n_\alpha U^\alpha\vol_\gamma
	\end{align}
	for every $\vec{U}$. These orientations are the ones that make Stokes' theorem hold. Finally, $\partial\Sigma$ can be oriented as the boundary of $\Sigma$. Thus $\vol_{\overline{\gamma}}$ is given (in the non-degenerate case, of course) by
	\begin{align}\label{eq: reverse orientation}
	\overline{\jmath}^*(\iota_{\vec{V}}\vol_\gamma)=\mu_a V^a\vol_{\overline{\gamma}}\qquad\quad\longrightarrow\quad\qquad\overline{\imath}^*(\iota_{\vec{W}}\vol_{\overline{g}})=+\overline{m}_{\overline{\alpha}} W^{\overline{\alpha}}\vol_{\overline{\gamma}}
	\end{align}
	If we define  $\nu_\perp:=- n_\alpha \nu^\alpha$ and $m_\perp:=- n_\alpha \jmath^\alpha_{\overline{\alpha}}m^{\overline{\alpha}}$, the last equations follows easily from
	\begin{align}\label{eq: nu y m}
	\nu^\alpha=\nu_\perp n^\alpha+m_\perp\,\imath^\alpha_a\mu^a\quad \qquad\qquad\qquad \jmath^\alpha_{\overline{\alpha}}m^{\overline{\alpha}}=m_\perp n^\alpha+\nu_\perp\,\imath^\alpha_a\mu^a
	\end{align}
	
	\begin{remark}\label{remark: stokes minus}\mbox{}\\
		If we use the Stokes' theorem from $\lateral M$ to $\partial(\lateral M)=\partial\Sigma_i\cup\partial\Sigma_f$, a minus sign appears.
	\end{remark}

	\section{Geometric structures in CPS}\label{section: space of fields}
	As in the previous section, we take $M$\!, which admits a foliation, a connected and oriented $n$-manifold. The values of the fields of $\F$ over $\lateral M$ may or may not be fixed. Although we will not consider it in the following, these techniques can be applied to ``nice'' constrained systems (submanifolds of the infinite jet bundle $J^\infty\!E$, see appendix \ref{Appendix: jets} and the example in page \pageref{example: constrain systems}).
	
	\subsection{Lagrangians}
	\begin{definition}\label{star}\mbox{}\\
		We define a \textbf{pair of Lagrangians} as an element of
		\[\Lag(M):=\OOmega_{\mathrm{loc}}^{(n,0)}((M,\partial M)\times\F)\]
	\end{definition}
	From sections \ref{section: geometry Mxpartial M} and \ref{section: geometry MxF} we know that
	\begin{equation}\label{eq: relation Lagrangians}[(L_1,\overline{\ell}_1)]=[(L_2,\overline{\ell}_2)]\ \quad \equiv\ \quad(L_2,\overline{\ell}_2)=(L_1,\overline{\ell}_1)+\underline{\d}(Y,\overline{y})\ \quad \equiv\ \quad \begin{array}{l}L_2=L_1+\d Y\\\overline{\ell}_2=\overline{\ell}_1+\jmath^*Y-\d\overline{y}\end{array}
	\end{equation}

	\subsection{Action}\label{section: action}
	\begin{definition}\mbox{}\\
		A \textbf{local action} is a map $\SS:\F\to\R$ of the form
		\begin{equation}\label{eq: action}\definicion
		\tikzmarkin{Sphiprima}(0.25,-0.45)(-0.25,0.6)
		\tikzmarkin{Sphi}(0.2,-0.45)(-0.2,0.6)
		\SS(\phi)=\int_{(M,\partial M)}(L,\overline{\ell})(\phi)
		\tikzmarkend{Sphi}
		\tikzmarkend{Sphiprima}
		\end{equation}
		for some local Lagrangians $(L,\overline{\ell})\in\Lag(M)$.
	\end{definition}
	The previous integral only makes sense if we project the Lagrangians over $M$ as explained in section \ref{section: geometry MxF}. Nonetheless, as it is always clear from the context, sometimes we will omit the field and write simply expressions like
	\[\SS=\int_{(M,\partial M)}(L,\overline{\ell})\]
	
	\begin{remark}\label{remark: pair of pairs}\mbox{}\\
		The previous definition can be tweaked to allow more general actions:
		\begin{equation}\label{eq: accion corner 1}
		\SS=\int_{(M,\lateral M)}(L,\overline{\ell})-\int_{\underline{\partial}(M,\lateral M)}(\Lambda,\overline{\lambda})
		\end{equation}
		with $(L,\overline{\ell})\in\Omega^n((M,\lateral M)\times\F)$ and $(\Lambda,\overline{\lambda})\in\Omega^{n-1}(\underline{\partial}(M,\lateral M)\times\F)$. This action is computed by integrating $L$ over the bulk $M$\!, $\overline{\ell}$ over the lateral boundary $\lateral M$, $\Lambda$ over the lids $\Sigma_i\cup\Sigma_f$, and $\overline{\lambda}$ over $\partial\Sigma_i\cup\partial\Sigma_f$ (``corner'' terms). These actions appear for instance in \cite{booth2005horizon,hartle1981boundary,sorkin1975time}.\vspace*{2ex}
		
		In section \ref{section: geometry Mxpartial M} we mentioned that $(M,N)$ behaves like a manifold with boundary $\underline{\partial}(M,N)$ and that, in particular, $(M,\partial M)$ behaves as a manifold with no boundary. We can then repeat the argument and consider the pair of pairs $((M,\lateral M),\underline{\partial}(M,\lateral M))$, the relative inclusion $\underline{\jmath}:\underline{\partial}(M,\lateral M)\hookrightarrow(M,\lateral M)$ and the analogous relative boundary, exterior derivative and integral
		\begin{align*}
		&\bullet\ \underline{\underline{\partial}}\Big((M,\lateral M),\underline{\partial}(M,\lateral M)\Big):=\Big(\underline{\partial}(M,\lateral M)\setminus\underline{\partial}(M,\lateral M),\underline{\partial}^2(M,\lateral M)\Big)=\varnothing\\[1.2ex]
		&\bullet\ \underline{\underline{\d}}\Big((\alpha,\overline{\beta}),(\gamma,\overline{\delta})\Big):=\Big(\underline{\d}(\alpha,\overline{\beta}),\underline{\jmath}^*(\alpha,\overline{\beta})-\underline{\d}(\gamma,\overline{\delta})\Big)\\[0.8ex]
		&\bullet\ \int_{((M,\lateral M),\underline{\partial}(M,\lateral M))}\Big((\alpha,\overline{\beta}),(\gamma,\overline{\delta})\Big):=\int_{(M,\lateral M)}(\alpha,\overline{\beta})-\int_{\underline{\partial}(M,\lateral M)}(\gamma,\overline{\delta})
		\end{align*}
		As this new ``manifold'' has no boundary, it behaves like a true manifold without boundary. In fact, theorem \ref{theorem: Wald exact - close} holds and the action \eqref{eq: accion corner 1} can be rewritten as
		\begin{equation}\label{eq: accion corner 2}
		\SS=\int_{((M,\lateral M),\underline{\partial}(M,\lateral M))}\big((L,\overline{\ell}),(\Lambda,\overline{\lambda})\big)
		\end{equation}
		which is formally equivalent to \eqref{eq: action} which, in turn, is equivalent to an integral over a manifold without boundary. We have decided to work with \eqref{eq: action} to show explicitly the equivalence with an action over a space with no boundary. The interested reader will have little problem adapting our computations to the action \eqref{eq: accion corner 2}. Moreover, both actions are equivalent if we consider ``Dirichlet'' conditions over $\underline{\partial}(M,\lateral M)$ i.e.~fixing the initial and final values of $\phi\in\F$.
	\end{remark}
	
	\begin{definition}\label{def: S-equiv}\mbox{}\\
		Two pairs of Lagrangians $(L_i,\overline{\ell}_i)\in\Lag(M)$ are $\boldsymbol{\int}$\!\textbf{-equivalent}, and denoted $(L_1,\overline{\ell}_1)\equivint(L_2,\overline{\ell}_2)$, if for every $\phi\in\F$, we have
		\[\int_{(M,\partial M)} (L_1,\overline{\ell}_1)(\phi)=\int_{(M,\partial M)} (L_2,\overline{\ell}_2)(\phi)\]
	\end{definition}

	This is an equivalence relation and each class is associated with an action $\SS$ that we denote $\SS=[\![(L,\overline{\ell})]\!]$. We have two equivalence relations, \eqref{eq: relation Lagrangians} and \ref{def: S-equiv}, which in general are different. Nonetheless, they are equal for \textbf{contractible bundles} i.e.~fibered bundles with contractible fibers ($M$ is not necessarily contractible). Within this category we find vector bundles, affine bundles (e.g.~Yang-Mills theories), some principle bundles (e.g.~Dirac monopole or BPST), and some quotient bundles (e.g.~Riemannian metrics or, up to a topological obstruction given by the Euler class, the bundle of Lorentzian metrics).
	
	\begin{lemma}\label{lemma: equivalence relations}\mbox{}
		\begin{itemize}
			\item If $[(L_1,\overline{\ell}_1)]=[(L_2,\overline{\ell}_2)]$, then $(L_1,\overline{\ell}_1)\equivint(L_2,\overline{\ell}_2)$.
			\item If $E\to M$ is a contractible bundle and $(L_1,\overline{\ell}_1)\equivint(L_2,\overline{\ell}_2)$, then $[(L_1,\overline{\ell}_1)]=[(L_2,\overline{\ell}_2)]$.
		\end{itemize}
	\end{lemma}
	\begin{proof}\mbox{}\\
		The first point is clear from the relative Stokes' theorem \eqref{eq: stokes pair}. The second one is proven in \ref{theorem: trivial action trivial lagrangians}.
	\end{proof}

	The condition that $E\to M$ is a contractible bundle cannot be removed. To see this, we can consider trivial Lagrangians, also known as null or closed Lagrangians, in spaces without boundary. They are non-exact Lagrangians whose Euler-Lagrange equations are zero. They are characterized by $H^n(E)$ (see theorem \ref{theorem: null lagrangians}). Indeed, in \cite[page 207]{anderson1989variational} there are examples of null Lagrangians which are non-exact Lagrangians $L$ over non-contractible bundles $E\to M$ such that, for every $\phi\in\F$, $L(\phi)=\d Y_\phi$ for some $Y_\phi\in\Omega^{n-1}(M)$. Applying Stokes' theorem we obtain that $\SS=[\![L]\!]=0$ but $L\neq \d Y$ because the potential $Y_\phi$ is non-local in $\phi$ i.e.~$L\equivint 0$ but $[L]\neq0$. \vspace*{2ex}
	
	The same happens always if $M$ has boundary because $H_c^n(M)=0$ \cite{weintraub2014differential}. Thus $L(\phi)$ is always exact but in general, $L$ is not. As in the previous paragraph, we could evaluate and use Stokes' theorem to obtain a boundary integral. Of course, we still get Euler-Lagrange equations because the boundary integrand is non-local (depends on the whole field $\phi$) and the usual computations of variations are not valid.\vspace*{2ex}
	
	Summarizing, for contractible bundles the action $\S:=[\![(L,\overline{\ell})]\!]$ is the initial object from which everything derives and nothing depends on the choice of representative Lagrangians. Otherwise, we have to be careful because one might be using ill-behaved representatives like the ones mentioned in the previous paragraph. Of course, one can try to always work with nice representatives, but it is not straightforward because the Lagrangians $L$ are not the problem, their potentials $Y_\phi$ are. For non-contractible bundles, it is better to just consider $\sigma:=[(L,\overline{\ell})]$ as the initial object instead of $\SS$ and proceed analogously.
	
	\begin{remark}\label{remark: action inicial}\mbox{}\\
		From now on, the action $\S:=[\![(L,\overline{\ell})]\!]$ over a \textbf{contractible bundle} $E\to M$ is our starting point. As mentioned before, everything holds if one considers a general bundle $E$ and replaces $\SS$ by $\sigma:=[(L,\overline{\ell})]$.
	\end{remark}

	\subsection{Space of solutions}
	We define the \textbf{space of solutions} of the action $\SS$ as the set of its critical points:
	\[\definicion
	\tikzmarkin{definitionSOLprima}(0.25,-0.3)(-0.25,0.55)
	\tikzmarkin{definitionSOL}(0.2,-0.3)(-0.2,0.55)
	\Sol(\SS)=\Big\{\phi\in\F\ /\ \dd_\phi \SS\uptolids0\Big\}
	\tikzmarkend{definitionSOL}
	\tikzmarkend{definitionSOLprima}\]
	where $\uptolids$ means that the equality is up to integral terms over the lids $\underline{\partial}(M,\lateral M)$. Those terms vanish once we fix initial and final values on the fields. This space will depend on the functional form of $\SS$ as well as on the definition of $\F$. We denote the inclusion $\peqsub{\jj}{\SS}:\Sol(\SS)\hookrightarrow\F$.
	
	\subsection{Variations}\label{section: variations}
	\subsubsection*{Variation of the Lagrangians}
	
	Applying \eqref{eq: first variation} to the Lagrangian $L\in\OOmega^{(n,0)}(M\times \F)$, we see that $\dd L$ is ``decomposable'' i.e.~there exist unique $E_I\in\OOmega^{(n,0)}(M\times \F)$ and some $\Theta^L\in\OOmega^{(n-1,1)}(M\times \F)$, linear in $\dd\phi^I$ and their derivatives, such that 
	\begin{equation}\label{eq: dL}
	\dd L=E_I\wedge\dd\phi^I+\d\Theta^L
	\end{equation}
	In some mathematical references the first term appears as $\boldsymbol{E}:=E_I\wedge\dd\phi^I\in\OOmega^{(n,1)}(M\times\F)$, which is a source form \cite{zuckerman1987action}, and in some physical references is denoted
	\[E_I\wedge\dd \phi^I=\frac{\delta L}{\delta\phi^I}\delta\phi^I\d x^n\]
	We will stick with the proposed notation $E_I\wedge\dd \phi^I$. Now we want to obtain a decomposition similar to \eqref{eq: dL} but over the boundary. The same theorem guarantees that $\dd \overline{\ell}$ is decomposable but, notice that upon integration of the previous expression, an additional term $\jmath^*\Theta^L$ appears on the boundary. We impose the condition that $\jmath^*\Theta^L$ is also decomposable over the ``lateral boundary'' $\lateral M$ (this is usually achieved by including some boundary conditions in the definition of $\F$, as we will see in the examples of section \ref{section: examples}). With this additional hypothesis and equation \eqref{eq: first variation}, we get that
	\begin{equation}\label{eq: dl+Theta}
	\dd\overline{\ell}=\overline{b}_I\wedge\dd\phi^I+\jmath^*\Theta^L-\d\overline{\theta}^{(L,\overline{\ell})}
	\end{equation}
	over $\lateral M$ for some unique $\overline{b}_I\in\OOmega^{(n-1,0)}(\lateral M\times \F)$ and some $\overline{\theta}^{(L,\overline{\ell})}\in \OOmega^{(n-2,1)}(\lateral M\times \F)$. Notice that \eqref{eq: dl+Theta} does not hold in general over $\Sigma_i\cup\Sigma_f\subset\partial M$\!. However, restricting to $(M,\lateral M)\subset(M,\partial M)$ we have
	\begin{equation}\label{eq: d(L,l)}
	\comentario
	\tikzmarkin{dLprima}(0.2,-0.32)(-0.2,0.57)
	\tikzmarkin{dL}(0.2,-0.27)(-0.2,0.52)
	\underline{\dd}(L,\overline{\ell})=(E_I,\overline{b}_I)\wedge\dd\phi^I+\underline{\d}\Big(\Theta^L,\overline{\theta}^{(L,\overline{\ell})}\Big)
	\tikzmarkend{dL}
	\tikzmarkend{dLprima}
	\end{equation}
	\vspace*{-1ex}
	
	where the wedge here is acting component-wise. Notice that, in fact, this equality only makes sense over $\OOmega^{(n,1)}((M,\lateral M)\times\F)$, as $\overline{b}_I$ and $\overline{\theta}$ are only defined over the lateral boundary $\lateral M$\!.

	\subsubsection*{Variation of the Action}
	
	Now, from equations \eqref{eq: action}, \eqref{eq: dd complex}, and \eqref{eq: d(L,l)}, we can compute the variation of $\SS=[\![(L,\overline{\ell})]\!]$ to obtain
	\begin{align}\label{eq: dS}
	\comentario\tikzmarkin{variationSprima}(0.2,-0.5)(-0.2,0.65)
	\tikzmarkin{variationS}(0.2,-0.45)(-0.2,0.6)
	\dd\SS&=\int_{(M,\lateral M)}(E_I,\overline{b}_I)\wedge\dd\phi^I+\int_{\underline{\partial}(M,\lateral M)}\left(\underline{\jmath}^*\Big(\Theta^L,\overline{\theta}^{(L,\overline{\ell})}\Big)-(\dd\overline{\ell},0)\right)
	\tikzmarkend{variationS}
	\tikzmarkend{variationSprima}
	\end{align}
	where we have used equation \eqref{eq: stokes pair} and
	\[\int_{(M,\partial M)}(\alpha,\overline{\beta})=\int_{(M,\lateral M)}(\alpha,\overline{\beta})-\int_{\underline{\partial}(M,\lateral M)}(\overline{\beta},0)\]
	In order to compute the critical points, we recall that $\dd_\phi\SS\uptolids0$ is equivalent to requiring $\dd_\phi \SS(\mathbb{V}_{\!\phi})=0$ for every $\mathbb{V}_{\!\phi}\in T_\phi\F$ vanishing in a neighborhood of the lids. Notice that the last integral of \eqref{eq: dS} vanishes because it is linear in $\VV_{\!\phi}$ and a finite number of its derivatives, leading to
	\begin{align*}
	0=\dd_\phi\SS(\mathbb{V}_{\!\phi})&=\int_{(M,\lateral M)}\Big(E_I(\phi)\VV^I_{\!\phi},\overline{b}_I(\phi)\overline{\VV}^I_{\!\phi}\Big)=\int_ME_I(\phi)\VV^I_{\!\phi}-\int_{\lateral M}\overline{b}_I(\phi)\overline{\VV}^I_{\!\phi}
	\end{align*}
	The first integral vanishes if $E_I(\phi)=0$ as, in general, we assume that $\mathbb{V}_{\!\phi}$ is arbitrary away from the lids (for constrained systems one has to be more careful but usually one can easily deal with those particular cases). The second integral is a bit trickier precisely because it is more common to have some constraints i.e.~boundary conditions. If $\overline{\VV}^{\raisebox{-.5ex}{\,\scalebox{0.6}{$I$}}}_{\!\phi}$ is arbitrary, then $\overline{b}_I(\phi)=0$. This is the case, for instance, of Neumann (natural) boundary conditions. We can also have $\overline{\VV}^{\raisebox{-.5ex}{\,\scalebox{0.6}{$I$}}}_{\!\phi}=0$ if $\F$ is defined such that $\phi^I$ is fixed at the boundary i.e.~imposing Dirichlet boundary conditions (see for instance the first two examples of section \ref{section: examples} for a careful discussion). In that case, $\overline{b}_I$ is arbitrary. Nonetheless, for convenience, we usually define $\overline{b}_I:=0$ whenever $\overline{\VV}^{\raisebox{-.5ex}{\,\scalebox{0.6}{$I$}}}_{\!\phi}=0$. Of course, we can have a point-wise mixture of both cases or more complicated cases (that have to be studied individually). With this convention, the space of solutions is simply
	\[\comentario\tikzmarkin{caracterizacionSOL}(0.2,-0.35)(-0.2,0.6)
	\tikzmarkin{caracterizacionSOLprima}(0.2,-0.3)(-0.2,0.55)
	\Sol(\SS)=\Big\{\phi\in\F\ /\ (E_I,\overline{b}_I)(\phi)=0\Big\}\overset{\peqsub{\jj}{\SS}}{\hookrightarrow}\F
	\tikzmarkend{caracterizacionSOL}
	\tikzmarkend{caracterizacionSOLprima}\]
	This space does not depend on the chosen Lagrangians $(L,\overline{\ell})$ because it is defined as the set of critical points of $\SS$ (up to lid-terms). This, together with \eqref{eq: d(L,l)}, implies in particular the following result.
	
	\begin{lemma}\mbox{}\label{lemma: equivalent Lagrangians same eq}\\
		If $(L,\overline{\ell})\equivint(0,0)$, then $\underline{\dd}( L,\overline{\ell})=\underline{\d}\Big(\Theta^{L},\overline{\theta}^{(L,\overline{\ell})}\Big)$ for some $\Big(\Theta^{L},\overline{\theta}^{(L,\overline{\ell})}\Big)\in\OOmega^{(n-1,1)}((M,\lateral M)\times\F)$.
	\end{lemma}
	
	A last comment is in order now: although physically the important objects are the equations of motion with their boundary conditions, they are not here! The fundamental object is the action as it is the tool used to build the rest of the geometric structures. We can find different actions with the same space of solutions that lead to different symplectic geometries over the CPS. They all define the same physics (at least classically) but some formulations are better suited for our purposes than others. A very simple example is if we take $L_2=\lambda L_1$ for some $\lambda\in\R\setminus\{0,1\}$ and $\partial M=\varnothing$. The actions and the equations are \emph{different}, $\SS_2=\lambda\SS_1$ and $E^{(2)}=\lambda E^{(1)}$, although they define the same space of solutions. Another example is given by trivial Lagrangians as we mentioned at the end of section \ref{section: action}. In \ref{example: scalar field no equation} we will see a more elaborate example.\vspace*{2ex}
	
	Let us restate lemma \ref{lemma: equivalence relations} in a more useful way (as we explained in remark \ref{remark: action inicial}, we are only considering contractible bundles):
	
	\begin{lemma}\label{lemma: (L2,l2)}\mbox{}\\
		$(L_1,\overline{\ell}_1)\equivint(L_2,\overline{\ell}_2)$ if and only if there exists $(Y,\overline{y})\in\OOmega^{(n-1,0)}((M,\partial M)\times \F)$ with
		\begin{equation}\label{eq: L2=L1}\degeneracion
		\tikzmarkin{eqL1L2}(.2,-0.35)(-0.3,0.57)
		(L_2,\overline{\ell}_2)=(L_1,\overline{\ell}_1)+\underline{\d}(Y,\overline{y})\qquad\equiv\qquad \begin{array}{l}L_2=L_1+\d Y\\\overline{\ell}_2=\overline{\ell}_1+\jmath^*Y-\d \overline{y}\end{array}
		\tikzmarkend{eqL1L2}
		\end{equation}
	\end{lemma}

	Notice in particular that, given $(L,\overline{\ell})\in\Lag(M)$, if there exists $Y\in\OOmega^{(n-1,0)}(M\times \F)$ such that $\overline{\ell}+\jmath^*Y$ is exact, then $(L,\overline{\ell})\equivint(L+\d Y,0)$ and we can remove the boundary term of the action $\SS$.
	
	\begin{lemma}\label{lemma: (Theta2,theta2)}\mbox{}\\
		If  $(L_1,\overline{\ell}_1)\equivint(L_2,\overline{\ell}_2)$, then 
		\[\degeneracion
		\tikzmarkin{tupla}(0.1,-0.16)(-0.1,0.66)
		\Big(\Theta^{L_2},\overline{\theta}^{(L_2,\overline{\ell}_2)}\Big)=\Big(\Theta^{L_1},\overline{\theta}^{(L_1,\overline{\ell}_1)}\Big)+\underline{\d}(Z,\overline{z})+\underline{\dd}(Y,\overline{y})\qquad\equiv\qquad\begin{array}{l}\Theta^{L_2}=\Theta^{L_1}+\d y+\dd Y\\\overline{\theta}^{(L_2,\overline{\ell}_2)}=\overline{\theta}^{(L_1,\overline{\ell}_1)}+\jmath^*y-\d \overline{z}+\dd \overline{y}
		\tikzmarkend{tupla}\end{array}\]
		for some $(Z,\overline{z})\in\OOmega^{(n-2,1)}((M,\lateral M)\times\F)$, $(Y,\overline{y})\in\OOmega^{(n-1,0)}((M,\lateral M)\times \F)$. Moreover, if $\jmath^*\Theta^{L_1}$ is decomposable, then so is $\jmath^*\Theta^{L_2}$.
	\end{lemma}
	\begin{proof}\mbox{}\\
		From lemma \ref{lemma: (L2,l2)} we know that 	$(L_2,\overline{\ell}_2)=(L_1,\overline{\ell}_1)+\underline{\d}(Y,\overline{y})$ for some $(Y,\overline{y})\in\OOmega^{(n-1,0)}((M,\partial M)\times \F)$. Now, from lemma \ref{lemma: equivalent Lagrangians same eq} we have
		\[\underline{\d}\left(\Theta^{L_2}-\Theta^{L_1},\overline{\theta}^{(L_2,\overline{\ell}_2)}-\overline{\theta}^{(L_1,\overline{\ell}_1)}\right)=\underline{\dd}(L_2-L_1,\overline{\ell}_2-\overline{\ell}_1)=\underline{\dd}\,\underline{\d}(Y,\overline{y})\overset{\eqref{eq: d dd=dd d}}{=}\underline{\d}\,\underline{\dd}(Y,\overline{y})\]
		Applying the horizontal exactness theorem \ref{theorem: Wald exact - close}, we see that there exists $(Z,\overline{z})\in\OOmega^{(n-2,1)}((M,\lateral)\times \F)$ such that
		\[\Big(\Theta^{L_2}-\Theta^{L_1},\overline{\theta}^{(L_2,\overline{\ell}_2)}-\overline{\theta}^{(L_1,\overline{\ell}_1)}\Big)-\underline{\dd}(Y,\overline{y})=\underline{\d}(Z,\overline{z})\]	
		Finally, assume that $\jmath^*\Theta^{L_1}$ is decomposable. As $\dd\jmath^*Y$ is also decomposable according to equation \eqref{eq: first variation}, we see that $\jmath^*\Theta^{L_2}$ is decomposable.
	\end{proof}
	
	Notice in particular that even if we consider $(L_2,\overline{\ell}_2)=(L_1,\overline{\ell}_1)$ in the previous lemma, we obtain that $(\Theta^{L_1},\overline{\theta}^{(L_1,\overline{\ell}_1)})$ and $(\Theta^{L_2},\overline{\theta}^{(L_2,\overline{\ell}_2)})$ are only equal up to a $\underline{\d}$-exact term. However, this is good enough as this implies that they are equal on cohomology, which is what we need for the following.

	\subsection{Symplectic structure}\label{section: symplectic structure}
	
	\begin{definition}\label{def: OOmega}\mbox{}\\
		Given a local action $\SS=[\![(L,\overline{\ell})]\!]$ and a Cauchy embedding $\underline{\imath}=(\imath,\overline{\imath}):(\Sigma,\partial\Sigma)\hookrightarrow (M,\lateral M)$, we define its associated \textbf{(pre)symplectic form}
		\[\definicion
		\tikzmarkin{definitionOOmegaprima}(0.15,-0.48)(-0.25,0.63)
		\tikzmarkin{definitionOOmega}(0.1,-0.48)(-0.2,0.63)
		\OOmega_\SS^\imath:=\int_{(\Sigma,\partial\Sigma)}\underline{\dd}\,\underline{\imath}^*\Big(\Theta^L,\overline{\theta}^{(L,\overline{\ell})}\Big)\in\OOmega^2(\F)
		\tikzmarkend{definitionOOmega}
		\tikzmarkend{definitionOOmegaprima}\]
	\end{definition}
	As $\underline{\dd}^2=0$, $\OOmega_\SS^\imath$ is clearly closed but it might be degenerate. It will be useful in the following to define the \textbf{symplectic currents} 
	\begin{equation}\label{eq: omega and overline omega}\definicion
	\tikzmarkin{overlineomegaprima}(.25,-0.3)(-0.35,0.52)
	\tikzmarkin{overlineomega}(.2,-0.3)(-0.3,0.52)
	\Big(\Omega^\Theta,\overline{\omega}^{(\Theta,\overline{\theta})}\Big):=\underline{\dd}\Big(\Theta^L,\overline{\theta}^{(L,\overline{\ell})}\Big)\in\OOmega^{(n-1,2)}((M,\lateral M)\times\F)
	\tikzmarkend{overlineomega}
	\tikzmarkend{overlineomegaprima}
	\end{equation}
	We write $\Omega^\Theta$ instead of $\Omega^{(L,\Theta^L)}$ as it is clear that $\omega$ depends on $L$ through $\Theta$ (analogously for $\overline{\omega}$).
	
	\begin{proposition}\mbox{}\\
		$\OOmega_\SS^\imath$ does not depend on the chosen Lagrangians representatives.
	\end{proposition}
	\begin{proof}\mbox{}\\
		Consider $(L_1,\overline{\ell}_1)\equivint(L_2,\overline{\ell}_2)$. Applying lemma \ref{lemma: (Theta2,theta2)} and equations \eqref{eq: stokes pair} and \eqref{eq: d^2=0}, we have\allowdisplaybreaks[0]
		\begin{align*}
		(\OOmega_\SS^\imath)_2&=\int_{(\Sigma,\partial\Sigma)}\underline{\dd}\,\underline{\imath}^*\Big(\Theta^{L_2},\overline{\theta}^{(L_2,\overline{\ell}_2)}\Big)=\int_{(\Sigma,\partial\Sigma)}\underline{\imath}^*\Big\{\underline{\dd}\Big(\Theta^{L_1},\overline{\theta}^{(L_1,\overline{\ell}_1)}\Big)+\underline{\d}\,\underline{\dd}(Z,\overline{z})+\underline{\dd}^2(Y,\overline{y})\Big\}=(\OOmega_\SS^\imath)_1
		\end{align*}
		\mbox{}\vspace*{-6ex}
		
	\end{proof}

	\begin{proposition}\label{prop: OOmega_SS}\mbox{}\\
		The (pre)symplectic structure $\OOmega_\SS:=\peqsub{\jj}{\SS}^*\OOmega_\SS^\imath$ on $\Sol(\SS)$ is independent of the embedding $\imath:\Sigma\hookrightarrow M$\!.
	\end{proposition}
	\begin{proof}\mbox{}\\
		Consider two Cauchy surfaces $\Sigma_1:=\imath_1(\Sigma)$ and $\Sigma_2:=\imath_2(\Sigma)$, $\Sigma_1$ in the future of $\Sigma_2$, that do not intersect (if they do, we consider a third one not intersecting any of them and repeat the argument twice). Denote by $N$ the manifold bounded by both Cauchy surfaces. Its boundary is $\partial N=\Sigma_1\cup\lateral N\cup\Sigma_2$. We denote $\underline{i}:\underline{\partial}(N,\lateral N)\hookrightarrow (N,\lateral N)$. The exterior derivate of the symplectic currents \eqref{eq: omega and overline omega} is
		\begin{align*}
		&\underline{\d}\Big(\Omega^\Theta,\overline{\omega}^{(\Theta,\overline{\theta})}\Big)\overset{\eqref{eq: d dd=dd d}}{=}\underline{\dd}\,\underline{\d}\Big(\Theta^L,\overline{\theta}^{(L,\overline{\ell})}\Big)\overset{\eqref{eq: d(L,l)}}{=}\underline{\dd}\Big(\underline{\dd}\big(L,\overline{\ell})-(E_I,\overline{b}_I)\wedge\dd\phi^I\Big)\updown{\eqref{eq: d^2=0}}{\eqref{eq: Leibniz}}{=}-\underline{\dd}(E_I,\overline{b}_I)\wwedge \dd\phi^I
		\end{align*}
		Then, we have
		\begin{align*}
		-\!\int_{(N,\lateral N)}\underline{\dd}(E_I,\overline{b}_I)\wwedge \dd\phi^I=\int_{(N,\lateral N)}\underline{\d}\big(\Omega^\Theta,\overline{\omega}^{(\Theta,\overline{\theta})})\overset{\mathrm{\ref{eq: stokes pair}}}{=}\int_{\underline{\partial}(N,\lateral N)}\underline{i}^*(\Omega^\Theta,\overline{\omega}^{(\Theta,\overline{\theta})})=\OOmega_\SS^{\imath_1}-\OOmega_\SS^{\imath_2}
		\end{align*}
		In the last equality we have used that each connected component of $\underline{\partial}(N,\lateral N)=(\Sigma_1,\partial\Sigma_1)\cup(\Sigma_2,\partial\Sigma_2)$ has opposite orientation. Clearly $\underline{\jj}_\SS^*\underline{\dd}(E_I,\overline{b}_I)=\underline{\dd}\,\underline{\jj}_\SS^*(E_I,\overline{b}_I)=0$, so indeed $\peqsub{\jj}{\SS}^*\OOmega_\SS^{\imath_1}=\peqsub{\jj}{\SS}^*\OOmega_\SS^{\imath_2}$.
	\end{proof}

	\subsection{Symmetries}
	\subsubsection*{Definition and main properties}
	\begin{definition}\label{def: symmetry}\mbox{}\\
		We say that a vector field $\XX\in\campos(\F)$ is a \textbf{symmetry} of the action $\SS$ if  $\LL_\XX\SS\uptolids0$. We denote $\Sym(\SS)$ the set of vector fields which are symmetries of $\SS$.
	\end{definition}
	
	A symmetry vector field does not change the action. Let us prove that, as one might expect, it leaves invariant $\Sol(\SS)$. Notice that the converse is not true in general: there are symmetries of $\Sol(\SS)$ that do not give rise to symmetries of $\SS$ (e.g.~the scaling transformation \cite[p.255]{olver2000applications}).
	
	\begin{proposition}\label{proposition: XX tangente a Sol}\mbox{}\\
		If $\SS$ is a local action and $\XX\in\Sym(\SS)$, then $\XX_\phi\in T_\phi\Sol(\SS)$ for every $\phi\in\Sol(\SS)$.  In particular
		\begin{equation}\label{eq: overline XX}\definicion
		\tikzmarkin{overline XXprima}(.25,-0.2)(-0.35,0.45)
		\tikzmarkin{overline XX}(.2,-0.2)(-0.3,0.45)
		\overline{\XX}:=\XX|_{\Sol(\SS)}\in\campos(\Sol(\SS))
		\tikzmarkend{overline XX}
		\tikzmarkend{overline XXprima}
		\end{equation}
	\end{proposition}	
	\begin{proof}\mbox{}\\
		We have that $\Sol(\SS)=\{(E_I,\overline{b}_I)=0\}$. So we have to prove that those conditions are preserved i.e.~$\underline{\LL}_{\XX_\phi}(E_I,\overline{b}_I)=0$ for every $\phi\in\Sol(\SS)$. Taking the Lie derivative of \eqref{eq: dS} with respect to some vector field $\YY\in\campos(\F)$, we obtain
		\begin{align*}
		\LL_\YY&\dd\SS=\int_{(M,\lateral M)} \Big(\underline{\LL}_\YY( E_I, \overline{b}_I)\wedge\dd\phi^I+ (E_I,\overline{b}_I)\wedge\LL_\YY\dd\phi^I\Big)+\int_{\underline{\partial}(M,\lateral M)}\underline{\LL}_\YY\!\left(\underline{\jmath}^*(\Theta^L,\overline{\theta}^{(L,\overline{\ell})})-(\dd\overline{\ell},0)\right)
		\end{align*}		
		The previous expression is an element of $\OOmega^1(\F)$. Taking as base point any $\phi\in\Sol(\S)$, so $(E_I,\overline{b}_I)(\phi)=0$, and considering $\YY=\XX\in\Sym(\SS)$, so $\LL_\XX\SS\uptolids0$, leads to
		\begin{align*}
		0\uptolids(\dd\LL_\XX\SS)_\phi\overset{\eqref{eq: LLdd=ddLL}}{=}(\LL_\XX\dd\SS)_\phi&\uptolids\int_{(M,\lateral M)}\underline{\LL}_\XX(E_I,\overline{b}_I)(\phi)\wedge\dd\phi^I
		\end{align*}
		As it is usual in variational calculus, we evaluate this $1$-form of $\F$ at every vector $\ZZ$ that, as a map over $M$, vanishes at a neighborhood of the lids. As this set is dense, we obtain that $\underline{\LL}_\XX(E_I,\overline{b}_I)=0$ (and, in particular, the last integral vanishes).
	\end{proof}
	
	Our goal now is to determine if $\XX\in\Sym(\SS)$ restricted to the space of solutions $\overline{\XX}:=\XX|_{\Sol(\SS)}\in\campos(\Sol(\SS))$ is a Hamiltonian vector field i.e.~if there exists some function $\HH^\SS_{\overline{\XX}}:\Sol(\SS)\to\R$ such that 
	\begin{equation}
	\ii_{\overline{\XX}}\OOmega_\SS=\dd\HH^\SS_{\overline{\XX}}
	\end{equation}
	It is unlikely that this holds in general, but we can obtain an interesting subset of symmetries which are also Hamiltonian.
	
	\subsubsection*{\texorpdfstring{$\boldsymbol{\d}$}{d}-symmetries}
	If $\XX\in\Sym(\SS)$, then the inverse of the relative Stokes' theorem implies that $\underline{\LL}_\XX(L,\overline{\ell})(\phi)$ is $\underline{\d}$-exact over $(M,\partial M)$ for every $\phi\in\F$. However, in section \ref{section: action} we saw that this does not imply that $\underline{\LL}_\XX(L,\overline{\ell})$ is $\underline{\d}$-exact over $(M,\partial M)\times\F$. Moreover, here we do not have the analog of lemma \ref{lemma: equivalent Lagrangians same eq}, so in general we cannot expect to obtain a \emph{local} $\underline{\d}$-potential for $\underline{\LL}_\XX(L,\overline{\ell})$. This motivates the following definition.
	
	\begin{definition}\mbox{}\label{def: d-symmetry}
		\begin{itemize}
			\item $\XX\in\campos(\F)$ is a $\boldsymbol{\underline{d}}$\textbf{-symmetry} (or infinitesimal variational symmetry) of $(L,\overline{\ell})\in\Lag(M)$ if
			\begin{equation}\label{eq: (L_X L,L_x l)}
			\definicion
			\tikzmarkin{definitionAlphaBetaprima}(0.25,-0.3)(-0.25,0.55)
			\tikzmarkin{definitionAlphaBeta}(0.2,-0.3)(-0.2,0.55)
			\underline{\LL}_\XX(L,\overline{\ell})=\underline{\d}\Big(S^L_\XX,\overline{s}^{(L,\overline{\ell})}_\XX\Big)
			\tikzmarkend{definitionAlphaBeta}
			\tikzmarkend{definitionAlphaBetaprima}\end{equation}
			for some $\Big(S^L_\XX,\overline{s}^{(L,\overline{\ell})}_\XX\Big)\in\OOmega^{(n-1,0)}((M,\partial M)\times\F)$.
			\item $\XX\in\campos(\F)$ is a $\boldsymbol{\underline{d}}$\textbf{-symmetry} of $\SS=[\![(L,\overline{\ell})]\!]$ if it is a $\underline{\d}$-symmetry of $(L,\overline{\ell})$. We denote $\Sym_{\d}(\SS)$ the set of vector fields which are $\underline{d}$-exact symmetries of $\SS$.
		\end{itemize}
	\end{definition}
	The last definition does not depend on the representative if we consider an additional condition:
	\begin{remark}\mbox{}\\
		From now on, we assume $H^{n-1}(M,\partial M)=0$. Later in \ref{remark: Hn-1 not 0} we will comment on the consequences of having $H^{n-1}(M,\partial M)\neq0$.
	\end{remark}
	
	\begin{lemma}\label{lemma: (S,s)}\mbox{}\\
		If $\XX\in\campos(\F)$ is a $\underline{\d}$-symmetry of $(L_1,\overline{\ell}_1)$, it is also a $\underline{\d}$-symmetry of $(L_2,\overline{\ell}_2)\equivint(L_1,\overline{\ell}_1)$ with
		\begin{equation}\label{eq: alpha beta def}\degeneracion
		\tikzmarkin{degeneracionAlphaBeta}(0.2,-0.3)(-0.2,0.56)
		\Big(S^{L_2}_{\XX},\overline{s}^{(L_2,\overline{\ell}_2)}_{\XX}\Big)=\Big(S^{L_1}_{\XX},\overline{s}^{(L_1,\overline{\ell}_1)}_{\XX}\Big)+\underline{\LL}_\XX(Y,\overline{y})+\underline{\d}(A,\overline{a})
		\tikzmarkend{degeneracionAlphaBeta}
		\end{equation}
		for some $(A,\overline{a})\in\OOmega^{(n-2,0)}((M,\partial M)\times\F)$ and $(Y,\overline{y})\in\OOmega^{(n-1,0)}((M,\partial M)\times \F)$.
	\end{lemma}
	\begin{proof}\mbox{}\\
		From equations \eqref{eq: L2=L1} and \eqref{eq: (L_X L,L_x l)} we obtain
		\begin{align*}
		\underline{\LL}_\XX(L_2,\overline{\ell}_2)=\underline{\LL}_\XX(L_1,\overline{\ell}_1)+\underline{\LL}_\XX\, \underline{\d}( Y,\overline{y})\overset{\eqref{eq: d dd=dd d}}{=}\underline{\d}\left(\!\Big(S^{L_1}_{\XX},\overline{s}^{(L_1,\overline{\ell}_1)}_{\XX}\Big)+\underline{\LL}_\XX(Y,\overline{y})\right)
		\end{align*}
		The result follows from $H^{n-1}(M,\partial M)=0$. Notice that we cannot apply the horizontal exactness theorem  \ref{theorem: Wald exact - close} because both $(S,\overline{s})(0)$ and $(Y,\overline{y})(0)$ can be non-zero.
	\end{proof}

	\begin{proposition}\label{proposition: characterization symmetry}\mbox{}\\
		If $\SS$ is a local action, then $\Sym_{\d}(\SS)\subset\Sym(\SS)$.
	\end{proposition}
	\begin{proof}\mbox{}\\
		Applying the relative Stokes' theorem \eqref{eq: stokes pair} to \eqref{eq: (L_X L,L_x l)} leads to \ref{def: symmetry}.
	\end{proof}

	\subsubsection*{Currents, Charges, and Potentials of a \texorpdfstring{$\boldsymbol{\d}$}{d}-Symmetry}
	Given $\XX$, a $\underline{\d}$-symmetry of $(L,\overline{\ell})\in\Lag(M)$, we define its $\boldsymbol{\XX}$\textbf{-current} (or Noether current)
	\begin{equation}\label{eq: noether currents}
	\definicion
	\tikzmarkin{definitionJjprima}(0.25,-0.3)(-0.25,0.55)
	\tikzmarkin{definitionJj}(0.2,-0.3)(-0.2,0.55)
	\Big(J^\Theta_\XX,\overline{\jmath}^{(\Theta,\overline{\theta})}_\XX\Big):=\Big(S^L_\XX,\overline{s}_\XX^{(L,\overline{\ell})}\Big)-\underline{\ii}_\XX\Big(\Theta^L,\overline{\theta}^{(L,\overline{\ell})}\Big)\in\OOmega^{(n-1,0)}((M,\lateral M)\times\F)
	\tikzmarkend{definitionJj}
	\tikzmarkend{definitionJjprima}
	\end{equation}
	
	\begin{lemma}\label{lemma: (J,j)}\mbox{}\\
		Let $\XX\in\Sym_{\d}(\SS)$ and $(L_1,\overline{\ell}_1)\equivint(L_2,\overline{\ell}_2)$ two representatives of $\SS$, then
		\begin{equation}\label{eq: J2=J1+d}
		\degeneracion
		\tikzmarkin{degeneracionJj}(0.2,-0.3)(-0.2,0.55)
		\Big(J^{\Theta_2}_\XX,\overline{\jmath}^{(\Theta_2,\overline{\theta}_2)}_\XX\Big)=\Big(J^{\Theta_1}_\XX,\overline{\jmath}^{(\Theta_1,\overline{\theta}_1)}_\XX\Big)+\underline{\d}\Big(\!(A,\overline{a})-\underline{\ii}_\XX(Z,\overline{z})\Big)
		\tikzmarkend{degeneracionJj}
		\end{equation}
		for some $(A,\overline{a})\in\OOmega^{(n-2,0)}((M,\lateral M)\times \F)$ and $(Z,\overline{z})\in\OOmega^{(n-2,1)}((M,\lateral M)\times \F)$, so they are equal on cohomology.
	\end{lemma}
	\begin{proof}\mbox{}\\
		From lemmas \ref{lemma: (Theta2,theta2)} and \ref{lemma: (S,s)}, we have
		\begin{align*}
		\Big(J^{\Theta_2}_\XX,\overline{\jmath}^{(\Theta_2,\overline{\theta}_2)}_\XX\Big)&\!=\!\Big(S^{L_1}_\XX,\overline{s}^{(L_1,\overline{\ell}_1)}_\XX\Big)+\underline{\LL}_\XX(Y,\overline{y})+\underline{\d}(A,\overline{a})-\underline{\ii}_\XX\Big(\!\Big(\Theta^{L_1},\overline{\theta}^{(L_1,\overline{\ell}_1)}\Big)+\underline{\d}(Z,\overline{z})+\underline{\dd}(Y,\overline{y})\Big)\updown{\eqref{eq: noether currents}}{\eqref{eq: Cartan}}{=}\\
		&\!=\!\Big(J^{\Theta_1}_\XX,\overline{\jmath}^{(\Theta_1,\overline{\theta}_1)}_\XX\Big)+\underline{\d}\big(A-\ii_\XX Z,\overline{a}-\ii_\XX\overline{z}\big)
		\end{align*}
		
		\mbox{}\vspace*{-7.5ex}
		
	\end{proof}	
	
	\begin{definition}\label{def: Noether charge}\mbox{}\\
		Let $\SS=[\![(L,\overline{\ell})]\!]$ be a local action and $\underline{\imath}=(\imath,\overline{\imath}):(\Sigma,\partial\Sigma)\hookrightarrow (M,\lateral M)$ some Cauchy embedding. For every $\XX\in\Sym_{\d}(\SS)$, we define the $\boldsymbol{\XX}$\textbf{-charge} (or Noether charge) 
		\[
		\definicion
		\tikzmarkin{definitionHprima}(0.25,-0.45)(-0.25,0.6)
		\tikzmarkin{definitionH}(0.2,-0.45)(-0.2,0.6)
		\HH_{\XX}^{\SS,\imath}:=\int_{(\Sigma,\partial\Sigma)}\underline{\imath}^*\Big(J^\Theta_\XX,\overline{\jmath}^{(\Theta,\overline{\theta})}_\XX\Big)\in\OOmega^0(\F)
		\tikzmarkend{definitionH}
		\tikzmarkend{definitionHprima}\]
	\end{definition}

	\begin{lemma}\mbox{}\\
		$\HH_\XX^{\SS,\imath}$ is $\R$-linear in $\XX$ and it only depends on $\SS$, $\XX$, and $\imath:\Sigma\hookrightarrow M$\!.
	\end{lemma}
	\begin{proof}\mbox{}\\
		From definition \eqref{eq: (L_X L,L_x l)}, if we use that $H^{n-1}(M,\partial M)=0$ and the fact that $\underline{\LL}_\XX$ is $\R$-linear in $\XX$, we obtain \[\Big(S^L_{\XX}+\lambda S^L_{\YY}-S^L_{\XX+\lambda\YY},\overline{s}_\XX^{(L,\overline{\ell})}+\lambda\overline{s}_\YY^{(L,\overline{\ell})}-\overline{s}_{\XX+\lambda\YY}^{(L,\overline{\ell})}\Big)=\underline{\d}\big(B,\overline{b}\big)\]
		for some fixed $\lambda\in\R$. From that and the linearity of $\underline{\ii}_\XX$ in $\XX$, we can easily prove the linearity of $\HH_\XX^{\SS,\imath}$ from its definition and the relative Stokes' theorem \eqref{eq: stokes pair}.\vspace*{2ex}
		
		The independence of $(L,\overline{\ell})$ is a direct consequence of lemma \ref{lemma: (J,j)} and the relative Stokes' theorem.
	\end{proof}

	\begin{proposition}\label{proposition: HHXX}\mbox{}\\
		$\HH^\SS_{\overline{\XX}}:=\peqsub{\jj}{\SS}^*\HH_\XX^{\SS,\imath}$  only depends on $\SS$ and on $\overline{\XX}:=\XX|_{\Sol(\SS)}$.
	\end{proposition}
	\begin{proof}\mbox{}\\	
		On one hand we have
		\begin{align}\label{eq: (dJ,dj)}
		\begin{split}
		&\underline{\d}\Big(J^\Theta_\XX,\overline{\jmath}^{(\Theta,\overline{\theta})}_\XX\Big)\updown{\eqref{eq: (L_X L,L_x l)}}{\eqref{eq: d(L,l)}}{=}\underline{\LL}_\XX(L,\overline{\ell})-\underline{\ii}_\XX\Big(\underline{\dd}(L,\overline{\ell})-(E_I,\overline{b}_I)\wedge\dd\phi^I\Big)\overset{\eqref{eq: Cartan}}{=}(E_I,\overline{b}_I)\LL_\XX\phi^I
		\end{split}
		\end{align}
		We can then mimic the argument of the proof of \ref{prop: OOmega_SS} to show that $\peqsub{\jj}{\SS}^*\HH_\XX^{\SS,\imath}$ does not depend on the embedding. Consider now $\XX_1,\XX_2\in\Sym_{\d}(\SS)$ with $\overline{\XX}_1=\overline{\XX}_2$. Then
		\[\underline{\d}\,\underline{\jj}_\SS^*\Big(S^L_{\XX_1}-S^L_{\XX_2},\overline{s}_{\XX_1}^{(L,\overline{\ell})}-\overline{s}_{\XX_2}^{(L,\overline{\ell})}\Big)\overset{\eqref{eq: (L_X L,L_x l)}}{=}\underline{\jj}_\SS^*\Big(\underline{\LL}_{\XX_1}(L,\overline{\ell})-\underline{\LL}_{\XX_2}(L,\overline{\ell})\Big)\updown{\eqref{eq: pullback interior product}}{\eqref{eq: Cartan}}{=}\underline{\LL}_{(\overline{\XX}_1-\overline{\XX}_2)}\underline{\jj}_\SS^*(L,\overline{\ell})=0\]
		where we have used  $(\peqsub{\jj}{\SS})_*\overline{\XX}_k=\XX_k$. Thus, from $H^{n-1}(M,\partial M)=0$,
		\[\underline{\jj}_\SS^*\Big(S_{\XX_2}^L,\overline{s}_{\XX_2}^{(L,\overline{\ell})}\Big)=\underline{\jj}_\SS^*\Big(S_{\XX_1}^L,\overline{s}_{\XX_1}^{(L,\overline{\ell})}\Big)+\underline{\d}(C,\overline{c})\]
		Finally
		\begin{align*}
		\HH^\SS_{\overline{\XX}_2}&=\peqsub{\jj}{\SS}^*\HH_{\XX_2}^{\SS,\imath}\overset{\eqref{eq: pullback interior product}}{=}\int_{(\Sigma,\partial\Sigma)} \underline{\imath}^*\left(\underline{\jj}_\SS^*\Big(S_{\XX_2}^L,\overline{s}_{\XX_2}^{(L,\overline{\ell})}\Big)-\underline{\ii}_{\overline{\XX}_2}\Big(\Theta^L,\overline{\theta}^{(L,\overline{\ell})}\Big)\right)=\\
		&=\int_{(\Sigma,\partial\Sigma)} \underline{\imath}^*\left(\underline{\jj}_\SS^*\Big(S_{\XX_1}^L,\overline{s}_{\XX_1}^{(L,\overline{\ell})}\Big)+\underline{\d}(C,\overline{c})-\underline{\ii}_{\overline{\XX}_1}\Big(\Theta^L,\overline{\theta}^{(L,\overline{\ell})}\Big)\right)\updown{\eqref{eq: stokes pair}}{\eqref{eq: pullback interior product}}{=}\peqsub{\jj}{\SS}^*\HH_{\XX_2}^\imath=\HH_{\overline{\XX}_2}
		\end{align*}
		\mbox{}\vspace*{-6ex}
		
	\end{proof}
	
	\begin{definition}\mbox{}\\
		Given $\XX\in\Sym_{\d}(\SS)$, we say that $(Q_\XX^\Theta,\overline{q}_\XX^{(\Theta,\overline{\theta})})\in\OOmega^{(n-2,0)}((M,\lateral M)\times\Sol(\SS))$ is the $\boldsymbol{\XX}$-\textbf{potential} (or Noether potential) if
		\[
		\definicion
		\tikzmarkin{Xcargaprima}(0.25,-0.3)(-0.25,0.5)
		\tikzmarkin{Xcarga}(0.2,-0.3)(-0.2,0.5)
		\underline{\d}\Big(Q_\XX^\Theta,\overline{q}_\XX^{(\Theta,\overline{\theta})}\Big)=\peqsub{\jj}{\SS}^*\Big(J_\XX^\Theta,\overline{\jmath}_\XX^{(\Theta,\overline{\theta})}\Big)
		\tikzmarkend{Xcarga}
		\tikzmarkend{Xcargaprima}\]
	\end{definition}
	From \eqref{eq: J2=J1+d}, it follows
	that
	\begin{equation}\label{eq: Q2=Q1+d}
	\degeneracion
	\tikzmarkin{degeneracionQq}(0.2,-0.3)(-0.2,0.52)
	\Big(Q_\XX^{\Theta_2},\overline{q}_\XX^{(\Theta_2,\overline{\theta}_2)}\Big)=\Big(Q_\XX^{\Theta_1},\overline{q}_\XX^{(\Theta_1,\overline{\theta}_1)}\Big)+(A,\overline{a})-\underline{\ii}_\XX(Z,\overline{z})+(B,\overline{b})
	\tikzmarkend{degeneracionQq}
	\end{equation}
	for some closed $(B,\overline{b})\in\Omega^{n-2}(M,\partial M)$, which might not be exact if $H^{n-2}(M,\partial M)\neq0$. From the relative Stokes' theorem \eqref{eq: stokes pair}, the $\XX$-charge is zero and so, as we will see in section \ref{section: gauge} below, $\overline{\XX}$ is a gauge vector field. Sometimes we only have a bulk potential $Q^\Theta_\XX$. In that case,  the $\XX$-charge can be written as a boundary integral.
	
	\begin{remark}\label{remark: Hn-1 not 0}\mbox{}\\
		If $H^{n-1}(M,\partial M)\neq0$, then we have to replace $\underline{\d}(A,\overline{a})$ in \eqref{eq: alpha beta def} and \eqref{eq: J2=J1+d} by a closed element $(T,\overline{t})$ that is not exact in general. In particular we have
		\[\int_{(\Sigma,\partial\Sigma)}\underline{\imath}^*\Big(J^{\Theta_2}_\XX,\overline{\jmath}^{(\Theta_2,\overline{\theta}_2)}_\XX\Big)=\int_{(\Sigma,\partial\Sigma)}\underline{\imath}^*\Big(J^{\Theta_1}_\XX,\overline{\jmath}^{(\Theta_1,\overline{\theta}_1)}_\XX\Big)+\int_{(\Sigma,\partial\Sigma)}\underline{\imath}^*(T,\overline{t})\]
		The last integral is purely topological as it is independent of the fields. This proves the following: if $H^{n-1}(M,\partial M)\neq0$, the $\XX$-charges are only defined up to topological terms. Moreover, there is one topological charge for every non-zero element $[(T,\overline{t})]$ in $H^{n-1}(M,\partial M)$. In particular, the vector field $\XX=0$ has non-zero charges so $\HH_\XX^{\SS,\imath}$ is not linear in $\XX$. This is unpleasant but not a big problem, as we are actually interested in the $\dd$-exterior derivative $\dd\HH_\XX^{\SS,\imath}$ which is linear in any case.
	\end{remark}
	
	\subsubsection*{Hamilton Equation}
	We have all the tools to prove, as we wanted, that the $\underline{\d}$-exact symmetries restricted to the space of solutions are Hamiltonian vector fields.
	\begin{theorem}\mbox{}\\
		If $\XX\in\Sym_{\d}(\SS)$ and we denote $\overline{\XX}:=\XX|_{\Sol(\SS)}\in\campos(\Sol(\SS))$ its restriction to the space of solutions, then
		\[
		\comentario
		\tikzmarkin{HamiltonianEqprima}(0.2,-0.28)(-0.2,0.45)
		\tikzmarkin{HamiltonianEq}(0.2,-0.23)(-0.2,0.4)
		\ii_{\overline{\XX}}\OOmega_\SS=\dd\HH^\SS_{\overline{\XX}}
		\tikzmarkend{HamiltonianEq}
		\tikzmarkend{HamiltonianEqprima}
		\]
	\end{theorem}
	\begin{proof}\mbox{}\\
		Let  $\underline{\imath}:(\Sigma,\partial\Sigma)\hookrightarrow(M,\lateral M)$ be a Cauchy embedding. Then, we have
		\begin{align*}
		\ii_{\overline{\XX}}\OOmega_\SS&-\dd\HH^\SS_{\overline{\XX}}=\ii_{\overline{\XX}}\peqsub{\jj}{\SS}^*\OOmega^\imath_\SS-\dd\peqsub{\jj}{\SS}^*\HH^{\SS,\imath}_\XX\overset{\eqref{eq: pullback interior product}}{=}\peqsub{\jj}{\SS}^*\left(\ii_{\XX}\OOmega^\imath_\SS-\dd\HH^{\SS,\imath}_\XX\right)\updown{\mathrm{\ref{def: OOmega}}}{\mathrm{\ref{def: Noether charge}}}{=}\\
		&=\peqsub{\jj}{\SS}^*\left(\int_{(\Sigma,\partial\Sigma)}\underline{\ii}_{\XX}\underline{\dd}\,\underline{\imath}^*\Big(\Theta^L,\overline{\theta}^{(L,\overline{\ell})}\Big)-\int_{(\Sigma,\partial\Sigma)}\underline{\dd}\,\underline{\imath}^*\Big(J^\Theta_\XX,\overline{\jmath}^{(\Theta,\overline{\theta})}_\XX\Big)\right)\updown{\eqref{eq: Cartan}}{\eqref{eq: noether currents}}{=}\\
		&=\peqsub{\jj}{\SS}^*\left(\LL_{\XX}\int_{(\Sigma,\partial\Sigma)}\underline{\imath}^*\Big(\Theta^L,\overline{\theta}^{(L,\overline{\ell})}\Big)-\dd\int_{(\Sigma,\partial\Sigma)} \underline{\imath}^*\Big(S^L_\XX,\overline{s}_\XX^{(L,\overline{\ell})}\Big)\right)
		\end{align*}
		To prove that the last line is zero for every $\XX\in\Sym_{\d}(\SS)$ and every embedding $\imath$, we follow the idea of the proof of proposition \ref{proposition: XX tangente a Sol}. However, instead of $M$\!, we consider the manifold $N$ bounded by $\Sigma_i$ and $\imath(\Sigma)$. Notice that the integrals over the bottom lid $\Sigma_i$ are zero following the same idea as in \ref{proposition: XX tangente a Sol} of evaluating this expression over all vectors vanishing at a neighborhood of the lid. Applying the relative Stokes theorem \eqref{eq: stokes pair}, we can rewrite the last line of the previous computation as
		\begin{align*}
		&\peqsub{\jj}{\SS}^*\left(\LL_{\XX}\int_{(N,\lateral N)}\underline{\d}\Big(\Theta^L,\overline{\theta}^{(L,\overline{\ell})}\Big)-\dd\int_{(N,\lateral N)}\underline{\d}\Big(S^L_\XX,\overline{s}_\XX^{(L,\overline{\ell})}\Big)\right)\updown{\eqref{eq: d(L,l)}}{\eqref{eq: (L_X L,L_x l)}}{=}\\
		&=\peqsub{\jj}{\SS}^*\left(\LL_{\XX}\int_{(N,\lateral N)}\Big(\underline{\dd}(L,\overline{\ell})-(E_I,\overline{b}_I)\wedge\dd\phi^I\Big)-\dd\int_{(N,\lateral N)}\underline{\LL}_\XX (L,\overline{\ell})\right)\overset{\eqref{eq: LLdd=ddLL}}{=}\\
		&=-\peqsub{\jj}{\SS}^*\int_{(N,\lateral N)}\Big(\underline{\LL}_{\XX}(E_I,\overline{b}_I)\wedge\dd\phi^I+(E_I,\overline{b}_I)\wedge\dd\LL_{\XX}\phi^I\Big)
		\end{align*}
		This expression is zero over the space of solutions from proposition \ref{proposition: XX tangente a Sol}.
	\end{proof}

	\subsection{Gauge vector fields}\label{section: gauge}
	A symmetry $\XX\in\Sym_{\d}(\SS)$ has a flow over $\F$ which, by definition, does not change the value of $\SS$. It moves points (fields) around and, if we restrict the flow to the critical surface $\Sol(\SS)$, we know from proposition \ref{proposition: XX tangente a Sol} that it moves critical points (solutions) to other critical points (solutions). This latter movement is governed by the $\XX$-charge $\HH^\SS_{\overline{\XX}}$, which is the Hamiltonian of $\overline{\XX}$ over $(\Sol(\SS),\OOmega_\SS)$.
	
	\begin{definition}\label{def: gauge field}\mbox{}\\
		We say that a non-zero $\overline{\XX}\in\campos(\Sol(\SS))$ is a \textbf{gauge vector field} if 
		\[\definicion
		\tikzmarkin{definicionGaugeprima}(0.25,-0.2)(-0.25,0.45)
		\tikzmarkin{definicionGauge}(0.2,-0.2)(-0.2,0.45)
		\ii_{\overline{\XX}}\OOmega_\SS=0
		\tikzmarkend{definicionGauge}
		\tikzmarkend{definicionGaugeprima}\]
		We denote $\mathrm{Gauge}(\SS)\subset\campos(\Sol(\SS))$ the set of all gauge vector fields.
	\end{definition}
	A gauge vector field moves along the degenerate directions of the presymplectic form $\OOmega_\SS$. Of course, if $\OOmega_\SS$ is truly symplectic, there is no degenerate direction and, therefore, no gauge vector field.\vspace*{2ex}
	
	$\Sym(\SS)\subset\campos(\F)$ while $\mathrm{Gauge}(\SS)\subset\campos(\Sol(\SS))$. A gauge vector field is not a symmetry although it might be extandable to one. Conversely, $\XX\in\Sym_{\d}(\SS)$ induces the gauge vector field $\overline{\XX}:=\XX|_{\Sol(\SS)}$ if $\dd\HH^\SS_{\overline{\XX}}=0$.

	\subsection[From \texorpdfstring{$M$}{M}-vector fields to \texorpdfstring{$\F$}{F}-vector fields]{From \texorpdfstring{$\boldsymbol{M}$}{M}-vector fields to \texorpdfstring{$\boldsymbol{\F}$}{F}-vector fields}
	The purpose of this section is to define a canonical vector field $\XX_\xi\in\campos(\F)$ associated with $\xi\in\mathfrak{X}(M)$. We assume that $\xi$ is tangent to $\lateral M$ but not necessarily to $\Sigma_i$ and $\Sigma_f$ (we can always extend the interval $[t_i,t_f]$). We cannot stress enough the importance of the different base manifolds, $M$ and $\F$, of both vector fields. The key feature that we will exploit is that a field $\phi$ can be interpreted in two ways:
	\begin{itemize}
		\item As a tensor field on $M$\!. That is, a section $\phi:M\to E$ of some bundle $E\overset{\pi}{\to}M$ such that $\phi_p:=\phi(p)\in E_p:=\pi^{-1}(E)$. In particular, we can take its Lie derivative $(\L_\xi\phi)_p=\partial_\tau|_{0}(\varphi^\xi_\tau)^*\phi_p$, where $\{\varphi^\xi_\tau\}_\tau\subset\mathrm{Diff}(M)$ is the flow of $\xi$. If $\F$ is reasonable enough (in particular, the Lie derivative has to be well defined as it would happen if $E$ is what is called natural bundle \cite{Michor}), we will have $\L_\xi\phi\in\F$.
		\item As a point of $\F$. In particular, a vector field $\XX_\xi:\F\to T\F$ over $\F$ is a section of $T\F$. Thus, $(\XX_\xi)_\phi:=\XX_\xi(\phi)\in T_\phi\F\cong\F$. The last isomorphism comes from the fact that $\F$ is linear. The non-linear case is not as straightforward but in concrete examples, one can usually perform analogous constructions (see example \ref{example: parametrized YM} where $\F=\Omega^1(M)\times\mathrm{Diff}(M)$).
	\end{itemize}
	With those remarks, we can define
	\begin{equation}\label{eq: definicion XX_xi}
	\definicion
	\tikzmarkin{definicionXgiprima}(0.25,-0.2)(-0.25,0.4)
	\tikzmarkin{definicionXgi}(0.2,-0.2)(-0.2,0.4)
	(\XX_\xi^I)_{\phi}=\L_\xi\phi^I
	\tikzmarkend{definicionXgi}
	\tikzmarkend{definicionXgiprima}
	\end{equation}
	\begin{lemma}\mbox{}\\
		Given $\xi\in\mathfrak{X}(M)$, we have that 
		\begin{equation}\label{eq: LL_XX_xi=L_xi}
		\comentario
		\tikzmarkin{Xgiprima}(0.2,-0.25)(-0.2,0.45)
		\tikzmarkin{Xgi}(0.2,-0.2)(-0.2,0.4)
		\LL_{\XX_\xi}\phi^I=\L_\xi\phi^I
		\tikzmarkend{Xgi}
		\tikzmarkend{Xgiprima}
		\end{equation}
	\end{lemma}
	\begin{proof}\mbox{}\\
		We recall from section \ref{section: diff geom F} that $\mathrm{Eval}^I\!(\phi)=\phi^I\in\F^I$. We consider now a path $\{\phi^I_\tau\}_\tau\subset\F^I$ such that $\phi^I_0=\phi^I$ and $\partial_\tau|_0\phi^I_\tau=(\XX_\xi^I)_{\phi}\in E_{\phi}$.
		\begin{align*}
		\LL_{\XX_\xi}\phi^I\overset{\eqref{eq: Cartan}}{=}\ii_{\XX_\xi}\dd \phi^I=\ii_{(\XX_\xi)_\phi}\dd_{\phi}\mathrm{Eval}^I=\left.\frac{\d}{\d \tau}\right|_{\tau=0}\mathrm{Eval}^I\!(\phi_\tau)=\left.\frac{\d}{\d \tau}\right|_{\tau=0}\phi^I_\tau=(\XX_\xi^I)_{\phi}\overset{\eqref{eq: definicion XX_xi}}{=}\L_\xi\phi^I
		\end{align*}
		
		\mbox{}\vspace*{-6ex}
		
	\end{proof}
	
	It is important to stress the very different nature of both sides of \eqref{eq: LL_XX_xi=L_xi}. On the LHS we have the Lie derivative of the evaluation function $\mathrm{Eval}^I$ in the $\XX_\xi\in\campos(\F)$ direction (computation over $\F$), while on the RHS we have the Lie derivative of $\phi\in\F$ in the $\xi\in\mathfrak{X}(M)$ direction (computation over $M$).\vspace*{2ex}

	One last comment is in order now. If we consider a vector field $V\in\mathfrak{X}(M)$ in a finite-dimensional manifold $M$\!, we can write it in some local coordinates $\{x^i\}_i$ as \[V=\sum_{i=1}^n\d x^i(V)\frac{\partial}{\partial x^i}\]
	Sometimes, it is customary to proceed analogously for vector fields $\VV\in\campos(\F)$ in an infinite-dimensional manifold $\F$ and write something like
	\[\VV=\int_M \dd \phi^I(\VV)\frac{\delta}{\delta\phi^I}\qquad \text{ in which case  }\qquad\XX_\xi=\int_M\L_\xi\phi\frac{\delta}{\delta \phi^I}\] 
	but we strongly advise against this practice. First, because $\phi^I$ does not play the same role as the usual coordinates. Second, because it is hard to give a rigorous meaning to $\delta/\delta\phi^I$. Finally, because there is simply no need. For all purposes, we just need  $\dd\phi^I(\VV)=\VV^I_{\!\phi}$.

	\subsubsection*{Currents, Charges, and Potentials of a space-time vector field}
	In general $\XX_\xi$ is not a $\underline{\d}$-symmetry, so we can not define $\HH_{\XX_\xi}^{\SS,\imath}$. However, we can still define a $\xi$-charge $\QQ_\xi$ associated with $\xi\in\mathfrak{X}(M)$. Only under certain circumstances, to be studied at the end of this section, $\overline{\XX}_\xi$ is the Hamiltonian vector field generated by $\QQ_\xi$.
	
	\begin{definition}\label{def: objetos para L}\mbox{}\\
		Given $\xi\in\mathfrak{X}(M)$ and $(L,\overline{\ell})\in\Lag(M)$, we define the $\boldsymbol{\xi}$\textbf{-current} and \textbf{$\boldsymbol{\xi}$-charge}
		\[
		\definicion
		\tikzmarkin{Noetherprima}(0.2,-0.4)(-0.15,0.6)
		\tikzmarkin{Noether}(0.15,-0.4)(-0.1,0.6)
		\Big(J^\Theta_\xi,\overline{\jmath}_\xi^{(\Theta,\overline{\theta})}\Big):=\Big(\!S_\xi^L,\overline{s}_{\overline{\xi}}^{\overline{\ell}}\,\Big)-\underline{\ii}_{\XX_\xi}\Big(\Theta^L,\overline{\theta}^{(L,\overline{\ell})}\Big)\qquad\qquad\qquad\QQ^{(L,\overline{\ell}),\imath}_{\xi}:=\int_{(\Sigma,\partial\Sigma)}\underline{ \imath}^*\Big(J^\Theta_\xi,\overline{\jmath}_\xi^{(\Theta,\overline{\theta})}\Big)\in\OOmega^0(\F)
		\tikzmarkend{Noether}
		\tikzmarkend{Noetherprima}\]
		\vspace*{-1.5ex}
		
		where $\Big(S^L_\xi,\overline{s}_{\overline{\xi}}^{\overline{\ell}}\Big):=\underline{\iota}_\xi (L,\overline{\ell})$.\vspace*{2ex}
		
		We say that $(Q_\xi^\Theta,\overline{q}_\xi^{(\Theta,\overline{\theta})})\in\OOmega^{(n-2,0)}((M,\lateral M)\times\Sol(\SS))$ are $\boldsymbol{\xi}$\textbf{-potentials} (defined up to a closed term which might not be exact) if
		\[
		\definicion
		\tikzmarkin{preHamiltonianprima}(0.25,-0.3)(-0.25,0.5)
		\tikzmarkin{preHamiltonian}(0.2,-0.3)(-0.2,0.5)
		\underline{\d}\Big(Q_\xi^\Theta,\overline{q}_\xi^{(\Theta,\overline{\theta})}\Big)=\peqsub{\jj}{\SS}^*\Big(J^\Theta_\xi,\overline{\jmath}_\xi^{(\Theta,\overline{\theta})}\Big)
		\tikzmarkend{preHamiltonian}
		\tikzmarkend{preHamiltonianprima}\]
		
	\end{definition}
	Notice that the $\xi$-charge is associated with a particular pair of Lagrangians $(L,\overline{\ell})$ and not with the action $\SS=[\![(L,\overline{\ell})]\!]$ like the $\XX$-charge $\HH_\XX^{\SS,\imath}$. Those charges are candidates to relevant quantities, although its precise physical interpretation depends on the problem at hand.
	
	\begin{lemma}\label{lemma: deg objetos L}\mbox{}\\
		Given $(L_2,\overline{\ell}_2)\equivint(L_1,\overline{\ell}_1)$, we have\allowdisplaybreaks[0]
		\begin{align*}
		\degeneracion
		\tikzmarkin{NoetherDEG}(0.15,-0.45)(-0.15,0.5)
		&\Big(J^{\Theta_2}_\xi,\overline{\jmath}_\xi^{(\Theta_2,\overline{\theta}_2)}\Big)=\Big(J^{\Theta_1}_\xi,\overline{\jmath}_\xi^{(\Theta_1,\overline{\theta}_1)}\Big)+\Big(\underline{\L}_\xi-\underline{\LL}_{\XX_\xi}\Big)(Y,\overline{y})-\underline{\d}\Big(\underline{\iota}_\xi(Y,\overline{y})+\underline{\ii}_{\XX_\xi}(Z,\overline{z})\Big)\\
		&\QQ^{(L_2,\overline{\ell}_2),\imath}_{\xi}=\QQ^{(L_1,\overline{\ell}_1),\imath}_{\xi}+\int_{(\Sigma,\partial\Sigma)}\underline{\imath}^*\big(\underline{\L}_\xi -\underline{\LL}_{\XX_\xi})(Y,\overline{y})\hspace*{29ex}\mbox{}
		\tikzmarkend{NoetherDEG}
		\end{align*}
	\end{lemma}
	\begin{proof}\mbox{}\\
		The result follows from \eqref{eq: Cartan}, \eqref{eq: d dd=dd d}, lemmas \ref{lemma: (L2,l2)} and \ref{lemma: (Theta2,theta2)}, and the relative Stokes' theorem \ref{eq: stokes pair}.
	\end{proof}

	The operator $\L_\xi-\LL_{\XX_\xi}$ has a very clear interpretation. Recall first that $\phi\in\F$ represents the dynamical fields while $\tilde{\phi}\in\widetilde{\F}$ denotes the background objects (see section \ref{section: diff geom F}). The latter cannot vary over $\F$, i.e.~$\dd\tilde{\phi}=0$, however they play a non-trivial role in the computations over $M$\!. To see how, we duplicate all our geometric structures of $\F$ into $\widetilde{\F}$. For instance, we consider $\widetilde{\dd}$, $\widetilde{\LL}$, and so on.  We also define  $(\widetilde{\XX}_\xi)_{\tilde{\phi}}=\L_\xi\tilde{\phi}\in T_{\tilde{\phi}}\widetilde{\F}$. Notice in particular that we have $\widetilde{\dd}\phi=0$ and $\widetilde{\LL}_{\widetilde{\XX}_\xi}\tilde{\phi}=\L_\xi \tilde{\phi}$. 
	
	\begin{lemma}\label{lemma: L=LL+LL tilde}\mbox{}
		\begin{equation}\label{eq: L=LL+LL tilde}
		\comentario
		\tikzmarkin{LLLprima}(0.2,-0.32)(-0.2,0.5)
		\tikzmarkin{LLL}(0.2,-0.27)(-0.2,0.45)
		\underline{\L}_\xi=\underline{\LL}_{\XX_\xi}+\widetilde{\underline{\LL}}_{\widetilde{\XX}_\xi}
		\tikzmarkend{LLL}
		\tikzmarkend{LLLprima}
		\end{equation}
	\end{lemma}
	\begin{proof}\mbox{}\\
		From \eqref{eq: LL_XX_xi=L_xi}, we have $\LL_{\XX_\xi}\phi=\L_\xi \phi$ and $\widetilde{\LL}_{\widetilde{\XX}_\xi}\tilde{\phi}=\L_\xi \tilde{\phi}$. Now applying \eqref{eq: Cartan}-\eqref{eq: Leibniz imath} leads to the result.
	\end{proof}
	
	This result is what one should expect: $\LL$ is unaware of $\tilde{\phi}\in\widetilde{\F}$ while $\widetilde{\LL}$ is unaware of $\phi\in\F$. However, $\L$ sees equally $\F$ and $\widetilde{\F}$. It is important to mention that in some examples there is no clear distinction between the dynamical and background objects. For instance, one might consider $\F$ as the space of metrics on $M$ with some prescribed scalar curvature $R$. Nonetheless, on these cases everything works out if one avoids $\widetilde{\underline{\LL}}_{\widetilde{\XX}_\xi}$ (and the space $\widetilde{\F}$ altogether) and uses $\L_\xi-\LL_{\XX_\xi}$ instead. This operator has also been considered before in \cite{Freidel}.
	
	\subsubsection*{Flux law}
	\begin{proposition}\label{proposition: flux law}\mbox{}\\
		Let $(L,\overline{\ell})\in\Lag(M)$ and $\underline{\imath}_1,\underline{\imath}_2:(\Sigma,\partial\Sigma)\hookrightarrow(M,\lateral M)$ be two non-intersecting embeddings and $N$ the manifold bounded by those two Cauchy hypersurfaces. Then, we have the following flux law	
		\[
		\comentario\tikzmarkin{fluxlawprima}(0.2,-0.5)(-0.2,0.65)
		\tikzmarkin{fluxlaw}(0.2,-0.45)(-0.2,0.6)
		\peqsub{\jj}{\SS}^*\QQ^{(L,\overline{\ell}),\imath_2}_\xi-\peqsub{\jj}{\SS}^*\QQ^{(L,\overline{\ell}),\imath_1}_\xi=\int_{(N,\lateral N)}\widetilde{\underline{\LL}}_{\tilde{\XX}_\xi}(L,\overline{\ell})
		\tikzmarkend{fluxlaw}
		\tikzmarkend{fluxlawprima}\]	
	\end{proposition}
	\begin{proof}\mbox{}\\
		We denote $\underline{i}:\underline{\partial}(N,\lateral N)\hookrightarrow (N,\lateral N)$. The exterior derivate of the $\xi$-current is
		\begin{align}\label{eq: (dJxi,djxi)}
		\begin{split}
		\underline{\d}\Big(J^\Theta_\xi,\overline{\jmath}^{(\Theta,\overline{\theta})}_\xi\Big)\!(\phi)&\overset{\eqref{def: objetos para L}}{=}\underline{\d}\left(\underline{\iota}_\xi(L,\overline{\ell})-\underline{\ii}_{\XX_\xi}\Big(\Theta^L,\overline{\theta}^{(L,\overline{\ell})}\Big)\right)\!(\phi)\updown{\eqref{eq: Cartan}\eqref{eq: d dd=dd d}}{\eqref{eq: d(L,l)}\eqref{eq: L=LL+LL tilde}\eqref{eq: LL_XX_xi=L_xi}}{=}\\
		&\hspace*{2.1ex} =\widetilde{\underline{\LL}}_{\widetilde{\XX}_\xi}(L,\overline{\ell})(\phi)+(E_I,\overline{b}_I)(\phi)\L_\xi\phi^I
		\end{split}\end{align} 
		Notice, by the way, that if $(L,\overline{\ell})$ is $\xi$-invariant, then \eqref{eq: (dJxi,djxi)} coincides with \eqref{eq: (dJ,dj)}.
		Integrating \eqref{eq: (dJxi,djxi)}
		\begin{align*}
		\int_{(N,\lateral N)}\Big(\widetilde{\underline{\LL}}_{\widetilde{\XX}_\xi}&(L,\overline{\ell})(\phi)+(E_I,\overline{b}_I)(\phi)\L_\xi\phi^I\Big)=\int_{(N,\lateral N)}\underline{\d}\Big(J^\Theta_\xi,\overline{\jmath}^{(\Theta,\overline{\theta})}_\xi\Big)\!(\phi)\overset{\eqref{eq: stokes pair}}{=}\\
		&=\int_{\underline{\partial}(N,\lateral N)}\underline{i}^*\Big(J^\Theta_\xi,\overline{\jmath}^{(\Theta,\overline{\theta})}_\xi\Big)\!(\phi)=\peqsub{\jj}{\SS}^*\QQ^{(L,\overline{\ell}),\imath_2}_\xi-\peqsub{\jj}{\SS}^*\QQ^{(L,\overline{\ell}),\imath_1}_\xi
		\end{align*}
		where we assume that $\imath_2(\Sigma)$ lies in the future of $\imath_1(\Sigma)$ and we have taken into account that each connected component of $\underline{\partial}(N,\lateral N)=(\Sigma_1,\partial\Sigma_1)\cup(\Sigma_2,\partial\Sigma_2)$ has opposite orientation.
	\end{proof}
	
	We see that a $\xi$-charge might depend on the Cauchy surface so, in particular, its pullback to the space of solutions is not a Hamiltonian function in general.
	
	\begin{lemma}\mbox{}
				\[
		\comentario\tikzmarkin{dQprima}(0.15,-0.5)(-0.15,0.65)
		\tikzmarkin{dQ}(0.15,-0.45)(-0.15,0.6)
	\dd\QQ^{(L,\overline{\ell}),\imath}_\xi=\ii_{\XX_\xi}\OOmega^\imath_\SS+\int_{(\Sigma,\partial\Sigma)}\underline{\imath}^*\Big(\underline{\imath}_\xi(E_I,\overline{b}_I)\wedge\dd\phi^I\Big)+\int_{(\Sigma,\partial\Sigma)}\underline{\imath}^*\widetilde{\underline{\LL}}_{\tilde{\XX}_\xi}\Big(\Theta^L,\overline{\theta}^{(L,\overline{\ell})}\Big)
		\tikzmarkend{dQ}
		\tikzmarkend{dQprima}\]
	\end{lemma}
	\begin{proof}\mbox{}
		\begin{align*}
		\dd\QQ^{(L,\overline{\ell}),\imath}_\xi&\updown{\eqref{def: objetos para L}}{\eqref{eq: d dd=dd d}}{=}\int_{(\Sigma,\partial\Sigma)}\underline{\imath}^*\Big(\underline{\iota}_\xi\underline{\dd}(L,\overline{\ell})-\underline{\dd}\,\underline{\ii}_{\XX_\xi}\Big(\Theta^L,\overline{\theta}^{(L,\overline{\ell})}\Big)\Big)\updown{\eqref{eq: d(L,l)}\eqref{eq: Cartan}}{\eqref{eq: stokes pair}\eqref{eq: omega and overline omega}}{=}\\
		&=\int_{(\Sigma,\partial\Sigma)}\underline{\imath}^*\Big(\underline{\iota}_\xi(E_I,\overline{b}_I)\wedge\dd\phi^I+\big(\underline{\L}_\xi-\underline{\LL}_{\XX_\xi}\big)\Big(\Theta^L,\overline{\theta}^{(L,\overline{\ell})}\Big)+\underline{\ii}_{\XX_\xi}\Big(\Omega^\Theta,\overline{\omega}^{(\Theta,\overline{\theta})}\Big)\Big)
		\end{align*}
		And the results follows from \eqref{eq: L=LL+LL tilde} and \eqref{def: OOmega}.
	\end{proof}
	
	\subsubsection*{When is \texorpdfstring{$\boldsymbol{\XX_\xi}$}{Xₓ} a \texorpdfstring{$\boldsymbol{\d}$}{d}-symmetry?}
	
	We say that $\xi\in\mathfrak{X}(M)$ is a symmetry of $(L,\overline{\ell})\in\Lag(M)$, or that $(L,\overline{\ell})$ is $\xi$-invariant, if
	\begin{equation}\label{eq: L_xi=LL_XX}
	\underline{\L}_\xi(L,\overline{\ell})=\underline{\LL}_{\XX_\xi}(L,\overline{\ell})\qquad\qquad\equiv\qquad\qquad\widetilde{\underline{\LL}}_{\widetilde{\XX}_\xi}(L,\overline{\ell})=0
	\end{equation}
	Notice that, in this case, $(L_2,\overline{\ell}_2)=(L,\overline{\ell})+\underline{\d}(Y,\overline{y})$ is $\xi$-invariant if and only if $(Y,\overline{y})$ is $\xi$-invariant.\vspace*{2ex}
	
	As a final remark, notice that a sufficient (but not necessary) condition for $L$ to be $\xi$-invariant is that $\L_\xi\tilde{\phi}=0$ which, in turn, implies that $\xi$ leaves the background objects invariant. That happens, for instance, when we take $\xi$ as a Killing vector field of a fixed metric $g\in\widetilde{\F}$ of the theory.
	
	\begin{proposition}\label{proposition: XX_xi symmetry}\mbox{}\\
		If $(L,\overline{\ell})\in\Lag(M)$ is $\xi$-invariant, then $\XX_\xi$ is a $\underline{\d}$-symmetry of $\SS=[\![(L,\overline{\ell})]\!]$. 
	\end{proposition}
	\begin{proof}\mbox{}
		\[\underline{\LL}_{\XX_\xi}(L,\overline{\ell})\overset{\eqref{eq: L_xi=LL_XX}}{=}\underline{\L}_\xi(L,\overline{\ell})\overset{\eqref{eq: LLdd=ddLL}}{=}\underline{\d}\,\underline{\iota}_\xi(L,\overline{\ell})\]
		where we recall that $\xi$ is tangent to $\lateral M$\!. From definition \ref{def: d-symmetry} we see that $\XX_\xi$ is a $\underline{\d}$-symmetry with
		\begin{equation}\label{eq: (S,s)=imat(L,l)}
		\Big(S_{\XX_\xi}^L,\overline{s}_{\XX_\xi}^{(L,\overline{\ell})}\Big)=\underline{\iota}_\xi( L,\overline{\ell})=\big(S_\xi^L,\overline{s}_{\overline{\xi}}^{\overline{\ell}}\big)
		\end{equation}
		\mbox{}\vspace*{-5.5ex}
		
	\end{proof}
	
	\begin{remark}\label{remark: inverse xi-inv}\mbox{}\\
		It is plausible that the converse holds. Indeed, if $\XX_\xi$ is a $\underline{d}$-symmetry, then
		\[\underline{\d}\Big(S^L_{\XX_\xi},\overline{s}^{(L,\overline{\ell})}_{\XX_\xi}\Big)=\underline{\LL}_{\XX_\xi}(L,\overline{\ell})\]
		We want now a  $\xi$-invariant $(L_2,\overline{\ell}_2)\equivint(L,\overline{\ell})$. It is not hard to prove that this is equivalent, by lemma \ref{lemma: (L2,l2)}, to solving
		\[\widetilde{\LL}_{\widetilde{\XX}_\xi}(Y,\overline{y})=\Big(S^L_{\XX_\xi},\overline{s}^{(L,\overline{\ell})}_{\XX_\xi}\Big)-\underline{\iota}_\xi(L,\overline{\ell})+\underline{\d}(M,\overline{m})\]
		for  $(Y,\overline{y})$ and $(M,\overline{m})$. There are always local solutions and, under the right topological conditions, we expect that those solutions can be glued together although we have not investigated this further.
	\end{remark}
	
	If $(L,\overline{\ell})$ is $\xi$-invariant, we have seen that $\XX_\xi\in\Sym_{\d}(\SS)$. Thus we have both the $\XX_\xi$-charge and the $\xi$-charge. Let us see that, in that case, they are the same.
	
	\begin{lemma}\mbox{}\\
		If $(L,\overline{\ell})\in\Lag(M)$ is $\xi$-invariant, then $\HH^{\SS,\imath}_{\XX_\xi}=\QQ^{(L,\overline{\ell}),\imath}_{\xi}$. In particular, $\peqsub{\jj}{\SS}^*\QQ^{(L,\overline{\ell}),\imath}_{\xi}$ is independent of $(L,\overline{\ell})$ and $\imath:\Sigma\hookrightarrow M$\!.
	\end{lemma}
	\begin{proof}\mbox{}\\
		Plugging \eqref{eq: (S,s)=imat(L,l)} on the $\XX_\xi$-charge \ref{def: Noether charge} leads to the $\xi$-charge \ref{def: objetos para L}. The last statement is immediate from proposition \ref{proposition: HHXX}.
	\end{proof}
	
	It might seem strange that $\HH^{\SS,\imath}_{\XX_\xi}$ does not depend on the choice of Lagrangians while $\QQ^{(L,\overline{\ell}),\imath}_{\xi}$ does. Notice first that, from lemma \ref{lemma: deg objetos L}, we have
	\[\QQ^{(L_2,\overline{\ell}_2),\imath}_{\xi}=\QQ^{(L,\overline{\ell}),\imath}_{\xi}+\int_{(\Sigma,\partial\Sigma)}\underline{\imath}^*\underline{\widetilde{\LL}}_{\XX_\xi}(Y,\overline{y})\]
	
	The dependence on the representative comes from the fact that $(S_\xi^{L_2},\overline{s}_{\overline{\xi}}^{\overline{\ell}_2})$ does not transform as equation \eqref{eq: alpha beta def} unless $(\underline{\L}_\xi-\underline{\LL}_{\XX_\xi})(Y,\overline{y})=0$, in which case the last integral vanishes. Summarizing: if $\XX_\xi$ is a symmetry, then the $\xi$-charge does not depend on the representative but, if we pick a $\xi$-invariant pair of Lagrangian, the integrand is simpler in general.

	\begin{corollary}\label{proposition: xi Hamilton eq}\mbox{}\\
		If $(L,\overline{\ell})\in\Lag(M)$ is $\xi$-invariant
		\[
		\comentario
		\tikzmarkin{HamiltonianEqxiprima}(0.2,-0.33)(-0.2,0.55)
		\tikzmarkin{HamiltonianEqxi}(0.2,-0.28)(-0.2,0.5)
		\ii_{\overline{\XX}_\xi}\OOmega_\SS=\dd\peqsub{\jj}{\SS}^*\QQ^{(L,\overline{\ell}),\imath}_{\xi}
		\tikzmarkend{HamiltonianEqxi}
		\tikzmarkend{HamiltonianEqxiprima}
		\]
	\end{corollary}

	\subsection{Diff-invariant theories}\label{section: diff invariant}
	From now on, we assume that we have the action of diffeomorphisms $\mathrm{Diff}(M)\times\F\to\F$, given by $(\psi,\phi)\mapsto\psi^*\phi\in\F$ through pullbacks and/or pushforward according to the tensorial character of $\phi$. If $\F$ has no constraints, this is certainly the case. Otherwise, we might have to restrict the action to a smaller group $\mathrm{Diff}_0(M)$ of diffeomorphism respecting the constraints. If we consider Dirichlet conditions, where the values of the fields are zero at the boundary, we still have the full group of diffeomorphism acting on $\F$. However, for more complicated boundary conditions such as non-homogeneous Dirichlet conditions or constraints in the bulk, things may be a bit more complicated. For instance, if we consider the space of metrics with a fixed value at the boundary
	\[\F=\{g\in\mathrm{Met}(M)\ /\ \jmath^*g=\overline{g}_0\}\]
	then the diffeomorphism acting on $\F$ must restrict to a $\overline{g}_0$-isometry over the boundary i.e.~$\psi^*\overline{g}_0=\overline{g}_0$.\vspace*{2ex}
	
	An action $\SS$ is \textbf{Diff-invariant} if $\SS(\psi^*\phi)=\SS(\phi)$ for every $\psi\in\mathrm{Diff}(M)$. A form $\alpha\in\OOmega^{(r,0)}(M\times \F)$ is \textbf{Diff-invariant} if  $\alpha(\psi^*\phi)=(\psi^*\alpha)(\phi)$. It is important to realize the very different character of the action of the diffeomorphism in the last equation. On the LHS we are using the aforementioned action while on the RHS we use the usual action over $\OOmega^{(r,s)}(M\times\F)$, so $\psi^*\alpha\in\OOmega^{(r,0)}(M\times \F)$ which, in turn, implies that $(\psi^*\alpha)(\phi)\in\Omega^r(M)$ according to equation \eqref{eq: projection form}.\vspace*{2ex}
	
	If $(L,\overline{\ell})\in\Lag(M)$ is Diff-invariant, then so is $\SS=[\![(L,\overline{\ell})]\!]$. However, the converse is not true: we can build some $(L_2,\overline{\ell}_2)\equivint(L,\overline{\ell})$ which is not Diff-invariant by taking $(Y,\overline{y})$ not Diff-invariant in \eqref{eq: L2=L1}.
	
	\begin{lemma}\label{lemma: LL_XX=L_xi}\mbox{}\\
		If $(L,\overline{\ell})\in\Lag(M)$ is Diff-invariant, then $(L,\overline{\ell})$ is $\xi$-invariant for every $\xi\in\mathfrak{X}(M)$.
	\end{lemma}
	\begin{proof}\mbox{}\\
		For a given $\xi\in\mathfrak{X}(M)$, we take $\{\psi_t\}_t$ a path of diffeomorphisms with $\psi_0=\mathrm{Id}$ and $\xi=\partial_t|_0\psi_t$. Then
		\begin{align*}
		\underline{\L}_\xi&(L,\overline{\ell})(\phi)=\left.\frac{\d}{\d t}\right|_{t=0}\psi_t^*(L,\overline{\ell})(\phi)=\left.\frac{\d}{\d t}\right|_{t=0}(L,\overline{\ell})(\psi_t^*\phi)=\underline{\dd}_\phi(L,\overline{\ell})\left(\!\left.\frac{\d}{\d t}\right|_{t=0}\!\psi_t^*\phi\right)=\\
		&=\underline{\dd}_\phi(L,\overline{\ell})(\L_\xi\phi)\overset{\eqref{eq: definicion XX_xi}}{=}\underline{\dd}_\phi(L,\overline{\ell})(\XX_\xi)=\Big(\underline{\ii}_{\XX_\xi}\underline{\dd}(L,\overline{\ell})\Big)\!(\phi)\overset{\eqref{eq: Cartan}}{=}\underline{\LL}_{\XX_\xi}(L,\overline{\ell})(\phi)
		\end{align*}	
		\mbox{}\vspace*{-6.5ex}
		
	\end{proof}
	
	With this last lemma, we can apply all the results of the previous section.

	\begin{proposition}\label{proposition: Diff invariant L S}\mbox{}
		\begin{itemize}
			\item If $(L,\overline{\ell})\in\Lag(M)$ is Diff-invariant, then $\XX_\xi\in\Sym_{\d}([\![(L,\overline{\ell})]\!])$ for every $\xi\in\mathfrak{X}(M)$.
			\item If $(L,\overline{\ell})\in\Lag(M)$ is Diff-invariant, then so is $\SS=[\![(L,\overline{\ell})]\!]$.
		\end{itemize}
	\end{proposition}
	
	\begin{remark}\mbox{}\\
		Again, it is plausible that some sort of converse holds. Namely, if $\SS$ is Diff-invariant, then there exists a representative $(L,\overline{\ell})\in\Lag(M)$ that is Diff-invariant. A similar heuristic argument as the one used in \ref{remark: inverse xi-inv} applies here but, again, we have not investigated this issue further.
	\end{remark}
	
	\begin{lemma}\label{lemma: J=dQ}\mbox{}
		\begin{itemize}
			\item If $L$ is Diff-invariant, then there exists a $\xi$-charge $Q^\Theta_\xi\in\OOmega^{(n-2,0)}(M\times\Sol(\SS))$. In particular,
			\begin{equation}\label{eq: QQ=int borde}
			\QQ^{(L,\overline{\ell}),\imath}_\xi=\int_{\partial\Sigma}\overline{\imath}^*\left(Q^\Theta_\xi-\peqsub{\jj}{\SS}^*\overline{\jmath}_\xi^{(\Theta,\overline{\theta})}\right)
			\end{equation}
			\item If $(L,\overline{\ell})$ is Diff-invariant, then there exist $\xi$-charges $(Q_\xi^\Theta,\overline{q}_\xi^{(\Theta,\overline{\theta})})\in\OOmega^{(n-2,0)}((M,\lateral M)\times\Sol(\SS))$. In particular,
			\begin{equation}
			\QQ^\SS_\xi=0
			\end{equation}
		\end{itemize}
	\end{lemma}
	\begin{proof}\mbox{}\\
		We can view the Noether current as $J^\Theta_{\bullet}\in\OOmega^{(n-1,0)}(M\times(\F\times\mathfrak{X}(M)))$ i.e.~taking $\xi\in\mathfrak{X}(M)$ as a dynamical field. We now fix $\phi\in\F$, so we have $J^\Theta_{\bullet}(\phi)\in\OOmega^{(n-1,0)}(M\times\F')$ where $\F':=\mathfrak{X}(M)$. From equations \eqref{eq: (dJxi,djxi)} and \eqref{eq: L_xi=LL_XX}, we deduce
		\[\d J^\Theta_\bullet(\phi)=E_I(\phi)\L_\bullet\phi^I\in\OOmega^{(n,0)}(M\times\F')\]
		If $\phi\in\Sol(\SS)$ we see that $J^\Theta_\bullet(\phi)$ is closed and $J_0^\Theta(\phi)=0$ so, by the horizontal exactness theorem \ref{theorem: Wald exact - close} applied to $\F'$, it follows that it is $\d$-exact. The same argument using pairs applies to the second statement. The values of the $\xi$-charges are obtained applying the (relative) Stokes' theorem.
	\end{proof}
	\begin{remark}\mbox{}\\
		Notice that in \eqref{eq: QQ=int borde} we could have considered the $\xi$-potential $Q^\Theta_\xi+Q'$ for some $Q'\in\Omega^{n-2}(M)$ closed but not exact. It might seem that an additional term appears but in fact it vanishes
		\[\int_{\partial\Sigma}\overline{\imath}^*Q'=\int_{\Sigma}\imath^*\d Q'=0\qquad\longrightarrow\qquad\text{in particular }\overline{\imath}^*Q'\text{ is exact}\]
	\end{remark}

	\section{CPS-Algorithm}\label{section: summary of algorithm}
	In this section we provide the following four-step algorithm, that we denote CPS-algorithm, to obtain the symplectic structure over the space of solutions.
	
	\begin{enumerate}[label=\protect\circled{\arabic*}]\setcounter{enumi}{-1}
		\item Given an action $\SS:\F\to\R$ of a well defined theory, choose any Lagrangians $(L,\overline{\ell})$ with
		\[\SS=\int_M L-\int_{\partial M}\overline{\ell}\]
		\item Compute $\dd L=E_I\wedge\dd\phi^I+\d\Theta^L$. Choose any $\Theta^L$.
		\item Compute $\dd\overline{\ell}-\jmath^*\Theta^L=\overline{b}_I\wedge\dd\phi^I-\d\overline{\theta}^{(L,\overline{\ell})}$  over $\lateral M$\!. Choose any $\overline{\theta}^{(L,\overline{\ell})}$.
		\item Define $\Sol(\SS)=\{\phi\in\F\ /\ E_I(\phi)=0, \ \ \overline{b}_I(\phi)=0\}$ and the inclusion $\peqsub{\jj}{\SS}:\Sol(\SS)\hookrightarrow\F$.
		\item Compute the presymplectic structure
		\[\OOmega_\SS^\imath=\dd\left(\int_\Sigma\imath^*\Theta^L-\int_{\partial\Sigma}\overline\imath^*\overline{\theta}^{(L,\overline{\ell})}\right)\qquad\longrightarrow\qquad\OOmega_\SS=\peqsub{\jj}{\SS}^*\OOmega^\imath_\SS\]
		where $\imath:\Sigma\hookrightarrow M$ is a Cauchy embedding and $\overline{\imath}:=\imath|_{\partial\Sigma}:\partial\Sigma\hookrightarrow\lateral M$ is its restriction.
	\end{enumerate}
	Once we have the presymplectic structure, it is useful to add the following two steps that will provide a deeper insight into the theory at hand.
	
	\begin{enumerate}[label=\protect\circled{\arabic*}]\setcounter{enumi}{4}
		\item Study symmetries. Study if $\XX_\xi$ is a $\underline{\d}$-symmetry. Obtain $\xi$-currents, $\xi$-charges, and $\xi$-flux laws.
		\[\begin{array}{|l}
		\!J^\Theta_\xi=\iota_\xi L-\ii_{\XX_\xi}\Theta\\
		\!\overline{\jmath}^{(\Theta,\overline{\theta})}_{\xi}=-\iota_{\overline{\xi}}\overline{\ell}-\ii_{\XX_\xi}\overline{\theta}
		\end{array}\qquad\qquad\qquad \QQ^{(L,\overline{\ell}),\imath}_\xi=\int_\Sigma\imath^*J^\Theta_\xi-\int_{\partial\Sigma}\overline{\imath}^*\overline{\jmath}^{(\Theta,\overline{\theta})}_{\xi}\]
		\item If possible, compare with the presymplectic structure coming from the Hamiltonian formulation.
	\end{enumerate}

	\section{Examples}\label{section: examples}
	Let $(M,g)$ be a connected, oriented, globally hyperbolic $n$-space-time with boundary (possibly empty). Let $\jmath:\partial M\hookrightarrow M$ be the inclusion and $\overline{g}:=\jmath^*g$. Without loss of generality, we define the Lagrangians instead of the action. Moreover, as we have to perform first the computations over the bulk and subsequently over the boundary, we do not use the relative notation. Nonetheless, keep in mind that several results, like the $\xi$-charges, admit a very compact expression using the relative framework.
	
	\subsection{Scalar field with Robin boundary conditions}
	We consider $\F=\Omega^0(M)$ and define
	
	\[\definicion
	\tikzmarkin{definitionL1prima}(0.25,-0.42)(-0.25,0.62)
	\tikzmarkin{definitionL1}(0.2,-0.42)(-0.2,0.62)
	L(\phi)=\left(\frac{1}{2}g^{-1}(\d\phi,\d\phi)+V(\phi)\right)\!\vol_g\qquad\qquad\overline{\ell}(\phi)=\frac{1}{2}f\cdot \overline{\phi}^2\vol_{\overline{g}}
	\tikzmarkend{definitionL1}
	\tikzmarkend{definitionL1prima}
	\]
	for some $f:\partial M\to\R$. We have defined $\overline{\phi}:=\jmath^*\phi$.
	\begin{flalign*}
	&\pushleft{\circled{1}}\MoveEqLeft[1]
	\dd_\phi L&=\Big(\nabla^{\alpha}\phi\nabla_{\!\alpha}\dd\phi+V'(\phi)\dd\phi\Big)\vol_g=\\
	&\MoveEqLeft[1]&=\nabla_{\!\alpha}(\dd\phi\nabla^\alpha\phi)\vol_g-\Big(\square_g\phi-V'(\phi)\Big)\dd\phi\vol_g=\\
	&\MoveEqLeft[1]&=-\Big(\square_g\phi-V'(\phi)\Big)\vol_g\dd\phi+\d\big(\iota_{\vec{V}}\vol_g\big)\quad\qquad\qquad\longrightarrow\quad\qquad\qquad\Theta^L(\phi)=\iota_{\vec{V}}\vol_g\MoveEqLeft[1]
	\end{flalign*}
	where we define $\square_g:=g^{\alpha\beta}\nabla_{\!\alpha}\nabla_{\!\beta}$ and $\vec{V}:=\dd\phi\vec{\nabla}\phi$. Notice that in the last equality we have used that, by definition of divergence, $(\mathrm{div}_g\vec{V})\vol_g=\L_{\vec{V}}\vol_g=\d\iota_{\vec{V}}\vol_g$.
	\begin{flalign*}
	&\pushleft{\circled{2}}\MoveEqLeft[1] 
	\dd_\phi\overline{\ell}-\jmath^*\Theta^L(\phi)&=f\overline{\phi}\,\dd\overline{\phi}\vol_{\overline{g}}-\jmath^*(\iota_{\vec{V}}\vol_g)\overset{\eqref{eq: orientation sigma}}{=}\\
	&\MoveEqLeft[1]&=\left(f\overline{\phi}\,\dd\overline{\phi}-\jmath^*(\nu_\alpha V^\alpha)\right)\vol_{\overline{g}}=\\
	&\MoveEqLeft[1]&=-\left(\jmath^*\nabla_{\!\vec{\nu}}\phi-f\overline{\phi}\right)\vol_{\overline{g}}\,\dd\overline{\phi}\quad\qquad\longrightarrow\quad\qquad\overline{\theta}^{(L,\overline{\ell})}(\phi)=0\MoveEqLeft[1]
	\end{flalign*}
	\begin{flalign*}
	&\pushleft{\circled{3}}\MoveEqLeft[1]& 
	\begin{array}{|l}
	E(\phi)=-\Big(\square_g\phi-V'(\phi)\Big)\vol_g\\%
	\overline{b}(\phi)=-\left(\jmath^*\nabla_{\!\vec{\nu}}\phi-f\overline{\phi}\right)\vol_{\overline{g}}%
	\end{array}\qquad \Sol\left([\![(L,\overline{\ell})]\!]\right)=\Big\{\phi\in\Omega^0(M)\ /\ E(\phi)=0,\quad \overline{b}(\phi)=0\Big\}
	\end{flalign*} 
	Notice that we obtain the so-called Robin boundary conditions $\jmath^*\nabla_{\!\vec{\nu}}\phi=f\overline{\phi}$. As a particular case, if $f=0$, we obtain Neumann boundary conditions $\jmath^*\nabla_{\!\vec{\nu}}\phi=0$.
	
	\begin{flalign*}
	&\pushleft{\circled{4}}\MoveEqLeft[1]
	(\OOmega_\SS^\imath)_\phi&=\dd\int_\Sigma\imath^*\Theta^L(\phi)-\dd\int_{\partial\Sigma}\overline{\imath}^*\overline{\theta}^{(L,\overline{\ell})}(\phi)=\\
	&\MoveEqLeft[1]&=\int_\Sigma\dd\imath^*(\iota_{\vec{V}}\vol_g)-0\overset{\eqref{eq: orientation sigma}}{=}-\int_\Sigma\dd \imath^*\!\left(n^\alpha V_\alpha \vol_\gamma\right)=\\
	&\MoveEqLeft[1]&=-\int_\Sigma  \dd(\dd\varphi \,\imath^*\nabla_{\!\vec{n}}\phi) \vol_\gamma\overset{\eqref{eq: Leibniz}}{=}\int_\Sigma \dd\varphi \wwedge \dd\imath^*\!(\nabla_{\!\vec{n}}\phi) \vol_\gamma=\\
	&\MoveEqLeft[1]&=\int_\Sigma \dd\varphi \wwedge \imath^*\nabla_{\vec{n}}\dd\phi \vol_\gamma=\int_\Sigma \dd\varphi \wwedge \dd\imath^*\!\left(\L_{\vec{n}}\phi\right)\vol_\gamma\MoveEqLeft[1]
	\end{flalign*} 
	where $\vec{n}\in\mathfrak{X}(M)$ is the vector field $g$-normal to $\imath(\Sigma)\subset M$ and $\varphi:=\imath^*\phi\in\Omega^0(\Sigma)$. Recall that $\dd$ only acts on $\phi$ and not to $g$, as the latter is a fixed background structure of $M$\!.\vspace*{2ex}
	
	\circled{5} Using \eqref{eq: L_xi=LL_XX} and \ref{proposition: XX_xi symmetry}, we see that $\XX_\xi$ is a symmetry if and only if $\L_\xi g=0$ and $\L_{\overline{\xi}} f=0$. Using now \eqref{eq: orientation sigma},  \eqref{eq: reverse orientation}, and \ref{def: objetos para L} we obtain
	\[\QQ^{(L,\overline{\ell}),\imath}_\xi(\phi)=\int_\Sigma n_\alpha \xi_\beta\left(\nabla^\alpha\phi\nabla^\beta\phi-g^{\alpha\beta}\left[\frac{1}{2}g^{-1}(\d\phi,\d\phi)+V(\phi)\right]\right)\vol_\gamma-\int_{\partial\Sigma}\overline{m}_{\overline{\alpha}}\overline{\xi}_{\overline{\beta}}\left(\overline{g}^{\overline{\alpha}\overline{\beta}}\frac{1}{2}f\cdot{}\overline{\phi}^2\right)\vol_{\overline{\gamma}}\]
	The first term in parenthesis is twice the energy-momentum tensor $T^{\alpha\beta}$ given by $\widetilde{\dd}L=(T^{\alpha\beta}\vol_g)\widetilde{\dd} g_{\alpha\beta}$. Analogously for the boundary with $\widetilde{\dd}\overline{\ell}=(\overline{t}^{\overline{\alpha}\overline{\beta}}\vol_{\overline{g}})\widetilde{\dd}\overline{g}_{\overline{\alpha}\overline{\beta}}$. This leads precisely to the flux law \ref{proposition: flux law}.\vspace*{2ex}
	
	\circled{6} Let us now consider the Hamiltonian decomposition. Given an embedding $\imath:\Sigma\hookrightarrow M$, we can break the objects involved into the tangential and perpendicular part. In this case we have $\varphi:=\imath^*\phi$ and the metric $g\leftrightarrow(\gamma,N,\vec{N})$ where $\gamma:=\imath^*g$, $N$ is the lapse, and $\vec{N}$ the shift. The momentum, after identifying the cotangent with the tangent bundle (so we work with densities), is given \cite{margalef2016hamiltonianb,tesis} by
	\begin{equation}\label{eq: p campo esacalar}
	p=\frac{v-\L_{\vec{N}}\varphi}{N}=\imath^*\!\left(\frac{\L_{\partial_t}\phi-\L_{\vec{N}}\phi}{N}\right)=\imath^*\L_{\vec{n}}\phi
	\end{equation}
	where $v\in T_\phi \F$. The canonical symplectic form of the cotangent bundle is
	\[\Omega_{(\varphi,p)}((V_\varphi,V_p),(W_\varphi,W_p))=\int_\Sigma\Big(V_\varphi W_p-W_\varphi V_p\Big)\vol_\gamma\quad\qquad\equiv\quad\qquad\Omega_{(\varphi,p)}=\int_\Sigma\dd\varphi\wwedge\dd p \,\vol_\gamma\]
	while over $\Sol(\SS)$ we have
	\[(\OOmega_\SS)_\phi=\int_\Sigma \dd(\imath^*\phi) \wwedge \dd\imath^*\!\left(\L_{\vec{n}}\phi\right)\vol_\gamma\]
	So, by using equation \eqref{eq: p campo esacalar}, we obtain $(\mathcal{P}^*\Omega)_\phi=(\OOmega_\SS)_{\phi}$, where $\mathcal{P}:\Sol(\SS)\to\mathrm{CDS}(\Sigma)$ is the polarization map given the initial Cauchy data $(\imath^*\phi,\imath^*\L_{\vec{n}}\phi)\in T^*\Omega^0(\Sigma)$. 
	
	\subsection{Scalar field with Dirichlet boundary conditions}
	We consider $\F=\Omega^0(M)_0$ i.e.~scalar fields $\phi$ of $M$ with $\overline{\phi}:=\jmath^*\phi=0$ (in particular, $\dd\overline{\phi}=0$). We define
	\[\definicion
	\tikzmarkin{definitionL2prima}(0.25,-0.42)(-0.25,0.62)
	\tikzmarkin{definitionL2}(0.2,-0.42)(-0.2,0.62)
	L(\phi)=\left(\frac{1}{2}g^{-1}(\d\phi,\d\phi)+V(\phi)\right)\!\vol_g\qquad\qquad\overline{\ell}(\phi)=0
	\tikzmarkend{definitionL2}
	\tikzmarkend{definitionL2prima}
	\]

	\begin{flalign*}
	&\pushleft{\circled{1}}\MoveEqLeft[1]& 
	\dd_\phi L=-\Big(\square_g\phi-V'(\phi)\Big)\vol_g\dd\phi+\d\big(\iota_{\vec{V}}\vol_g\big)\qquad\longrightarrow\qquad\Theta^L(\phi)=\iota_{\vec{V}}\vol_g\MoveEqLeft[1]
	\end{flalign*}
	\begin{flalign*}
	&\pushleft{\circled{2}}\MoveEqLeft[1]&\dd_\phi\overline{\ell}-\jmath^*\Theta^L(\phi)=-\jmath^*\nabla_{\!\vec{\nu}}\phi\vol_{\overline{g}}\,\dd\overline{\phi}=0\qquad\longrightarrow\qquad\overline{\theta}^{(L,\overline{\ell})}(\phi)=0\MoveEqLeft[1]
	\end{flalign*}
	\begin{flalign*}
	&\pushleft{\circled{3}}\MoveEqLeft[1]&
	\begin{array}{|l}
	E(\phi)=-\Big(\square_g\phi-V'(\phi)\Big)\vol_g\\
	\overline{b}(\phi)=0
	\end{array}\qquad \Sol\left([\![(L,\overline{\ell})]\!]\right)=\Big\{\phi\in\Omega^0(M)_0\ /\ E(\phi)=0,\quad \overline{b}(\phi)=0\Big\}
	\end{flalign*}
	\begin{flalign*}
	&\pushleft{\circled{4}}\MoveEqLeft[1]&
	(\OOmega_\SS^\imath)_\phi=\int_\Sigma \dd\varphi \wwedge \imath^*\!\left(\L_{\vec{n}}\dd\phi\right)\vol_\gamma\MoveEqLeft[1]
	\end{flalign*}
	
	\circled{5}  This is analogous to the previous case i.e.~$\XX_\xi$ is a symmetry if and only if $\L_\xi g=0$.\vspace*{2ex}
	
	\circled{6} This is analogous to the previous case.\vspace*{1ex}
	
	\begin{remark}\mbox{}\\
		From now on, unless otherwise stated, we consider Neumann boundary conditions i.e.~$\overline{\ell}=0$ and $\phi\in\F$ is arbitrary over the boundary. Of course, it is easy to add a Robin term at the boundary or fix the values at the boundary to obtain Dirichlet conditions.
	\end{remark}

	\subsection{Some curious examples}\label{subsection: curious examples}
	\subsubsection*{Scalar field derived from Lagrange multipliers}
	It is interesting to obtain the scalar field equation, with $V=0$ for simplicity, using Lagrange multipliers. For that we consider $\F=\Omega^0(M)\times\Omega^0(M)$.
	\[\definicion
	\tikzmarkin{definitionL3prima}(0.25,-0.2)(-0.25,0.45)
	\tikzmarkin{definitionL3}(0.2,-0.2)(-0.2,0.45)
	L_1(\phi,\lambda)=-\lambda\square_g\phi\,\vol_g\qquad\qquad \overline{\ell}_1(\phi,\lambda)=-\overline{\lambda}\jmath^*(\nabla_{\!\vec{\nu}}\phi)\vol_{\overline{g}}
	\tikzmarkend{definitionL3}
	\tikzmarkend{definitionL3prima}\]
	
	When we vary in $\lambda$, we obtain directly the Klein-Gordon equation with Neumann boundary conditions. However, we also have to vary in $\phi$ and for that, it is better to take $Y=\iota_{\lambda\vec{\nabla}\phi}\vol_g$ in \eqref{eq: L2=L1}, to obtain the equivalent Lagrangians $L_2(\phi,\lambda)=g^{-1}(\d\phi,\d\lambda)\vol_g$ and $\overline{\ell}_2=0$.
	
	\begin{flalign*}
	&\pushleft{\circled{1}}\MoveEqLeft[1]&\dd_{(\phi,\lambda)} L_2=-\Big(\dd\phi\square_g\lambda+\dd\lambda\square_g\phi\Big)\vol_g+\d\iota_{\vec{W}}\vol_g\qquad\longrightarrow\qquad\Theta^L(\phi,\lambda)=\iota_{\vec{W}}\vol_g\MoveEqLeft[1]
	\end{flalign*}
	where $\vec{W}:=\dd\phi\vec{\nabla}\lambda+\dd\lambda\vec{\nabla}\phi$.
	\begin{flalign*}
	&\pushleft{\circled{2}}\MoveEqLeft[1]&\dd_{(\phi,\lambda)}\overline{\ell}_2-\jmath^*\Theta^L(\phi,\lambda)=-\left(\dd\overline{\phi}\jmath^*\nabla_{\!\vec{\nu}}\lambda+\dd\overline{\lambda}\jmath^*\nabla_{\!\vec{\nu}}\phi\right)\vol_{\overline{g}}\qquad\longrightarrow\qquad\overline{\theta}^{(L,\overline{\ell})}(\phi,\lambda)=0\MoveEqLeft[1]
	\end{flalign*}
	\circled{3} From the previous computations we obtain
	\[
	\begin{array}{|lc|l}
	E_1(\phi,\lambda)=-\square_g\lambda\,\vol_g &\quad & \overline{b}_1(\phi,\lambda)=-\jmath^*(\nabla_{\!\vec{\nu}}\lambda)\vol_{\overline{g}}\\
	E_2(\phi,\lambda)=-\square_g\phi\,\vol_g &\ & \overline{b}_2(\phi,\lambda)=-\jmath^*(\nabla_{\!\vec{\nu}}\phi)\vol_{\overline{g}}
	\end{array}\]
	\[\Sol\left(\SS\right)=\left\{\phi\in\F\ / \begin{array}{l}E_1(\phi,\lambda)=0,\\ E_2(\phi,\lambda)=0,\end{array}\ \ \begin{array}{l}\overline{b}_1(\phi,\lambda)=0\\\overline{b}_2(\phi,\lambda)=0\end{array}\!\right\}\]
	So this theory actually describes two uncoupled scalar fields with Neumann boundary conditions.\vspace*{2ex}
	
	\circled{4} The symplectic structure is given by
	\begin{align*}
	(\OOmega_\SS^\imath)_{(\phi,\lambda)}=\int_\Sigma \Big(\dd\varphi \wwedge \imath^*\!\left(\L_{\vec{n}}\dd\mu\right)+\dd\mu\wwedge \imath^*\!\left(\L_{\vec{n}}\dd\phi\right)\Big)\vol_\gamma
	\end{align*}
	where $\mu:=\imath^*\lambda$, $\varphi:=\imath^*\phi$, and $\vec{n}$ is the $g$-normal vector field to $\imath(\Sigma)\subset M$\!.\vspace*{2ex}
	
	\circled{5} Once again, $\XX_\xi$ is a symmetry if and only if $\L_\xi g=0$.\vspace*{2ex}
	
	\circled{6} This is analogous to the previous cases.

	\begin{remark}\label{remark: Theta no descomsable}\mbox{}\\
		If we had considered $L_3=L_1$ and $\overline{\ell}_3=0$, we would have obtained the same equations because
		\[\dd_{(\phi,\lambda)}L_3=-\left(\dd\lambda\square_g\phi+\dd\phi\square_g\lambda\right)\vol_g+\d\iota_{\vec{U}}\vol_g\]
		where $\vec{U}=\dd\phi\vec{\nabla}\lambda-\lambda\vec{\nabla}\dd\phi$. However, now the covariant derivative of $\dd\phi$ appears in $\Theta^{L_3}$. Thus
		\[\dd_{(\phi,\lambda)}\overline{\ell}-\jmath^*\Theta^{L_3}(\phi,\lambda)=\Big(\dd\overline{\phi}\jmath^*(\nabla_{\!\vec{\nu}}\lambda)-\overline{\lambda}\jmath^*(\nabla_{\!\vec{\nu}}\,\dd\phi)\Big)\vol_{\overline{g}}\]
		As $\vec{\nu}$ is not tangent to the boundary, we cannot ``integrate by parts''. Thus, this expression is not decomposable (see \eqref{eq: dl+Theta}) unless we require both fields to have Dirichlet boundary conditions or Neumann boundary conditions (they have to be included in the definition of $\F$).
	\end{remark}

	\subsubsection*{Scalar field with no equation of motion}\label{example: scalar field no equation}\label{example: constrain systems}
	It is interesting to realize that some conditions over $\phi$ can be included \emph{a priori} in the definition $\F$ or can arise \emph{a posteriori} from the variations of $\SS$. Furthermore, the way to introduce them might change completely the symplectic form obtained by the CPS algorithm. For simple conditions such as Robin or Dirichlet boundary conditions, it is clear how to do that, but for more complicated ones it is not. To illustrate this problem, let us consider a naive example where this phenomenon appears in the bulk.\vspace*{2ex}
	
	Consider $(M,g)$ with no boundary and $\F=\{\phi\in\Omega^0(M)\ /\ \square_g\phi=0\}$. Then we define the Lagrangians
	\[\definicion
	\tikzmarkin{definitionLfunnyprima}(0.25,-0.33)(-0.25,0.57)
	\tikzmarkin{definitionLfunny}(0.2,-0.33)(-0.2,0.57)
	L_1(\phi)=\frac{1}{2}g^{-1}(\d\phi,\d\phi)\qquad\qquad\qquad\qquad L_2(\phi)=0
	\tikzmarkend{definitionLfunny}
	\tikzmarkend{definitionLfunnyprima}
	\]
	
	If we compute the variations we obtain, using the computations of the previous examples, that
	\[\begin{array}{l}
	E_{(1)}(\phi)=-\square_g\phi\,\vol_g\\[1ex]
	\Theta^{L_1}(\phi)=\iota_{\vec{V}}\vol_g
	\end{array}\qquad\qquad\qquad\qquad
	\begin{array}{l}E_{(2)}(\phi)=0\\[1ex]
	\Theta^{L_2}(\phi)=0\end{array}\]
	Notice that both $E_{(1)}$ and $E_{(2)}$ are satisfied by every $\phi\in\F$, so we have that
	\[\Sol(\SS_1)=\F=\Sol(\SS_2)\]
	However
	\[(\OOmega_{\SS_1}\!)_\phi=\int_\Sigma \dd\varphi \wwedge \imath^*\!\left(\L_{\vec{n}}\dd\phi\right)\vol_\gamma\qquad\qquad\qquad\qquad (\OOmega_{\SS_2})_\phi=0\]
	These trivial examples show several things. 
	\begin{itemize}
		\item These two Lagrangians are \textbf{not} equal up to an exact form and have different equations of motion. However, they define the same space of solutions.
		\item The symplectic structure is not always the pullback of the canonical symplectic structure (they define the same theory but have different (pre)symplectic structure).
		\item The symplectic structure $\OOmega_\SS$ depends strongly on the particular form of $\SS$ and not only on $\Sol(\SS)$. \item This example is more useful in the context of boundary conditions: imposing them on $\F$ might spoil the final structures and the use we can make of them. So, in general, it is better to obtain them \emph{a posteriori} from the variations of $\SS$ than include them \emph{a priori} in the definition of $\F$.
		\item One of the problems with this example is that $\F$ is not an open set in $\Omega^0(M)$ and the computations of the variations are more delicate than what we have shown here. The proper way to deal with it is jet bundle framework (see appendix \ref{Appendix: jets} and \cite[11.7 (pag 121)]{vinogradov1984c}).
	\end{itemize}

	\subsection{Chern-Simons}
	Let $\F=\Omega^1(M)$ with $\mathrm{dim}(M)=2k+1$. If $A\in\F$, we denote $\overline{A}:=\jmath^*A\in\Omega^1(\partial M)$ and $a:=\imath^*A\in\Omega^1(\Sigma)$. Now, we define
	\[\definicion
	\tikzmarkin{definitionLCSprima}(0.25,-0.38)(-0.25,0.62)
	\tikzmarkin{definitionLCS}(0.2,-0.38)(-0.2,0.62)
	L(A)=-\frac{1}{k+1}A\wedge(\d A)^k\qquad\qquad\overline{\ell}(A)=0
	\tikzmarkend{definitionLCS}
	\tikzmarkend{definitionLCSprima}
	\]
	where $\alpha^k=\alpha\wedge\overset{(k)}{\cdots}\wedge\alpha$. 
	\begin{flalign*}
	&\pushleft{\circled{1}}\MoveEqLeft[1]\dd_{A} L&
	=-(\d A)^k\wedge\dd A+\frac{k}{k+1}\d\Big(A\wedge(\d A)^{k-1}\wedge\dd A\Big)\quad\longrightarrow\quad\Theta^L(A)=\frac{k}{k+1}A\wedge(\d A)^{k-1}\wedge\dd A\
	\end{flalign*}
	\begin{flalign*}
	&\pushleft{\circled{2}}\MoveEqLeft[1]&\dd_A\overline{\ell}-\jmath^*\Theta^L(A)=-\frac{k}{k+1}\overline{A}\wedge(\d \overline{A})^{k-1}\wedge\dd \overline{A}\qquad\longrightarrow\qquad\overline{\theta}^{(L,\overline{\ell})}(A)=0\MoveEqLeft[1]
	\end{flalign*}
	\begin{flalign*}
	&\pushleft{\circled{3}}\MoveEqLeft[1]&
	\begin{array}{|l}\rule{0ex}{3ex}
	E(A)=-(\d A)^k\\[1ex]
	\overline{b}(A)=-\dfrac{k}{k+1}\overline{A}\wedge(\d\overline{A})^{k-1}
	\end{array}\qquad\qquad\Sol\left(\SS\right)=\Big\{A\in\F \ / \ E(A)=0,\quad \overline{b}(A)=0 \Big\}
	\end{flalign*}
	\begin{flalign*}
	&\pushleft{\circled{4}}\MoveEqLeft[1]
	(\OOmega_\SS^\imath)_A&
	\updown{\eqref{eq: Leibniz}}{\eqref{eq: d dd=dd d}\eqref{eq: stokes pair}}{=}\frac{k}{2}\int_\Sigma \Big((\d a)^{k-1}\wwedge\dd a\wwedge \dd a\Big)-\frac{k(k-1)}{2(k+1)}\int_{\partial\Sigma}\Big(\overline{a}\wwedge (\d\overline{a})^{k-2}\wwedge\dd\overline{a}\wwedge\dd\overline{a} \Big)\MoveEqLeft[1]
	\end{flalign*}
	
	\circled{5} In this example we have no background object, so it is clear that $\LL_{\XX_\xi}L=\L_\xi L$ and the theory is invariant under diffeomorphisms. Thus, $\XX_\xi\in\Sym_{\d}(\SS)$ for every $\xi\in\mathfrak{X}(M)$. Moreover
	\begin{align*}
	&(\ii_{\XX_\xi}\OOmega_\SS^\imath)_A=k\int_\Sigma \imath^*\Big((\d A)^{k-1}\wedge\L_\xi A\wedge \dd A\Big)-\frac{k(k-1)}{k+1}\int_{\partial\Sigma}\overline{\imath}^*\Big(\overline{A}\wedge (\d\overline{A})^{k-2}\wedge\L_{\overline{\xi}}\overline{A}\wedge\dd\overline{A} \Big)\updown{\eqref{eq: Cartan}}{\eqref{eq: Leibniz}\eqref{eq: Leibniz imath}}{=}\\
	&=k\int_\Sigma \imath^*\Big(\d\{(\d A)^{k-1}\wedge\iota_\xi A\wedge \dd A\}-\frac{1}{k}\dd(\d A)^k\wedge\iota_\xi A+\frac{1}{k}\iota_\xi(\d A)^k\wedge \dd A\Big){}-{}\\
	&\phantom{=}-\frac{k}{k+1}\int_{\partial\Sigma}\overline{\imath}^*\Big((k-1)(\d\overline{A})^{k-1}\wedge\iota_{\overline{\xi}}\overline{A}\wedge\dd\overline{A}-\overline{A}\wedge\dd(\d \overline{A})^{k-1}\wedge\iota_{\overline{\xi}}\overline{A} +\overline{A}\wedge \iota_{\overline{\xi}}(\d\overline{A})^{k-1}\wedge\dd\overline{A} \Big)\updown{\mathrm{\eqref{eq: stokes pair}}}{\eqref{eq: diagrama}}{=}\\
	&=\int_\Sigma\imath^*\Big((\iota_\xi E)\wedge\dd A-(\iota_\xi A)\dd E\Big)-\int_{\partial\Sigma}\overline{\imath}^*\Big((\iota_{\overline{\xi}} \overline{b})\wedge\dd \overline{A}+(\iota_{\overline{\xi}} \overline{A})\dd \overline{b}\Big)
	\end{align*}
	so $\overline{\XX}_\xi:=\XX_\xi|_{\Sol(\SS)}$ is a gauge vector field. Its associated $\xi$-charge is given by
	\begin{align*}
	\QQ^{(L,\overline{\ell}),\imath}_\xi(A)&=\frac{1}{k+1}\int_\Sigma \imath^*\Big((\iota_\xi A)\wedge(\d A)^k-kA\wedge (\iota_\xi\d A)\wedge (\d A)^{k-1}+k A\wedge(\d A)^{k-1}\wedge \L_\xi A\Big)\updown{\eqref{eq: Cartan}}{\eqref{eq: Leibniz}}{=}\\
	&=\frac{1}{k+1}\int_\Sigma\imath^*\Big((\iota_\xi A)\wedge(\d A)^k-k\d(A\wedge \iota_\xi A\wedge (\d A)^{k-1})+k(\iota_\xi A)\wedge (\d A)^k\Big)\overset{\eqref{eq: diagrama}}{=}\\
	&=\int_\Sigma\imath^*(\iota_\xi A)\imath^*E+\int_{\partial\Sigma}\overline{\imath}^*(\iota_{\overline{\xi}}\overline{A})\overline{\imath}^*\overline{b}
	\end{align*}
	Of course, $\dd\peqsub{\jj}{\SS}^*\QQ^{(L,\overline{\ell}),\imath}_\xi\in\OOmega^1(\Sol(\SS))$ is zero, giving an alternative proof that $\overline{\XX}_\xi\in\mathrm{Gauge}(\SS)$. Consider now the vector field $(\XX_{(\lambda)})_A=\d \lambda\in T_A\F\cong \F$ for some $\lambda\in\Omega^0(M)$. Notice that it is a constant vector field over $\F$, as it does not depend on $A$. Its interior product with the symplectic form gives
	\[(\ii_{\XX_{(\lambda)}}\OOmega_\SS^\imath)_A=\int_\Sigma\imath^*(\lambda\wedge\dd E)+\int_{\partial\Sigma}(\jmath\smallcirc\overline{\imath})^*\lambda\wedge\left(\dd \overline{\imath}^*\overline{b}+\frac{k}{k+1}(\d \overline{a})^{k-1}\wedge\dd \overline{a}\right)\]
	If we consider Dirichlet boundary conditions (then $\dd\overline{a}=0$) or we assume that $\jmath^*\lambda=0$, then the last term vanishes. However, in general it does not, even when restricted to $\Sol(\SS)$. This proves that the inclusion of a boundary can spoil a gauge freedom.\vspace*{2ex}
	
	\circled{6} Let us consider the space-time decomposition $A\leftrightarrow(A_\perp,A^\top)$ for some fixed foliation and metric $g$. Those objects are related by $A=n\wedge A_\perp+A^\top$. It is not hard to prove \cite[ch.6]{tesis} that
	\[\d A=n\wedge (\d A)_\perp+(\d A)^\top\qquad\qquad\qquad\begin{array}{|l}\!(\d A)_\perp=\varepsilon\dfrac{\L_{\partial_t}A^\top-\L_{\vec{N}}A^\top}{\mathbf{N}}-\dfrac{\d^\top(\mathbf{N}A_\perp)}{\mathbf{N}}\\[2ex]\!(\d A)^\top=\d^\top\! A^\top\phantom{\dfrac{1}{2}}\end{array}\]
	where $\d^\top:=\d-\varepsilon n\wedge\iota_{\vec{n}}$, $\varepsilon=\iota_{\vec{n}}n=-1$, $\mathbf{N}$ is the lapse, and $\vec{N}$ the shift. Thus
	\[\frac{1}{k+1}A\wedge(\d A)^k=\frac{1}{k+1}n\wedge\Big(A_\perp\wedge\d^\top\!A^\top-k A^\top(\d A)_\perp\Big)\wedge(\d^\top\!A^\top)^{k-1}\]
	Identifying $\L_{\partial_t}A^\top$ as the velocity and using the fact $n=\varepsilon \mathbf{N}(\d t)$, it follows \cite{tesis} that the momenta are given by
	\begin{equation}\label{eq: p CS}
	P_\perp=0\qquad\qquad P^\top=\frac{k}{k+1}\imath^*\Big(A^\top\wedge(\d^\top\!A^\top)^{k-1}\Big)
	\end{equation}
	The canonical symplectic form of the first constraint manifold (where $P_\perp=0$) is
	\[\Omega_{(A_\perp,A^\top,P^\top)}=\int_\Sigma\dd A^\top\wwedge\dd P^\top \,\vol_\gamma\]
	while over $\Sol(\SS)$ we have, from the first line of the computation of \circled{4}, that
	\[(\OOmega_\SS)_A=\int_\Sigma \dd a\wwedge\dd\left(\frac{k}{k+1}a\wedge(\d a)^{k-1}\right)\]
	Thanks to equation \eqref{eq: p CS}, we see that $(\mathcal{P}^*\Omega)_{A}=(\OOmega_\SS)_{\mathcal{P}(A)}$ where $\mathcal{P}(A)=(A_\perp,A^\top,P^\top)$.
	
	\subsection{Yang-Mills}
	We consider a Lie algebra $\mathfrak{g}$ with a scalar product invariant under the adjoint representation (every semisimple or abelian algebra satisfies this condition) and define $\F=\Omega^1(M,\mathfrak{g})$. We can expand every  $A\in\F$ over some coordinate patch $\{x^\mu\}$ as
	\[A=A_\mu\d x^\mu \qquad\qquad A_\mu\in\mathfrak{g}\]
	Thanks to the invariant scalar product of $\mathfrak{g}$, we can identify $\mathfrak{g}^*$ with $\mathfrak{g}$. In particular, it allows us to define the trace $\mathrm{Tr}_{\mathfrak{g}}$ of two elements of $\mathfrak{g}$. If we use $\{I,J,\dots\}$ as abstract indices for the algebra, we have
	\[\mathrm{Tr}_{\mathfrak{g}}(A\times B)=A_I B^I\]
	We also have the Lie bracket $[\,,]$ of $\mathfrak{g}$ which induces the Lie bracket $[\,\wedge\,]$ of elements of $\F$ given by
	\[[A\wedge B]=[A_\mu,B_\nu]\d x^\mu \wedge  \d x^\nu\]
	meaning that over the algebra we take the Lie bracket and over forms the wedge product. This is easily generalized to forms of any degree (not necessarily the same) with values in $\mathfrak{g}$. In abstract index notation on $\mathfrak{g}$, we see that $[\,\wedge\,]$ is equivalent to defining some $f_{IJ}^{\ \ K}$ such that
	\[[A\wedge B]^K=f_{IJ}^{\ \ K}A^I\wedge B^J\]
	The following properties are a direct consequence of the combined properties of $[\,,]$ and $\wedge$ 
	\begin{align}
	&[\alpha\wedge\beta]=-(-1)^{|\alpha||\beta|}[\beta\wedge\alpha]\\
	&(-1)^{|\alpha||\gamma|}\big[\alpha\wedge[\beta\wedge\gamma]\big]+(-1)^{|\gamma||\beta|}[\gamma\wedge[\alpha\wedge\beta]\big]+(-1)^{|\alpha||\beta|}\big[\beta\wedge[\gamma\wedge\alpha]\big]=0\\
	&\mathrm{Tr}_{\mathfrak{g}}\Big(\alpha\wedge[\beta\wedge\gamma]\Big)=\mathrm{Tr}_{\mathfrak{g}}\Big([\alpha\wedge\beta]\wedge\gamma]\Big)\label{eq: Tr( [ wedge ])}
	\end{align}
	Finally, given $A\in\F$, we define its covariant derivative and curvature
	\[\mathcal{D}\alpha:=\d\alpha+[A\wedge\alpha]\qquad\qquad F:=\d A+\frac{1}{2}[A\wedge A]\]
	Notice that when computing the variations, the dependence in $A$ of $\mathcal{D}$ and $F$ has to be taken into account. For instance, it is easy to prove using the aforementioned properties that
	\begin{equation}\label{eq: dd F=D ddA}
	\dd F=\mathcal{D}\dd A
	\end{equation}
	It will also be useful to have the following properties which are a direct consequence of the previous definitions, equation \eqref{eq: Tr( [ wedge ])}, and the Leibniz rule
	\begin{align}
	&\mathrm{Tr}_{\mathfrak{g}}\Big(\mathcal{D}\alpha\wedge\beta\Big)=(-1)^{|\alpha|+1}\mathrm{Tr}_{\mathfrak{g}}\Big(\alpha\wedge\mathcal{D}\beta\Big)+\d\mathrm{Tr}_{\mathfrak{g}}\Big(\alpha\wedge\beta\Big)\label{eq: Leibniz D}\\
	&\mathcal{D}[\alpha\wedge\beta]=[\mathcal{D}\alpha\wedge\beta]+(-1)^{|\alpha|}[\alpha\wedge\mathcal{D}\beta]\\
	&\mathcal{D}^2\alpha=[F\wedge\alpha]\label{eq: D^2=[F}\\
	&\mathcal{D}F=0
	\end{align}
	With all these tools at hand, let us study the Yang-Mills theory given by the Lagrangian
	\[\definicion
	\tikzmarkin{definitionLYMprima}(0.25,-0.38)(-0.25,0.62)
	\tikzmarkin{definitionLYM}(0.2,-0.38)(-0.2,0.62)
	L(A)=-\frac{1}{2}\mathrm{Tr}_{\mathfrak{g}}\Big(F\wedge\star_gF\Big)\qquad\qquad\qquad\overline{\ell}(A)=0
	\tikzmarkend{definitionLYM}
	\tikzmarkend{definitionLYMprima}
	\]
	where $\star_g:\Omega^k(M)\to\Omega^{n-k}(M)$ is the Hodge star operator.
	
	\begin{flalign*}
	&\pushleft{\circled{1}}\MoveEqLeft[1]\dd_A L&\overset{\eqref{eq: dd F=D ddA}}{=}-\mathrm{Tr}_{\mathfrak{g}}\Big(\mathcal{D}\dd A\wedge\star_g F\Big)\overset{\eqref{eq: Leibniz D}}{=}\\
	&\MoveEqLeft[1]&\ \ =-\mathrm{Tr}_{\mathfrak{g}}\Big(\dd A\wedge\mathcal{D}\star_gF+\d(\dd A\wedge\star_gF)\Big)\qquad\longrightarrow\qquad\Theta^L(A)=-\mathrm{Tr}_{\mathfrak{g}}\Big(\dd A\wedge\star_gF\Big)\MoveEqLeft[1]
	\end{flalign*}
	\begin{flalign*}
	&\pushleft{\circled{2}}\MoveEqLeft[1]&\dd_A\overline{\ell}-\jmath^*\Theta^L(A)=\mathrm{Tr}_{\mathfrak{g}}\Big(\dd \overline{A}\wedge\jmath^*(\star_gF)\Big)\qquad\longrightarrow\qquad\overline{\theta}^{(L,\overline{\ell})}(A)=0\MoveEqLeft[1]
	\end{flalign*}
	\begin{flalign*}
	&\pushleft{\circled{3}}\MoveEqLeft[1]&
	\begin{array}{|l}E(A)=-\mathcal{D}\star_gF\\[1ex]\overline{b}(A)=\jmath^*(\star_gF)
	\end{array}\qquad\qquad\Sol\left(\SS\right)=\Big\{A\in\F \ / \ E(A)=0,\quad \overline{b}(A)=0 \Big\}\MoveEqLeft[1]
	\end{flalign*}
	\begin{flalign*}
	&\pushleft{\circled{4}}\MoveEqLeft[1]&
	(\OOmega_\SS^\imath)_A\updown{\eqref{eq: Leibniz}}{\eqref{eq: dd F=D ddA}}{=}\int_\Sigma\Big(\dd a\wwedge \imath^*(\star_g\mathcal{D}\dd A)\Big)\MoveEqLeft[1]
	\end{flalign*}

	\circled{5} As $g$ as the sole background object, we can prove that $\XX_\xi\in\Sym_{\d}(\SS)$ if $\xi$ is a $g$-Killing vector field. Using that $\L_\xi A=\iota_\xi F+\mathcal{D}\iota_\xi A$ and that, by definition, $\alpha\wedge\star_g\beta=\langle\alpha,\beta\rangle_g\vol_g$, we obtain
	\begin{align*}
	&\QQ^{(L,\overline{\ell}),\imath}_\xi(A)=\int_\Sigma\imath^*\!\left(-(\iota_\xi F)\wedge\star_g F-(\mathcal{D}\iota_\xi A)\wedge\star_g F-\frac{1}{2}\iota_\xi(F\wedge\star_g F)\right)=\\
	&=\int_\Sigma\imath^*\!\left(-\varepsilon\iota_{\vec{n}}\Big(n\wedge(\iota_\xi F)\wedge\star_g F\Big)+(\iota_\xi A)\wedge(\mathcal{D}\star_g F)-\frac{1}{2}\langle F,F\rangle_{\!g}\iota_\xi\vol_g\right)-\int_{\partial\Sigma}\overline{\jmath}^*\imath^*(\iota_\xi A\wedge\star_g F)\overset{\eqref{eq: diagrama}}{=}\\
	&=\int_\Sigma\imath^*\!\left(-\langle n\wedge\iota_\xi F, F\rangle_{\!g} \varepsilon\iota_{\vec{n}}\vol_g-\frac{\varepsilon}{2}\langle F,F\rangle_{\!g} n_\alpha \xi^\alpha \vol_\gamma\right)+\int_\Sigma\imath^*\!\Big((\iota_\xi A)\wedge E\Big)-\int_{\partial\Sigma}\overline{\imath}^*\jmath^*(\iota_{\xi} A\wedge \star_g F)\overset{\eqref{eq: pullback interior product}}{=}\\
	&=\int_\Sigma n_\alpha \xi_\alpha \left(-F^{\beta}_{\ \gamma}F^{\alpha\gamma}-\frac{\varepsilon}{4}g^{\alpha\beta}F^{\gamma\delta}F_{\gamma\delta}\right)\vol_\gamma+\int_\Sigma\imath^*\!\Big((\iota_\xi A)\wedge E\Big)+\int_{\partial\Sigma}\overline{\imath}^*(\iota_{\overline{\xi}} \overline{A}\wedge\overline{b})
	\end{align*}
	The first term in parenthesis is twice the energy-momentum tensor $T^{\alpha\beta}$ given by $\widetilde{\dd}L=(T^{\alpha\beta}\vol_g)\widetilde{\dd} g_{\alpha\beta}$.\vspace*{2ex}
	
	For a fixed $\lambda\in\Omega^0(M)$, we define the vector field $\XX_{(\lambda)}\in\campos(\F)$ given by $\dd A(\XX_{(\lambda)})=\mathcal{D}\lambda$. Unlike in the Chern-Simons case, this vector field is not constant as $\mathcal{D}$ depends on $A$. It is not hard to prove that
	\[(\ii_{\XX_{(\lambda)}}\OOmega_\SS^\imath)_A\updown{\eqref{eq: Leibniz D}\eqref{eq: D^2=[F}}{\eqref{eq: Tr( [ wedge ])}\eqref{eq: diagrama}}{=}\int_\Sigma\imath^*\mathrm{Tr}_{\mathfrak{g}}\Big(\lambda\dd E\Big)-\int_{\partial\Sigma}\overline{\imath}^*\mathrm{Tr}_{\mathfrak{g}}\Big((\jmath^*\lambda)\dd\overline{b}\Big)\]
	We see that $\XX_{(\lambda)}|_{\Sol(\SS)}\in\mathrm{Gauge}(\SS)$ for every $\lambda\in\Omega^0(M)$.\vspace*{2ex}
	
	\circled{6} It is convenient to rewrite now the symplectic structure. Using \eqref{eq: dd F=D ddA}, $\iota_{\vec{n}}F=\L_{\vec{n}}A-\mathcal{D}\iota_{\vec{n}}A$, and $g^{\alpha\beta}=\varepsilon n^\alpha n^\beta+\imath^\alpha_a\imath^\beta_b\gamma^{ab}$ (space-time decomposition for $\imath(\Sigma)\subset M$), we obtain
	\begin{align*}
	(\OOmega_\SS^\imath)_A&=\int_\Sigma \varepsilon\imath^*\iota_{\vec{n}}\Big(n\wedge\dd A\wwedge\!\star_g\dd F\Big)=\int_\Sigma \varepsilon\imath^*\Big(\langle n\wedge\dd A,\dd F\rangle_{\!g}\iota_{\vec{n}}\vol_g\Big)\overset{\eqref{eq: orientation sigma}}{=}\\
	&=\varepsilon\int_\Sigma \Big\langle\dd a,\dd\imath^*\iota_{\vec{n}}F\!\Big\rangle_{\!\!\gamma}\vol_\gamma\overset{(\varepsilon=-1)}{=}\\
	&=\int_\Sigma \Big\langle\dd a,\dd\imath^*\Big(\mathcal{D}\iota_{\vec{n}}A-\L_{\vec{n}} A\Big)\!\Big\rangle_{\!\!\gamma}\vol_\gamma
	\end{align*}
	Let us now prove that if we break $A\leftrightarrow(A_\perp,A^\top)$, then the parenthesis in the last line of the previous computation is precisely the momentum $p$ associated with $a=\imath^*A^\top$ ($p_\perp$ turns out to be zero). First, we can adapt the computations from \cite[ch.4]{tesis} to the non-Abelian case to prove that
	\[F=n\wedge F_\perp+F^\top\qquad\qquad\qquad\begin{array}{|l}\!F_\perp=\varepsilon\dfrac{\L_{\partial_t}A^\top-\L_{\vec{N}}A^\top}{\mathbf{N}}-\dfrac{\mathcal{D}^\top(\mathbf{N}A_\perp)}{\mathbf{N}}\\[2ex]\!F^\top=\d^\top\! A^\top+\dfrac{1}{2}[A^\top\wedge A^\top]\end{array}\]
	where $\mathcal{D}^\top$ is the covariant derivative associated to $A^\top$ and $\d^\top:=\d-\varepsilon n\wedge\iota_{\vec{n}}$. Thus
	\[\mathcal{D}^\top\!\alpha=\mathcal{D}\alpha-n\wedge\Big(\iota_{\vec{n}}\d\alpha+[A_\perp\wedge\alpha]\Big)\]
	Proceeding as in the previous example (see also \cite{tesis}), it can be checked that $F_\perp$ is precisely the momentum (up to the corresponding pullback). Finally, using
	\[\vec{n}=\frac{\partial_t-\vec{N}}{\mathbf{N}}\qquad\qquad\qquad n=\varepsilon \mathbf{N}(\d t)\]
	we can compute
	\begin{align*}
	&\mathcal{D}\iota_{\vec{n}}A-\L_{\vec{n}}A=\varepsilon\d A_\perp+\varepsilon[A\wedge A_\perp]-\frac{1}{\mathbf{N}}(\L_{\partial_t}-\L_{\vec{N}})\Big(n\wedge A_\perp+A^\top\Big)-\d(1/\mathbf{N})\wedge \mathbf{N}\iota_{\vec{n}}A=\\
	&\qquad=\varepsilon\frac{\mathbf{N}\d A_\perp+[A\wedge(\mathbf{N}A_\perp)]}{\mathbf{N}}-\frac{1}{\mathbf{N}}\L_{\partial_t-\vec{N}}(\varepsilon\d t\wedge \mathbf{N} A_\perp)-\frac{\L_{\partial_t}A^\top-\L_{\vec{N}}A^\top}{\mathbf{N}}+\varepsilon\frac{\d\mathbf{N}}{\mathbf{N}}\wedge A_\perp=\\
	&\qquad=\varepsilon\frac{\d(\mathbf{N}A_\perp)+[A\wedge(\mathbf{N}A_\perp)]}{\mathbf{N}}-\frac{\varepsilon}{\mathbf{N}} \d t\wedge\L_{\mathbf{N}\vec{n}}(\mathbf{N}A_\perp)-\varepsilon F_\perp-\varepsilon\frac{\mathcal{D}^\top(\mathbf{N}A_\perp)}{\mathbf{N}}-=\\
	&\qquad= F_\perp+\varepsilon\frac{(\mathcal{D}-\mathcal{D}^\perp)(\mathbf{N}A_\perp)-n\wedge\L_{\vec{n}}(\mathbf{N}A_\perp)}{\mathbf{N}}=\\
	&\qquad= F_\perp+\varepsilon n\wedge[A_\perp\wedge A_\perp]=F_\perp
	\end{align*}
	So indeed the symplectic form corresponds to the canonical one over the first constraint manifold.
	
	\subsection{Parametrized Yang-Mills}\label{example: parametrized YM}
	For our last example we consider the Yang-Mills but allowing the metric to vary in a very specific way, namely, through pullbacks by diffeomorphisms. We consider again a Lie algebra $\mathfrak{g}$ with a scalar product invariant under the adjoint representation, a globally hyperbolic space-time $(M,g)$, the space of fields $\F=\Omega^1(M)\times\mathrm{Diff}(M)$ which is non-linear, and
	\[\definicion
	\tikzmarkin{definitionLPYMprima}(0.25,-0.38)(-0.25,0.62)
	\tikzmarkin{definitionLPYM}(0.2,-0.38)(-0.2,0.62)
	L(A,Z)=-\frac{1}{2}\mathrm{Tr}_{\mathfrak{g}}(F\wedge_{Z^*\!g}F)\qquad\qquad\qquad\overline{\ell}(A,Z)=0
	\tikzmarkend{definitionLPYM}
	\tikzmarkend{definitionLPYMprima}
	\]
	As $L(Y^*A,Z\circ Y)=Y^*L(A,Z)$, from section \ref{section: diff invariant} we obtain that $L$ is invariant under diffeomorphisms.\vspace*{2ex}
	
	Let us compute the exterior derivative of $\peqsubfino{g}{Z}{-0.4ex}:=Z^*\!g$. For that, we consider some $\VV\in T_{(A,Z)}\F$ defined by a curve $\{(A,Z_\tau)\}_\tau$ of $\F$. In particular, $\VV$ has only component in the diffeormorphism direction, meaning that $\dd A(\VV)=0$ and $\dd Z(\VV)=\VV$, so in the following we forget about the $A$ component. Besides, recall that $\VV\in T_Z\mathrm{Diff}(M)$ defines a vector field over $Z$ i.e.~such that $\VV_{\!p}\in T_{Z(p)}M$\!. The fact that $\VV$ evaluated at $p$ does not belong to $T_pM$ is an inconvenience that can be solved as follows
	\begin{align}\label{eq: dd g(VV)=Lie}
	\begin{split}
	\dd \peqsubfino{g}{Z}{-0.4ex}(\VV)&=\left.\frac{\d}{\d\tau}\right|_{\tau=0}\peqsubfino{g}{{Z_\tau}}{-0.4ex}=\left.\frac{\d}{\d\tau}\right|_{\tau=0}Z^*_\tau g=\\
	&=\left.\frac{\d}{\d\tau}\right|_{\tau=0}(Z^{-1}_0\smallcirc Z_\tau)^*Z_0^*g=\left.\frac{\d}{\d\tau}\right|_{\tau=0}(Z_0^{-1}\smallcirc Z_\tau)^*\peqsubfino{g}{Z}{-0.4ex}
	\end{split}
	\end{align}
	where $Z_0=Z$. Notice that the new curve $\{Z^{-1}_0\smallcirc Z_\tau\}_\tau$ in $\mathrm{Diff}(M)$ passes through the identity with velocity
	\[\WW=\left.\frac{\d}{\d\tau}\right|_{\tau=0}(Z^{-1}_0\smallcirc Z_\tau)=(Z^{-1}_0)_*\left.\frac{\d}{\d\tau}\right|_{\tau=0}Z_\tau=(Z^{-1}_0)_*\VV\]
	Now $\WW_p\in T_p M$ as we wanted. In fact, notice that the last expression of \eqref{eq: dd g(VV)=Lie} is just the Lie derivative
	\[\dd \peqsubfino{g}{Z}{-0.4ex}(\VV)=\L_{\WW}\peqsubfino{g}{Z}{-0.4ex}=\L_{Z^{\scalebox{.6}{-}1}_*\dd Z(\VV)}\peqsubfino{g}{Z}{-0.4ex}\qquad\longrightarrow\qquad\dd \peqsubfino{g}{Z}{-0.4ex}=\peqsub{\L}{\underline{\dd Z}}\,\peqsubfino{g}{Z}{-0.4ex}\]
	Where we have defined $\underline{\dd Z}:=Z^{-1}_*\dd Z$. The expression on the right, where we used $\L_{\underline{\dd Z}}$, is just a useful notation and in order to have full meaning, it has to be evaluated over some $\VV$. Although we have to be careful with this notation, it will shorten the computations. For instance, $\iota_{\underline{\dd Z}}\iota_{\underline{\dd Z}}$ or $[\underline{\dd Z},\underline{\dd Z}]$ are not zero because $\wwedge$ is antisymmetric for elements of degree $(1,0)$ such as $\underline{\dd Z}$. In fact
	\[\iota_{\underline{\dd Z}}\iota_{\underline{\dd Z}}(\VV,\WW)=\iota_{\underline{\dd Z}(\VV)}\iota_{\underline{\dd Z}(\WW)}-\iota_{\underline{\dd Z}(\WW)}\iota_{\underline{\dd Z}(\VV)}=2\iota_{\underline{\dd Z}(\VV)}\iota_{\underline{\dd Z}(\WW)}\qquad [\underline{\dd Z},\underline{\dd Z}](\VV,\WW)=2[\underline{\dd Z}(\VV),\underline{\dd Z}(\WW)]\]
	It is necessary for the following to compute the $(0,2)$-form $\dd(\underline{\dd Z})$. To make it clearer, we are going to use different space-time indices for the domain space and the target space (although they are both the same). Indeed, we consider $Z:(M;\{\alpha,\beta,\dots\})\to(M;a,b,\{\dots\})$ so its pushforward $Z_*$ is denoted $Z_\alpha^a$ (it also denotes its pullback $Z^*$). We use the results of the appendix of \cite{tesis} together with the well known fact that the tangent vectors of $\mathrm{Diff}(M)$ are vector fields of $M$ to obtain
	\begin{align*}
	\dd(\underline{\dd Z})^\beta&=\dd\Big((Z^{-1})^\beta_b\dd Z^b\Big)=\dd(Z^{-1})^\beta_b\wwedge\dd Z^b=-(Z^{-1})^\beta_a(Z^{-1})^\gamma_b(\dd Z)^a_\gamma\wwedge\dd Z^b=\\
	&=-(Z^{-1})^\beta_a(Z^{-1})^\gamma_b Z^a_\alpha\nabla_{\!\gamma}((Z^{-1})^\alpha_c \dd Z^c)\wwedge\dd Z^b=-\delta^\beta_\alpha \nabla_{\!\gamma}((Z^{-1})^\alpha_c \dd Z^c)\wwedge((Z^{-1})^\gamma_b \dd Z^b)=\\
	&=-\nabla_{\!\gamma}\big(\underline{\dd Z}\big)^{\!\beta}\wwedge\big(\underline{\dd Z}\big)^{\!\gamma}=\big(\underline{\dd Z}\big)^{\!\gamma}\wwedge\nabla_{\!\gamma}\big(\underline{\dd Z}\big)^{\!\beta}=:\nabla_{\!\underline{\dd Z}}(\underline{\dd Z})^\beta
	\end{align*}
	Where the last equality is just the definition of the last term. Let us consider $\VV,\WW\in T_{(A,Z)}\F$. We denote $V:=\dd Z(\VV),W:=\dd Z(\WW)\in\mathfrak{X}(M)$, we then have
	\[\dd(\underline{\dd Z})(\VV,\WW)=\nabla_{\!Z^{\scalebox{0.6}{-}1}_*\!V}(Z^{\scalebox{0.6}{-}1}_*W)-\nabla_{\!Z^{\scalebox{.6}{-}1}_*\!W}(Z^{-1}_*W)=[Z^{-1}_*V,Z^{-1}_*W]=Z^{-1}_*[V,W]\]
	so $\dd(\underline{\dd Z})$ is, essentially, the Lie bracket once it is evaluated. It will be very useful to introduce
	\begin{equation}\label{def: DD}
	\DD:=\dd-\L_{\underline{\dd Z}}
	\end{equation}
	that measures the variation of the quantities once we subtracts the variation due to the diffomorphism. In particular, we have $\DD \peqsubfino{g}{Z}{-0.4ex}=0$. The following properties are immediate from the definition of $\DD$ and the properties of $\dd$ and $\L$.
	\begin{align*}
	&\bullet\ \dd \iota_{\underline{\dd Z}}\alpha=\iota\Big(\nabla_{\!\underline{\dd Z}}\underline{\dd Z}\Big)\alpha-\iota_{\underline{\dd Z}}\dd \alpha&&\bullet\ \DD^2=0\\
	&\bullet\ \DD(\alpha\wwedge\beta)=(\DD\alpha)\wwedge\beta+(-1)^{\|\alpha\|}\alpha\wwedge(\DD\beta)&&\bullet\ \DD(\star_{Z^*\!g}\alpha)=\star_{Z^*\!g}(\DD\alpha)\\
	&\bullet\ \DD F=\mathcal{D}\DD A
	\end{align*}
	\begin{flalign*}
	&\pushleft{\circled{1}}\MoveEqLeft[1]\dd_{(A,Z)}L&=\DD_{(A,Z)}L+\L_{\underline{\dd Z}}L\overset{\eqref{eq: Cartan}}{=}-\mathrm{Tr}_{\mathfrak{g}}\Big(\DD F\wedge\star_{Z^*\!g} F\Big)+\d\iota_{\underline{\dd Z}}L\overset{\eqref{eq: Leibniz D}}{=}\MoveEqLeft[1]\\
	&\MoveEqLeft[1]&=-\mathrm{Tr}_{\mathfrak{g}}\Big(\DD A\wedge\big(\mathcal{D}\star_{Z^*\!g} F)\Big)+\d\Big(\iota_{\underline{\dd Z}}L-\mathrm{Tr}_{\mathfrak{g}}\big(\DD A\wedge\star_{Z^*\!g} F\big)\Big)
	\end{flalign*}
	We take $\Theta^L(A,Z)=\iota_{\underline{\dd Z}}L-\mathrm{Tr}_{\mathfrak{g}}\big(\DD A\wedge\star_{Z^*\!g} F\big)$.
	\begin{flalign*}
	&\pushleft{\circled{2}}\MoveEqLeft[1]&\dd_{(A,Z)}\overline{\ell}-\jmath^*\Theta^L(A,Z)=\mathrm{Tr}_{\mathfrak{g}}\Big(\DD \overline{A}\wedge\jmath^*(\star_{Z^*\!g}F)\Big) \MoveEqLeft[1]
	\end{flalign*}
	where the first term of $\Theta^L$ vanishes because $\underline{\dd Z}$ is tangent to the boundary. We take $\overline{\theta}^{(L,\overline{\ell})}=0$.
	\begin{flalign*}
	&\pushleft{\circled{3}}\MoveEqLeft[1]&
	\begin{array}{|l}\!E_1(A,Z)=E_2(A,Z)=-\mathcal{D}\star_{Z^*\!g}F\\[1ex]\!\overline{b}_1(A,Z)=\overline{b}_2(A,Z)=\jmath^*(\star_{Z^*\!g}F)
	\end{array}\qquad\qquad\Sol\left(\SS\right)=\Big\{(A,Z)\in\F \ / \ \begin{array}{l}E_1(A,Z)=0\\\overline{b}_1(A,Z)=0\end{array} \Big\}
	\end{flalign*}
	Notice that the parametrization adds no additional equation or boundary condition.
	\begin{flalign*}
	&\pushleft{\circled{4}}\MoveEqLeft[1]&
	(\OOmega_\SS^\imath)_{(A,Z)}=\int_\Sigma\imath^*\Big(\iota\Big(\nabla_{\!\underline{\dd Z}}\underline{\dd Z}\Big)L-\iota_{\underline{\dd Z}}\dd L-\dd\mathrm{Tr}_{\mathfrak{g}}\big(\DD A\wedge\star_{Z^*\!g} F\big)\Big)=\MoveEqLeft[1]\\
	&\MoveEqLeft[1]&=\int_\Sigma \imath^*\Big(\iota\Big(\nabla_{\!\underline{\dd Z}}\underline{\dd Z}\Big)L-\iota_{\underline{\dd Z}}\mathrm{Tr}_{\mathfrak{g}}\Big(\DD A\wedge E\Big)-\iota_{\underline{\dd Z}}\d\Big(\iota_{\underline{\dd Z}}L-\mathrm{Tr}_{\mathfrak{g}}\big(\DD A\wedge b\big)\Big)-\dd\mathrm{Tr}_{\mathfrak{g}}\big(\DD A\wedge b\big)\Big)\updown{\eqref{eq: Cartan}}{\eqref{def: DD}}{=}\\
	&\MoveEqLeft[1]&=\int_\Sigma \imath^*\Big(\Big\{\iota\Big(\nabla_{\!\underline{\dd Z}}\underline{\dd Z}\Big)-\iota_{\underline{\dd Z}}\L_{\underline{\dd Z}}\Big\}L-\iota_{\underline{\dd Z}}\mathrm{Tr}_{\mathfrak{g}}\Big(\DD A\wedge E\Big)-\d\iota_{\underline{\dd Z}}\mathrm{Tr}_{\mathfrak{g}}(\DD A\wedge b)-\DD\mathrm{Tr}_{\mathfrak{g}}\big(\DD A\wedge b\big)\Big)\updown{\eqref{eq: diagrama}}{\eqref{eq: pullback interior product}}{=}\\
	&\MoveEqLeft[1]&=-\int_\Sigma\imath^*\iota_{\underline{\dd Z}}\mathrm{Tr}_{\mathfrak{g}}\Big(\DD A\wedge E\Big)-\int_{\partial\Sigma}\overline{\imath}^*\iota_{\underline{\overline{\dd Z}}}\mathrm{Tr}_{\mathfrak{g}}\Big(\DD \overline{A}\wedge\overline{b}\Big)+\int_\Sigma\imath^*\mathrm{Tr}_{\mathfrak{g}}\big(\DD A\wedge \star_{Z^*\!g}\mathcal{D}\DD A \big)
	\end{flalign*}
	When we integrate by parts, $\DD$ can be pulled back because $\underline{\dd Z}$ is tangent to the boundary. To get to the last line, we have used that the integral of the term with curly brackets vanishes. Indeed, if we evaluate this term at $\VV,\WW\in T_{(A,Z)}\F$, where we denote $\underline{V}:=\underline{\dd Z}(\VV)$ and $\underline{W}:=\underline{\dd Z}(\WW)$, we obtain
	\begin{align*}
	&\int_\Sigma\imath^*\Big\{\iota\Big(\nabla_{\!\underline{\dd Z}}\underline{\dd Z}\Big)-\iota_{\underline{\dd Z}}\L_{\underline{\dd Z}}\Big\}L(\VV,\WW)=\int_\Sigma\imath^*\big\{\iota_{[\underline{V},\underline{W}]}-\iota_{\underline{V}}\L_{\underline{W}}+\iota_{\underline{W}}\L_{\underline{V}}\big\}L\overset{\eqref{eq: Cartan}}{=}\\
	&=\int_\Sigma\imath^*\big\{[\L_{\underline{V}},\iota_{\underline{W}}]-\iota_{\underline{V}}\d\iota_{\underline{W}}+\iota_{\underline{W}}\L_{\underline{V}}\big\}L\overset{\eqref{eq: Cartan}}{=}\int_\Sigma\imath^*\d\iota_{\underline{V}}\iota_{\underline{W}}L=\int_{\partial\Sigma}\overline{\jmath}^*\imath^*\iota_{\underline{V}}\iota_{\underline{W}} L\updown{\eqref{eq: diagrama}}{\eqref{eq: pullback interior product}}{=}\int_{\partial\Sigma}\overline{\imath}^*\iota_{\overline{\underline{V}}}\iota_{\overline{\underline{W}}}\jmath^*L=0
	\end{align*}

	\circled{5} We have $\dd A(\XX_\xi)=\L_\xi A$. However, $\mathrm{Diff}(M)$ is not linear so we have to define $\dd Z(\XX_\xi)$. For that, we take advantage of the fact that $\underline{\dd Z}(\VV)$ is a vector field over $M$ to define $\underline{\dd Z}(\XX_\xi)=\xi$ or, equivalently, $\dd Z(\XX_\xi)=Z_*\xi$. With this definition, let us check that despite the presence of the background object $g$, we have that $\XX_\xi\in\Sym_{\d}(\SS)$ for every vector field $\xi\in\mathfrak{X}(M)$. First, notice that
	\[\ii_{\XX_\xi}\DD A=\ii_{\XX_\xi}\dd A-\ii_{\XX_\xi}\L_{\underline{\dd Z}}A=\L_\xi A-\L_\xi A=0\]
	Then
	\begin{align*}
	&\LL_{\XX_\xi}L=\ii_{\XX_\xi}\dd L=\d\iota_\xi L&&\longrightarrow&&S^L_{\XX_\xi}=\iota_\xi L(=S^L_\xi)\\
	&\LL_{\XX_\xi}\overline{\ell}=0=\jmath^*\iota_\xi L=\jmath^*S_\xi^L&&\longrightarrow&&\overline{s}^{(L,\overline{\ell})}_{\XX_\xi}=0
	\end{align*}
	
	So $\XX_\xi\in\Sym_{\d}(\SS)$ for every $\xi\in\mathfrak{X}(M)$. We can also check that $\overline{\XX}_\xi:=\XX_\xi|_{\Sol(\SS)}$ is a gauge vector field.
	\begin{align*}
	(\ii_{\XX_\xi}&\OOmega_\SS^\imath)_{(A,Z)}=-\int_{(\Sigma,\partial\Sigma)}\underline{\imath}^*\underline{\iota}_\xi\mathrm{Tr}_{\mathfrak{g}}\left(\DD A\wedge E,\DD \overline{A}\wedge\overline{b}\right)+\int_\Sigma\imath^*\Big\{\iota\Big([\xi,\underline{\dd Z}]\Big)-\iota_\xi\d\iota_{\underline{\dd Z}}+\iota_{\underline{\dd Z}}\d\iota_\xi\Big\}L\overset{\eqref{eq: Cartan}}{=}\\
	&=-\int_{(\Sigma,\partial\Sigma)}\underline{\imath}^*\underline{\iota}_\xi\mathrm{Tr}_{\mathfrak{g}}\left(\DD A\wedge E,\DD \overline{A}\wedge\overline{b}\right)+\int_\Sigma\imath^*\Big\{[\L_\xi,\iota_{\underline{\dd Z}}]-\L_\xi\iota_{\underline{\dd Z}}+\d\iota_\xi\iota_{\underline{\dd Z}}+\iota_{\underline{\dd Z}}\L_\xi\Big\}L=\\
	&=-\int_{(\Sigma,\partial\Sigma)}\underline{\imath}^*\underline{\iota}_\xi\mathrm{Tr}_{\mathfrak{g}}\left(\DD A\wedge E,\DD \overline{A}\wedge\overline{b}\right)+\int_{\partial\Sigma}\overline{\imath}^*\imath^*(\iota_\xi\iota_{\underline{\dd Z}}L)=\\
	&=-\int_\Sigma\imath^*\iota_\xi\mathrm{Tr}_{\mathfrak{g}}\Big(\DD A\wedge E\Big)-\int_{\partial\Sigma}\overline{\imath}^*\iota_{\overline{\xi}}\mathrm{Tr}_{\mathfrak{g}}\Big(\DD \overline{A}\wedge\overline{b}\Big)
	\end{align*}
	In the passage from the first to the second line we have used $\iota_{[V,W]}=[\L_V,\iota_W]$. To get to the last line we have used that $\xi$ and $\underline{\dd Z}$ are tangent to the boundary. Clearly $\overline{\XX}_\xi\in\mathrm{Gauge}(\SS)$. The same computation performed in the previous example shows that $\XX_{(\lambda)}$, given by $\dd A(\XX_{(\lambda)})=\mathcal{D}A$ and $\dd Z(\XX_{(\lambda)})=0$, is a $\underline{\d}$-symmetry.\vspace*{2ex}
	
	\circled{6} The symplectic form is, as one should expect, the same one obtained in the Yang-Mills example replacing $\dd A$ by $\DD A$ (recall that $\DD$ accounts for the variation up to diffeomorphism). Following \cite{tesis,margalef2016hamiltonian,margalef2016hamiltonianb} and the computation of the Yang-Mills example, one obtains again the isomorphism between the CPS symplectic structure and the one induced by the canonical symplectic form of the cotangent Hamiltonian framework.

	\section{Summary and Conclusions}
	In this paper, we have unraveled the geometric nature of the covariant phase space methods in manifolds with boundary and characterized the ambiguities that arise. To do that, we have developed the ``relative bicomplex framework'':
	\begin{itemize}
		\item We have taken the definition of the relative complex $(\Omega^k(M,N),\underline{\d})$ from \cite{bot1982differential} and we have extended it to include all the relevant geometric operations. In this framework, the relative manifold $(M,\partial M)$ has no relative boundary and the usual results for manifolds without boundary apply.
		\item We have taken the definition of the variational bicomplex \cite{anderson1989variational} and extended it to the relative case. This allows us to consider fields over relative pairs $(M,N)$.
	\end{itemize}
	This natural formalism, which also covers the cases with ``corner terms'' \cite{booth2005horizon,hartle1981boundary,sorkin1975time}, helps to clarify several common misconceptions regarding the role of some boundary terms. More specifically, we obtain the (pre)symplectic structure over the space of solutions which, in general, has a boundary contribution. We prove that this construction, in the case of contractible bundles, is intrinsically associated with the action and does not depend on the representative Lagrangians. We also prove that for non-contractible bundles this is not the case and some care is due. In any case, many physical applications are modeled over contractible bundles. For example, the bundle of connections of a Yang-Mills theory is an affine bundle or the bundle of metrics is modeled over a contractible quotient bundle.\vspace*{2ex}
	
	We study the symmetries of the action and find an interesting group of them, called $\underline{\d}$-symmetry, which are always Hamiltonian vector fields. We also define a map $\xi\in\mathfrak{X}(M)\mapsto\XX_\xi\in\campos(\F)$ and consider the associated $\xi$-currents $(J_\xi,\overline{\jmath}_\xi)$, $\xi$-charges $\QQ_\xi$, and $\xi$-potentials $(Q_\xi,\overline{q}_\xi)$. These are candidates to be relevant quantities although its specific physical meaning will depend on the problem at hand. We obtain a flux law for the $\xi$-charges and characterize when $\XX_\xi\in\campos(\F)$ is a symmetry. Finally, we provide the \textbf{CPS-algorithm} (expressed in the more standard non-relative language which is useful for concrete computations) and implement it in some relevant theories. We show that for the prototypical examples, like Yang-Mills or Chern-Simons, the CPS symplectic structure is isomorphic to the canonical one obtained from the cotangent Hamiltonian framework. However, we also provide a naive counterexample in \ref{subsection: curious examples} showing that this is not always the case. This is relevant especially when one considers theories with boundaries where some field-dependent objects are fixed over the boundary in a complicated manner (like isolated or dynamical horizons in general relativity). It remains to study in detail the necessary and sufficient conditions for this equivalence to hold.\vspace*{2ex}
	
	We plan to study in the near future several gravity theories with this powerful formalism. This study could shed some light on the strategies to follow in order to quantize those theories. Moreover, those techniques are very well suited to study some problems that arise in condensed matter theory, where the boundary plays a prominent role. For instance, the covariant phase space formalism provides a suitable framework to study quantum edge states and the appearance of degrees of freedom in the boundary. Another approach we plan to study is the multisymplectic formalism \cite{gotay1998momentum}  of the CPS methods with boundary, which is a natural generalization of the theory that has been developed in this paper.
	
	\section*{Acknowledgments}
	The first author wants to thank A.~Ashtekar, T.~de Lorenzo, N.~Khera, and M.~Schneider for the early discussion about the CPS methods. We want also to thank E.~Miranda and G.~Moreno for their help. Big thanks are due to F.~Barbero, M.~Castrillón, J.C.~Marrero, E.~Padrón, and the anonymous referee for the interesting discussions, the careful reading of this manuscript, and for providing useful references. Thanks also to A.~Martínez García for his help with some new glyphs. This work has been supported by the Spanish Ministerio de Ciencia Innovación y Universidades-Agencia Estatal de Investigación FIS2017-84440-C2-2-P grant. Juan Margalef-Bentabol is supported by the Eberly Research Funds of Penn State, by the NSF grant PHY-1806356, and by the Urania Stott fund of Pittsburgh foundation UN2017-92945.

	\newpage
	\begin{appendices}

		\section{Space of jets}\label{Appendix: jets}
		
		\subsection{Motivation}\label{section: motivation jets}
		
		Let $E\overset{\pi}{\to}M$ be a fibered manifold of rank $r$ and $\mathrm{dim}M=n$, and let $\F$ be its space of (local) sections. In this appendix we consider the non-relative version but the same definitions and results apply for the relative formulation introduced in section \ref{section: geometry Mxpartial M}.
		\begin{enumerate}\setcounter{enumi}{-1}
			\item An element $\phi\in\F$ is a smooth map $\phi:U\subset M\to E$ such that $\phi(p)\in E_p:=\pi^{-1}(p)$. Although unusual, $E$ can be thought of as all possible values that a field of $\F$ can take at $p$ for every $p\in M$\!. We can say that the bundle $\pi:E\to M$ has the information of $\F$ of degree $0$ (no derivatives involved). It will be useful to denote $\pi^0_M:=\pi$.
			\item An element $\VV\in\campos(\F)$ is a map $\VV:\F\to T\F$ such that $\VV_{\!\phi}\in T_\phi \F$. However, it can also be considered as a map $\VV_{\!\phi}:M\to TE$ such that $\VV_{\!\phi}(p)\in T_{\phi(p)}E_p$. To see that, consider a curve $\{\phi_\tau\}_\tau$ in $\F$ with $\phi_0=\phi$ and $\VV_{\!\phi}=\partial_\tau|_0\phi_\tau$. Then $\{\phi_\tau(p)\}_\tau$ is a curve in $E_p$ passing through $\phi(p)$ so $\VV_{\!\phi}(p)=\partial_\tau|_0\phi_\tau(p)\in T_{\phi(p)}E_p$. Thus, $TE$ can be thought of as all possible values that a field $\F$ can take at $p$ together with its possible ``velocities'' (infinitesimal displacements). So, loosely speaking, $\pi^1_M:TE\to M$ has the information of $\F$ up to degree $1$ (first derivatives involved).
			\item We want to generalize this concept of information up to any degree $k\in\N$. For that, we are going to define the $k$-jet bundle $J^k\E$. In analogy to the previous cases, the sections $W_{\!\phi}:M\to J^k\E$ will show that the bundle $\pi^k_M:J^k\E\to M$ gathers the information of $\F$ up to degree $k$ (involving derivatives up to order $k$).
		\end{enumerate}
		The goal of this appendix is to give the most important definitions and results about the spaces of jets without getting too much into the details. For the interested reader, we recommend \cite{anderson1989variational,olver2000applications,saunders1989geometry}.
		
		\subsection{Definition}\label{section: jets definition}
		We declare that two local sections $\phi_1,\phi_2\in\F$ are $\boldsymbol{k}$\textbf{-equivalent} at $p\in M$ if their Taylor expansions in an adapted chart (and hence in every adapted chart) are equal up to order $k$. This equivalence class, denoted $j^k_p(\phi)$, is called $\boldsymbol{k}$\textbf{-jet of $\boldsymbol{\phi}$ at} $\boldsymbol{p}$. Loosely speaking, it represents the coordinate-free expression of the Taylor expansion of $\phi$ at $p$ up to order $k$. Joining all these elements $j^k_p(\phi)$ for every $\phi\in\F$ we obtain the $\boldsymbol{k}$\textbf{-jet fiber} $J^k_pE$. Joining all these fibers $J^k_pE$ for every $p\in M$  we obtain the $\boldsymbol{k}$\textbf{-jet fibered bundle} $\pi^k_M:J^k\E\to M$ (with the differential structure induced by $E$, analogous to the one of $TE$). We denote $\F_k$ the space of its (local) sections i.e. $\phi_k:M\to J^k\E$ such that $\phi_k(p)\in J^k_pE$.\vspace*{2ex}
		
		$j^0_p(\phi)$ is completely determined by $\phi(p)\in E_p$, then $j^0\phi=\phi$, $J^0\E=E$ and $\F_0=\F$. Meanwhile, $j^1_p(\phi)$ is determined by $\phi\in\F$ and its derivatives in the $M$-directions. Analogously, $J^2\E$ gathers all possible values of the fields, $M$-derivatives, and second $M$-derivatives at all the points $p\in M$\!. Thus, $J^k\E$ is indeed a generalization that achieves the goal we set for ourselves in the previous section (see also the coordinate expression \eqref{eq: j infty phi} below).\vspace*{2ex}
		
		Working over some $J^k\E$ is not always enough. Indeed, some constructions using elements of order $k$ might result in objects of higher order. Hence, fixing $k$ \emph{a priori} restricts the allowed constructions. It is tempting to consider the union of all $J^k\E$. Unfortunately, the resulting space is ill-behaved. On the other hand, nothing prevents us from taking $k=\infty$ in our previous construction: two local sections $\phi_1,\phi_2\in\F$ are $\boldsymbol{\infty}$\textbf{-equivalent} at $p\in M$ if their Taylor coefficients at $p$ are the same at all orders. This defines the $\boldsymbol{\infty}$\textbf{-jets} $j^\infty_p(\phi)$ and the $\boldsymbol{\infty}$\textbf{-jet bundle} $\pi^\infty_M:J^\infty\!E\to M$\!. The only problem is that $J^\infty\!E$ is an infinite-dimensional manifold, which complicates everything, but it is still very well behaved (e.g.~it is paracompact and it admits partitions of unity). Indeed, it can be proved \cite{anderson1989variational,saunders1989geometry} that $J^\infty\!E$ is the projective limit of $\{J^k\E\}_k$ and that the notions of smooth functions, vector fields, and other relevant constructions are well defined in $J^\infty\!E$. Moreover, they only depend on some finite order (arbitrary large and it can increase when performing some operations, but finite nonetheless).\vspace*{2ex}
		
		It will be very useful to understand the following: given a local section $\phi\in\F$, we have $j^\infty_p(\phi)\in J^\infty_p E$. Dropping the dependence on $p$, we get a map $j^\infty(\phi):M\to J^\infty\!E$. To understand $j^\infty(\phi)$, we introduce some adapted coordinates $\{x^i,u^a,u^a_J\}$ on $J^\infty\!E$, where $\{x^i\}_{i=1\ldots n}$ is a chart of $M$\!, $\{u^a\}_{a=1\ldots r}$ are coordinates on the fiber $E_p$, and $J$ is a multi-index related to the partial derivatives with respect to $x^J$. One can think that $\{x^i\}$ represent all the base points in $U\subset M$, $\{u^a\}$ all possible values $\phi(p)$ of every $\phi\in\F$, $\{u^a_j\}$ all $x^j$-derivatives of every $\phi\in\F$, $\{u^a_{j_1j_2}\}$ all $x^{j_1}x^{j_2}$-derivatives of all the fields $\phi$, and so on. Of course, those coordinates are all independent. The function $j^\infty(\phi)$ maps $p=(x^1,\ldots,x^n)$ to the coordinates of $J^\infty\!E$ that ``match'' at all orders i.e.~$\{u^a\}$ are the values of $\phi(p)=(\phi^1(p),\ldots,\phi^r(p))$, $\{u^a_j\}$ its 1st-derivatives, $\{u^a_{j_1j_2}\}$ its 2nd-derivatives, and so on. We can write this as
		\begin{equation}\label{eq: j infty phi}
		(j^\infty(\phi))(x^i):=j_p^\infty\phi=\left(x^i=p^i,u^a=\phi^a(p),u^a_j=\frac{\partial \phi^a}{\partial x^j}(p),u^a_{j_1j_2}=\frac{\partial^2\phi^a}{\partial x^{j_1}\partial x^{j_2}}(p),\ldots\right)
		\end{equation}
		This is analogous to what happens if we consider $TM$, which gathers all possible points and velocities, and a curve $\gamma:I\to M$\!. There exists a natural lift $\sigma:I\to TM$ given by $\sigma(t)=(\gamma(t),\dot{\gamma}(t))$ in which the second component of $\sigma$ ``matches'' the velocity of the first component.\vspace*{2ex}
		
		As a final remark, notice that $E$ can be considered as the base manifold of the bundles $\pi^k_E:J^kE\to E$ where $\pi^k_E(j_p^k(\phi))=\phi(p)\in E_p$.
		
		\subsection[Geometric constructions over \texorpdfstring{$J^\infty\!E$}{JE}]{Geometric constructions over \texorpdfstring{$\boldsymbol{J^\infty\! E}$}{JE}}
		\subsubsection*{Generalized vector fields}
		From our discussion of section \ref{section: motivation jets}, it follows that $E$ can be viewed as formed by elements $e=\phi(p)\in E_p$ for $p\in M$ and $\phi\in\F$. Thus, a vector field $\VV\in\campos(\F)$ can be associated with a vector field $V:E\to TE$ on $E$ given by $V_e\equiv V_{\!\phi(p)}:=\VV_{\!\phi}(p)\in T_{\phi(p)}E_p$. This is equivalent to say that $V$ is a vector field over $\pi^0_E(=\mathrm{Id})$. That is, that the following diagram commutes\vspace*{-2ex}
		
		\begin{equation}\label{eq: vector field over pi0}
		\begin{tikzcd}[row sep = scriptsize]
		& TE\arrow[d,"{\pi^1_E}"]\\
		E\arrow[ur,"{V}"]\arrow[r,"{\pi^0_E}"']& E
		\end{tikzcd}
		\end{equation}
		We recall that $E$ has the information of degree $0$, so we have that $V$ is a vector field that only depends on the ``degree $0$ information''. We can generalize this dependence to allow that $V$ depends on ``higher-order information''. Thus, we define a \textbf{generalized vector field} of $E$ as a map $V:J^\infty\!E\to TE$ such that the following diagram commutes\vspace*{-2ex}
		
		\begin{equation}\label{eq: generalized vector field E}
		\begin{tikzcd}[row sep = scriptsize]
		& TE\arrow[d,"{\pi^1_E}"]\\
		J^\infty\!E\arrow[ur,"{V}"]\arrow[r,"{\pi^\infty_E}"']& E
		\end{tikzcd}
		\end{equation}
		We denote as $\mathfrak{X}^\infty_{\mathrm{gen}}(E)$ the set of generalized vector fields of $E$. Of course, there is a natural inclusion $\mathfrak{X}(E)\subset\mathfrak{X}^\infty_{\mathrm{gen}}(E)$. Likewise, we can define $\mathfrak{X}^\infty_{\mathrm{gen}}(M)$ as the set of \textbf{generalized vector fields} on $M$\!, maps $\xi:J^\infty\!E\to TM$ such that the following diagram commutes\vspace*{-1ex}
		
		\begin{equation}\label{eq: generalized vector field M}
		\begin{tikzcd}[row sep = scriptsize]
		& TM\arrow[d,"{\pi^{T\!M}_M}"]\\
		J^\infty\!E\arrow[ur,"{\xi}"]\arrow[r,"{\pi^\infty_M}"']& M
		\end{tikzcd}
		\end{equation}
		\begin{definition}\mbox{}
			\begin{itemize}
				\item A vector field $V$ over $E$ such that $(\pi^1_M)_*V=0$ is called \textbf{vertical}.
				\item A generalized vector field $V$ on $E$ which is also vertical is called an \textbf{evolutionary vector field}. We denote $\mathfrak{X}^\infty_{\mathrm{ev}}(E)$ the set of evolutionary vector fields.
			\end{itemize}
		\end{definition}
		
		In some adapted coordinates $\{x^i,u^a,u^a_J\}$ of $J^\infty\!E$, an evolutionary vector field of $E$ is of the form
		\begin{equation}\label{eq: evolutionary vector field}
		V=V^b[x^i,u^a,u^a_J]\frac{\partial}{\partial u^b}
		\end{equation}
		It is vertical because it has no ``horizontal'' components $V^j\partial/\partial x^j$. It is generalized because $V^b$ depends on elements of higher order $u^a_J$. Their importance stems from the fact that they provide the derivative in the direction of the fields i.e.~the usual variational calculus. Moreover, they are essential to obtain generalized symmetries and to formalize Noether's theorem as we will see at the end of this section.
		
		\subsubsection*{Prolongations and total vector fields}
		In the previous section we have defined generalized vector fields $V:J^\infty\!E\to TE$ and $\xi:J^\infty\!E\to TM$\!. Although they depend on $J^\infty\!E$, they are vector fields of $E$ and $M$ respectively. This section is devoted to constructing natural vector fields over $J^\infty\!E$.\vspace*{2ex}
		
		First recall that given a (local) section $\phi\in\F$, we have the map $j^\infty\!\phi:M\to J^\infty\!E$ given by \eqref{eq: j infty phi}. Therefore, its pushforward $(j^\infty\!\phi)_*:TM\to T(J^\infty\!E)$ allows us to lift vectors from $M$ to $J^\infty\!E$.
		
		\begin{definition}\mbox{}\\
			Given $\xi\in\mathfrak{X}(M)$ and $\phi\in\F$, we define the $\boldsymbol{\phi}$\textbf{-lift} of $\xi$ as $(j^\infty\!\phi)_*\xi:(j^\infty\!\phi)(M)\subset J^\infty\!E\to T(J^\infty\!E)$.
		\end{definition}
		As $\phi$ is a section, the projection of $(j^\infty\!\phi)_*\xi$ over $M$ is again $\xi$, so it is a true lift of $\xi$. Consider some coordinates $\{x^i,u^a,u^a_J\}$ and a curve $x(t)=(x^i(t))$ such that $x(0)=p$ and  $\xi_p:=\partial_t|_0x$, then
		\begin{align*}
		(j^\infty\!\phi)_*\xi_p&=\left.\frac{\d}{\d t}\right|_0\!(j^\infty\!\phi)(x^i(t))\overset{\eqref{eq: j infty phi}}{=}\left.\frac{\d}{\d t}\right|_0\!\left(x^i(t),u^a(t)=\phi^a(x^i(t)),u^a_j(t)=\frac{\partial \phi^a}{\partial x^j}(x^i(t)),\ldots\right)=\\
		&=\left(\xi^i_p,\frac{\partial\phi^a}{\partial x^m}(x^i(0))\left.\frac{\d}{\d t}\right|_0\! x^m(t),\frac{\partial^2\phi^a}{\partial x^j\partial x^m}(x^i(0))\left.\frac{\d}{\d t}\right|_0\! x^m(t),\ldots\right)=\left.\Big(\delta^i_m,u^a_m(p),u^a_{Jm}(p)\Big)\right|_\phi\xi^m_p
		\end{align*}
		Acting over a smooth function $G:J^\infty\!E\to \R$, we obtain
		\begin{align*}
		(j^\infty\!\phi)_*\xi_p(G)&=\d G((j^\infty\!\phi)_*\xi_p)=\\
		&=\frac{\partial G}{\partial x^i}\d x^i\Big((j^\infty\!\phi)_*\xi_p\Big)+\frac{\partial G}{\partial u^a}\d u^a\Big((j^\infty\!\phi)_*\xi_p\Big)+\frac{\partial G}{\partial u^a_J}\d u^a_J\Big((j^\infty\!\phi)_*\xi_p\Big)=\\
		&=\left(\xi^i_p\frac{\partial}{\partial x^i}+\xi^m_p(u^a_m|_\phi)\frac{\partial}{\partial u^a}+\xi^m_p(u^a_{J\cup\{m\}}|_\phi)\frac{\partial}{\partial u^a_J}\right)\!(G)
		\end{align*}
		Thus, the $\phi$-lift of $\xi$ can be written as $(j^\infty\!\phi)_*\xi=\xi^i(D_i|_\phi)$ where
		\begin{equation}\label{eq: Di}\definicion
		\tikzmarkin{Diprima}(0.25,-0.42)(-0.25,0.6)
		\tikzmarkin{Di}(0.2,-0.42)(-0.2,0.6)
		D_i:=\frac{\partial}{\partial x^i}+u^a_i\frac{\partial}{\partial u^a}+u^a_{J\cup\{i\}}\frac{\partial}{\partial u^a_J}\tikzmarkend{Di}
		\tikzmarkend{Diprima}
		\end{equation}
		is a vector field of $J^\infty\!E$ called the \textbf{total $\boldsymbol{i}$-derivative}. This is actually very similar to the well-known material derivative of fluid mechanics $D=\partial_t+\vec{u}\cdot{}\nabla=\partial_t+u^x\partial_x+u^y\partial_y+u^z\partial_z$. Locally, the $\{D_i\}$ span a subspace $\mathcal{H}$ of ``horizontal'' vector fields (they are not vertical). $\mathcal{H}$ in fact provides the canonical split
		\begin{equation}\label{eq: decomposition Tinfty E}
		T(J^\infty\!E) = \mathcal{H}\oplus \mathcal{V}\,,
		\end{equation}
		where $\mathcal{V}$ is the bundle of $\pi^\infty_M$-vertical vectors of $J^\infty\!E$. This decomposition is not possible on any $J^k\E$. This is another important reason to work with the infinite-dimensional manifold $J^\infty\!E$. Notice that $\mathcal{H}$ defines a connection on the bundle $\pi^\infty_M:J^\infty\!E\to M$, which turns out to be flat. In order to properly define $\mathcal{H}$, we need to introduce the following concept:
		
		\begin{definition}\mbox{}\\
			A form $\alpha\in\Omega(J^\infty\!E)$ is a \textbf{contact form} if $(j^\infty\!\phi)^* \alpha=0$ for every $\phi\in\F$. We denote the set of contact forms as $\mathcal{C}^*(J^\infty\!E)$.
		\end{definition}
		The most important example of $1$-contact form, in coordinates $\{x^i,u^a,u^a_J\}$, is
		\[\theta^a_J:=\d u^a_J-u^a_{J\cup\{m\}}\d x^m\qquad\longrightarrow\qquad\theta^a_J\big((j^\infty\!\phi)_*\xi\big)=\left.\left(u^a_{Jm}-\frac{\partial u^a_J}{\partial x^m}\right)\right|_\phi\xi^m\overset{\eqref{eq: j infty phi}}{=}0\]
		In fact, these contact forms allow us to identify which sections of $J^\infty\!E\to M$ come from a section $\phi$ of $E\to M$ because they force the coordinates to ``match'' as we explained when we derived equation \eqref{eq: j infty phi}. In fact, loosely speaking, they provide a true chain rule in $M$
		\[\d u^a_J\ ``="\ \frac{\partial u^a_J}{\partial x^m}\d x^m\]
		Recall that the coordinates are independent so, in general, this is not true (hence the quotation marks). It is then interesting to consider vector fields of $J^\infty\!E$ that belong to the kernel of theses contact forms. In that way, we obtain total derivatives that satisfy the chain rule.
		\begin{definition}\mbox{}\\
			The \textbf{total vector field} of  $\xi\in\mathfrak{X}^\infty_{\mathrm{gen}}(M)$ is the (unique) vector field $\mathrm{tot}(V)\in\mathfrak{X}(J^\infty\!E)$ such that
			\begin{enumerate}
				\item It is a lift of $\xi$ i.e.~$(\pi^\infty_M)_*(\mathrm{tot}\,\xi)=\xi$.
				\item It annihilates all contact $1$-forms i.e.~$\iota_{\mathrm{tot}(\xi)} \beta =0$ for every $\beta\in\mathcal{C}^1( J^\infty\!E)$.
			\end{enumerate}
		\end{definition}
		As its name suggested, it can be proved that the total $i$-derivative $D_i$ is the total derivative of $\partial_i$:
		\[
		D_i:=\mathrm{tot}\left(\frac{\partial}{\partial x^i}\right)\qquad\longrightarrow\qquad \mathrm{If}\ \ \xi=\xi^i\frac{\partial}{\partial x^i}\ \ \text{then}\ \ \mathrm{tot}(\xi)=\xi^i D_i
		\]
		Total vector fields are elements of $\mathfrak{X}(J^\infty\!E)$ lifted from elements of $\mathfrak{X}^\infty_{\mathrm{gen}}(M)$ annihilating the contact $1$-forms (and they actually generate all the vector fields killing the contact forms). We proceed now to define elements of $\mathfrak{X}(J^\infty\!E)$ lifted from elements of $\mathfrak{X}^\infty_{\mathrm{gen}}(E)$ and preserving the contact forms.
		
		\begin{definition}\mbox{}\\
			The \textbf{prolongation} of  $V\in\mathfrak{X}^\infty_{\mathrm{gen}}(E)$ is the (unique) vector field $\prolong(V)\in\mathfrak{X}(J^\infty\!E)$ such that
			\begin{enumerate}
				\item It is a lift of $V$ i.e.~$(\pi^\infty_E)_*(\prolong\,V)=V$.
				\item It preserves $\mathcal{C}^*(J^\infty\!E)$ i.e.~$\L_{\prolong(V)}\beta\in\mathcal{C}^*( J^\infty\!E)$ for every $\beta\in\mathcal{C}^*( J^\infty\!E)$.
			\end{enumerate}
		\end{definition}
		The idea behind this definition is that the flow $\varphi_t:E\to E$ of a vector field $V\in\mathfrak{X}^\infty_{\mathrm{gen}}(E)$ can be lifted, or prolonged, to a flow $\prolong(\varphi_t):J^\infty\!E\to J^\infty\!E$ (see \cite{olver2000applications,anderson1989variational} for more details). The vector field associated to $\prolong(\varphi_t)$ is precisely $\prolong(V)$. This equivalent definition, although geometrically more clear, it is less useful in practice because, in general, it is quite hard to compute the prolongation of a vector field. However, for $V\in\mathfrak{X}^\infty_{\mathrm{ev}}(E)$, we have a simple expression. If $V=V^a\partial/\partial u^a$ in some coordinates $\{x^i,u^a,u^a_J\}$, then
		\begin{equation}\label{eq: prolong}\prolong(V)=\sum_{|J|=0}^\infty (D_JV^a)\frac{\partial}{\partial u^a_J}\end{equation}
		The fact that the sum is infinite posses no convergence problem. This is so because for any vector field $W\in\mathfrak{X}(J^\infty\!E)$ and any real map $f\in\mathcal{C}^\infty(J^\infty\!E)$, the map $W(f)\in\mathcal{C}^\infty(J^\infty\!E)$ involves only finitely many terms i.e.~$f$ depends only on finitely many $u^a_J$.
		
		\subsubsection*{De Rham complex}
		The exterior derivative of the de Rham complex $(\Omega(J^\infty\!E),\d)$ decomposes into its horizontal and vertical part $\d=\peqsub{\d}{H}+\peqsub{\d}{V}$. Likewise, with the help of \eqref{eq: decomposition Tinfty E}, we have the following direct sum decomposition which is only valid for $k=\infty$
		\begin{equation}
		\Omega^p(J^\infty\!E)=\bigoplus_{r+s=p} \Omega^{(r,s)}(J^\infty\!E)\label{decompo}
		\end{equation}
		and a bicomplex diagram analogous to \eqref{eq: bicomplex}. Notice that, as $\d^2=0$, the maps
		\[
		\quad \peqsub{\d}{H}: \Omega^{(r,s)}(J^\infty\!E)\rightarrow \Omega^{(r+1,s)}(J^\infty\!E)\qquad\qquad\qquad \peqsub{\d}{V}: \Omega^{(r,s)}(J^\infty\!E)\rightarrow \Omega^{(r,s+1)}(J^\infty\!E)
		\]
		satisfy $\peqsub{\d}{H}^2=0$, $\peqsub{\d}{V}^2=0$, and $\peqsub{\d}{V}\peqsub{\d}{H}+\peqsub{\d}{H}\peqsub{\d}{V}=0$.
		
		\begin{theorem}[Evolutionary vector fields]\label{theorem: ev vector properties}\mbox{}\\
			Given $W\in\mathfrak{X}^\infty_{\mathrm{ev}}(E)$ and $\beta \in \Omega^{(r,s)}(J^\infty\!E)$, then $\L_{\prolong(W)} \beta \in \Omega^{(r,s)}(J^\infty\!E)$ and the following properties hold
			\[\begin{array}{lll}
			\L_{\prolong(W)} =\peqsub{\d}{V} \iota_{\prolong(W)}+\iota_{\prolong(W)}\peqsub{\d}{V} & \quad\qquad&\L_{\prolong(W)}\peqsub{\d}{H}=\peqsub{\d}{H}\L_{\prolong(W)} \\[2ex]
			0=\peqsub{\d}{H} \iota_{\prolong(W)} +\iota_{\prolong(W)}\peqsub{\d}{H} &&\L_{\prolong(W)}\peqsub{\d}{V}=\peqsub{\d}{V}\L_{\prolong(W)}
			\end{array}\]
		\end{theorem}
		We proceed to define some operators with the help of coordinates. For an intrinsic definition see \cite{anderson1989variational}.
		\begin{definition}\mbox{}\\
			A \textbf{total differential operator} of order $k\in\N$ is a map $P: \mathfrak{X}^\infty_{\mathrm{ev}}(E)\rightarrow \Omega^{(r,s)}(J^\infty\!E)$ that, locally, has the form
			\[\definicion
			\tikzmarkin{Diffoperatorprima}(0.25,-0.55)(-0.25,0.7)
			\tikzmarkin{Diffoperator}(0.2,-0.55)(-0.2,0.7)
			P(W) =\sum_{|J|=0}^k(D_JW^a) P^J_a\qquad P^J_\alpha \in \Omega^{(r,s)}(J^\infty\!E)
			\tikzmarkend{Diffoperator}
			\tikzmarkend{Diffoperatorprima}
			\]
			with $P^J_a\neq0$ for $|J|=k$.
		\end{definition}
		This definition is somewhat dual to equation \eqref{eq: prolong}. 
		
		\begin{theorem}[Integration by parts]\label{integracionpartes}\mbox{}\\
			Let $P: \mathfrak{X}^\infty_{\mathrm{ev}}(E)\rightarrow \Omega^{(n,s)}(J^\infty\!E)$ be a $k$-order total differential operator with $k\leq 2$. Then, there exists a unique globally defined, zeroth order operator $Q: \mathfrak{X}^\infty_{\mathrm{ev}}(E)\rightarrow \Omega^{(n,s)}(J^\infty\!E)$ and a  globally defined first order differential operator $R: \mathfrak{X}^\infty_{\mathrm{ev}}(E)\rightarrow \Omega^{(n-1,s)}(J^\infty\!E)$ satisfying
			\begin{equation}\label{partes}
			\comentario
			\tikzmarkin{descompPprima}(0.2,-0.28)(-0.2,0.5)
			\tikzmarkin{descompP}(0.2,-0.23)(-0.2,0.45)
			P(W)=Q(W)-\peqsub{\d}{H} R(W)
			\tikzmarkend{descompP}
			\tikzmarkend{descompPprima}
			\end{equation}
		\end{theorem}
		Equation \eqref{partes} also holds globally for $k\geq 3$, but the operator $R$ cannot be canonically constructed from $P$. As a corollary, it is possible to derive a global first variational formula: given $L\in \Omega^{(n,0)}(J^\infty\!E)$, then $\iota_{\prolong(W)}\peqsub{\d}{V} L\in \Omega^{(n,0)}(J^\infty\!E)$ defines a differential operator. In \cite{anderson1989variational} it is proved that
		\begin{equation}
		\iota_{\prolong(W)}\peqsub{\d}{V} L = \iota_{\prolong(W)} E(L)-\peqsub{\d}{H} (\iota_{\prolong(W)} \Theta)\overset{\eqref{theorem: ev vector properties}}{=}\iota_{\prolong(W)} \Big(E(L)+\peqsub{\d}{H} \Theta\Big)
		\end{equation}
		This formula is only valid for the prolongation of evolutionary vector fields. In order to remove this dependence, we use that, for $s\geq 1$, we can obtain the following decomposition 
		\[\Omega^{(n,s)}(J^\infty\!E)=I\big(\Omega^{(n,s)}(J^\infty\!E)\big) \oplus \peqsub{\d}{H}\big( \Omega^{(n,s)}(J^\infty\!E)\big)\]
		where $I$ is the so called interior Euler operator \cite[page 45]{anderson1989variational}. Therefore 
		\begin{equation}\label{eq: first variation}
		\comentario
		\tikzmarkin{descompdVLprima}(0.2,-0.28)(-0.2,0.5)
		\tikzmarkin{descompdVL}(0.2,-0.23)(-0.2,0.45)
		\peqsub{\d}{V} L = E(L)+\peqsub{\d}{H}  \Theta
		\tikzmarkend{descompdVL}
		\tikzmarkend{descompdVLprima}
		\end{equation}
		globally on $J^\infty\!E$. $E(L)$ is related to the EL equations. To see how, let us first make a small digression.\vspace*{2ex}
		
		If we consider the subset $\mathcal{E}'\subset J^\infty\!E$ where $E(L)$ vanishes, we are just considering some algebraic conditions over the coordinates of $e\in\mathcal{E}'$. This is not enough as the following example shows: consider the algebraic equation $u_x=u$ with $E=\R\times\R$. This corresponds to the differential equation $u'=u$ which, in turn, implies $u''=u'$, $u'''=u''$, and so on. However, algebraically $u_x=u$ does no imply $u_{xx}=u_x$, $u_{xxx}=u_{xx}$, etc. because the coordinates are independent. That is why those differential consequences have to be included by hand. We define the \textbf{space of solutions} of the theory given by $L$ as
		\[\mathcal{E}=\{e\in J^\infty\!E\ /\  E(L) \text{ and all its differential consequences vanish at } e\}\]
		
		with the inclusion $\jmath_L:\mathcal{E}\hookrightarrow J^\infty\!E$. If we consider a manifold with boundary, once we have \eqref{eq: first variation}, we could apply the same argument that led to equation \eqref{eq: dl+Theta}. It is important to realize that, in general, the operator $\overline{q}$ analogous to $Q$ (or $\overline{b}$ in the notation of section \ref{section: variations}) is not defined over $J^\infty\!(\partial E)$ but over a slightly more general bundle over $\partial M$\!. Indeed, some boundary conditions involve derivatives in directions transversal to $\partial M$ while $J^\infty\!(\partial E)$ can only account for those tangent to $\partial M$\!.
		
		\subsubsection*{Symmetries and Noether's theorem}
		A \boldsymbol{$\peqsub{\d}{H}$}\textbf{-symmetry} (or infinitesimal variational symmetry) of the Lagrangian $L\in \Omega^{(n,0)}(J^\infty\!E)$ is an evolutionary vector field $W\in\mathfrak{X}_{\mathrm{ev}}^\infty(E)$ satisfying
		\begin{equation}\label{eq: infsymm1}
		\definicion
		\tikzmarkin{infsymmetryprima}(0.25,-0.23)(-0.25,0.4)
		\tikzmarkin{infsymmetry}(0.2,-0.23)(-0.2,0.4)
		\mathcal{L}_{\prolong(W)}L=\peqsub{\d}{H} S_W^L
		\tikzmarkend{infsymmetry}
		\tikzmarkend{infsymmetryprima}
		\end{equation}
		for some $S_W^L\in\Omega^{(n-1,0)}(J^\infty\!E)$. In that case, we define the $\boldsymbol{W}$\!\textbf{-current} as
		\begin{equation}\label{eq: Wcurrent}
		\definicion
		\tikzmarkin{Wcurrentprima}(0.25,-0.23)(-0.25,0.43)
		\tikzmarkin{Wcurrent}(0.2,-0.23)(-0.2,0.43)
		J^\Theta_W :=S_W^L+\iota_W\Theta \in \Omega^{(n-1,0)}(J^\infty\!E)
		\tikzmarkend{Wcurrent}
		\tikzmarkend{Wcurrentprima}
		\end{equation}
		
		\begin{theorem}[Noether's theorem]\label{theorem: NoetherBicomplex}\mbox{}\\
			The $W$\!-current is conserved  over the space of solutions $\mathcal{E}$:
			\[\peqsub{\d}{H} (\jmath_L^*J_W^\Theta)=0\] 
		\end{theorem} 
		\begin{proof}\mbox{}
			\[\peqsub{\d}{H} J^\Theta_W\updown{\eqref{eq: Wcurrent}}{\eqref{eq: infsymm1}}{=}\mathcal{L}_{\prolong(W)}L+\peqsub{\d}{H}\iota_{\prolong(W)}\Theta\overset{\eqref{theorem: ev vector properties}}{=}\iota_{\prolong(W)}\Big(\peqsub{\d}{V} L-\peqsub{\d}{H} \Theta\Big)\overset{\eqref{eq: first variation}}{=}\iota_{\prolong(W)} E\]
			\mbox{}\vspace*{-7ex}
			
		\end{proof}
		
		\subsection{Connection with the standard physical formalism}\label{section: connection previous formalism}
		In the main body of the article we have dealt with objects in $M\times\F$ (the following discussion applies as well for the relative version). For that, we have considered that $\F$ has nice properties and we have performed computations in the usual fashion. However, we have just learned that the proper way is to consider $\phi(p)\in E_p$ instead of $\phi\in\F$, which allows us, in particular, to consider higher-order derivatives with the help of the jet bundles (it is not always clear to what space the derivatives of $\phi$ belong to). The manipulations are roughly speaking the same, but conceptually both approaches are very different.\vspace*{2ex}
		
		Let us focus for a moment in the forms $\OOmega^{(r,s)}(M\times \F)$ over $M\times\F$. We have the (horizontal) exterior derivative $\d$ of $M$ and the (vertical) exterior derivative $\dd$ of $\F$ that define a natural bigraded structure: the one associated with the product structure of $M\times \F$ and with the exterior derivative $\mathbf{d}:=\d+\dd$. From $\mathbf{d}^2=0$ we deduce that $\d \dd=-\dd \d$ which is \eqref{eq: d dd=dd d} upon considering the change of sign mentioned after \eqref{eq: d dd=dd d}. The same happens with the interior product with vertical fields and the horizontal exterior derivative, which anti-commute according to \ref{theorem: ev vector properties}. Of course, this is just a matter of convention and the important results remain the same.\vspace*{2ex}
		
		In order to connect this bigraded complex $\OOmega(M\times \F)$ with the de Rham complex $\Omega(J^\infty\!E)$, we use the evaluation map $\mathrm{Eval}^\infty: M\times \F\rightarrow J^\infty\!E$ given by $\mathrm{Eval}^\infty(p,\phi)=j_p^\infty(\phi)$ \cite{zuckerman1987action}. It allows us to define the sub-bicomplex  of $\OOmega(M\times \F)$ 
		\[\OOmega_{\mathrm{loc}}(M\times \F):=(\mathrm{Eval}^\infty)^*\Omega(J^\infty\!E)\qquad\longrightarrow\qquad\OOmega^p_{\mathrm{loc}}(M\times \F):= \bigoplus_{r+s=p} \OOmega^{(r,s)}_{\mathrm{loc}}(M\times \F)\]
		From \cite{takens1979global} we have that $(\mathrm{Eval}^\infty)^*:\Omega(J^\infty\!E)\to\OOmega( M\times \F)$ is injective, so $\Omega(J^\infty\!E)$ is isomorphic to its image, $\OOmega_{\mathrm{loc}}(M\times\F)$. This identification provides a dictionary to rewrite this paper in the jet language. For instance, $\dd\phi^a$ can be properly written, in some coordinates $\{x^i,u^a,u^a_J\}$, as $\mathbf{d}u^a$. Likewise, a Lagrangian $L\in\OOmega^{(n,0)}(M\times\F)$ can be understood as a horizontal element $L\in\Omega^n(J^\infty\!E)$ i.e.~in coordinates it is of the form $L=f(x^i,u^a,u^a_J)\d x^1\wedge\cdots\wedge \d x^n$ (with no $\d u$ term). The action $\SS\in\Omega^0(\F)$ and the symplectic form $\Omega_\phi\in \Omega^2(\F)$ are given by
		\begin{equation}\label{eq: S(phi)=int j^*phi L}
		S(\phi)=\int_M(j^\infty\!\phi)^*L\qquad\qquad\qquad\qquad\Omega^\imath_\phi=\int_\Sigma \imath^*(j^\infty\!\phi)^*\peqsub{\d}{V}\Theta
		\end{equation}
		The relevant formulas for the computations, in the $\{x^i,u^a,u^a_J\}$  coordinates, are
		\begin{align}
		&\peqsub{\d}{H} F=(D_iF)\mathbf{d} x^i && \peqsub{\d}{V} F=\frac{\partial F}{\partial u^a_J}\theta^a_J\\
		&\peqsub{\d}{H} x^i=\mathbf{d}x^i && \peqsub{\d}{V}x^i=0\\
		&\peqsub{\d}{H}u^a_J=u^a_{J\cup\{i\}}\mathbf{d}x^i &&\peqsub{\d}{V}u^a_J=\theta^a_J\\
		&\peqsub{\d}{H}\theta^a_I=\mathbf{d}x^i\wedge\theta^a_{I\cup\{i\}}&&\peqsub{\d}{V}\theta^\alpha_I=0\end{align}
		
		\newpage
		
		\section{Some important results}\label{Appendix: results}
		
		\subsection{Stokes' theorems}\label{section stokes theorems}
		
		In this section we assume that $M$ is an oriented and connected $n$-manifold with boundary $\partial M\overset{\jmath}{\hookrightarrow}M$, possibly empty, with the induced orientation.
		
		\centerline{
			\begin{tabular}[t]{|p{.5\textwidth}|p{.5\textwidth}|}\hline
				\begin{minipage}[t]{\linewidth} 
					\textbf{Stokes}\rule{0ex}{3ex}\\[0.2ex]
					If $\partial M=\varnothing$ and $\omega=\d\alpha\in\Omega_c^{n}(M)$, then
					\[\int_M\omega=0\]
				\end{minipage} 
				&
				\begin{minipage}[t]{\linewidth}
					\textbf{Inverse Stokes}\rule{0ex}{2.5ex}\\[0.2ex]
					If $\partial M=\varnothing$ and $\omega\in\Omega_c^n(M)$ such that
					\[\int_M\omega=0\]
					then $\exists\alpha\in\Omega_c^{n-1}(M)$ with $\omega=\d \alpha$.
				\end{minipage}
				\mbox{}\vspace*{0.1ex}\mbox{}
				\\ \hline
				\begin{minipage}[t]{\linewidth}
					\textbf{Stokes with boundary}\rule{0ex}{2.5ex}\\[0.2ex]
					If $\omega=\d\alpha\in\Omega_c^{n-1}(M)$, $\overline{\beta}=\jmath^*\alpha-\d\overline{\gamma}\in\Omega_c^{n-1}(\partial M)$, then
					\[\int_M\omega=\int_{\partial M}\overline{\beta}\]		
				\end{minipage}&
				\begin{minipage}[t]{\linewidth}
					\textbf{Inverse Stokes with boundary}\rule{0ex}{2.5ex}\\[0.2ex]
					If $\omega\in\Omega_c^n(M)$ and $\overline{\beta}\in\Omega_c^{n-1}(\partial M)$ are such that
					\[\int_M\omega=\int_{\partial M}\overline{\beta}\]
					then $\exists\alpha\in\Omega_c^{n-1}(M),\overline{\gamma}\in\Omega_c^{n-2}(\partial M)$  with $\omega=\d \alpha$ and $\overline{\beta}-\jmath^*\alpha=\d\overline{\gamma}$.
				\end{minipage}
				\mbox{}\vspace*{0.1ex}\mbox{}
				\\\hline
			\end{tabular}
		}
		\mbox{}
		
		The Stokes' theorems with and without boundary are standard results. The top-right theorem follows from the isomorphism $\int_M:H^n(M)\to\R$ given by de Rham's theorem \cite{weintraub2014differential}. The last one is a consequence of the fact that if $M$ has non-empty boundary, then $H^n(M)=0$ \cite[8.4.8]{weintraub2014differential}. Thus $\omega$ is exact and we can apply de Rham's theorem to the boundary. Following the theory developed in section \ref{section: geometry Mxpartial M}, we can state their relative versions.
		
		\centerline{
			\begin{tabular}[t]{|p{.5\textwidth}|p{.5\textwidth}|}\hline
				\begin{minipage}[t]{\linewidth} 
					\textbf{Relative Stokes}\rule{0ex}{2.5ex}\\[0.2ex]
					If $(\omega,\overline{\beta})=\underline{\d}(\alpha,\overline{\gamma})\in\Omega_c^{n}(M,\partial M)$, then
					\[\int_{(M,\partial M)}(\omega,\overline{\beta})=0\]	
				\end{minipage} 
				&
				\begin{minipage}[t]{\linewidth}
					\textbf{Inverse relative Stokes}\rule{0ex}{2.5ex}\\[0.2ex]
					If $(\omega,\overline{\beta})\in\Omega_c^n(M,\partial M)$ such that
					\[\int_M(\omega,\overline{\beta})=0\]
					then $\exists(\alpha,\overline{\gamma})\in\Omega_c^{n-1}(M,\partial M)$ with $(\omega,\overline{\beta})=\underline{\d}(\alpha,\overline{\gamma})$.			
				\end{minipage}
				\mbox{}\vspace*{0.1ex}\mbox{}
				\\ \hline
				\begin{minipage}[t]{\linewidth}
					\textbf{Relative Stokes with boundary}\rule{0ex}{2.5ex}\\[0.2ex]
					If $((\omega,\overline{\beta}),(\alpha,\overline{\gamma}))=\underline{\underline{\d}}((a,\overline{b}),(c,\overline{d}))$ then \[\int_{(M,N)}(\omega,\overline{\beta})=\int_{\underline{\partial}(M,N)}(\alpha,\overline{\gamma})\]
				\end{minipage}&
				\begin{minipage}[t]{\linewidth}
					\textbf{Inverse relative Stokes with boundary}\rule{0ex}{2.5ex}\\[0.2ex]
					If $((\omega,\overline{\beta}),(\alpha,\overline{\gamma}))\in\Omega_c^n((M,N),\underline{\partial}(M,N))$ are such that
					\[\int_{(M,N)}(\omega,\overline{\beta})=\int_{\underline{\partial}(M,N)}(\alpha,\overline{\gamma})\]
					then $\exists((a,\overline{b}),(c,\overline{d}))\in\Omega_c^{n-1}((M,N),\underline{\partial}(M,N))$ with $((\omega,\overline{\beta}),(\alpha,\overline{\gamma}))=\underline{\underline{\d}}((a,\overline{b}),(c,\overline{d}))$.		
				\end{minipage}
				\mbox{}\vspace*{0.1ex}\mbox{}
				\\\hline
			\end{tabular}
		}
		\mbox{}
		
		The first row is equal to the second row of the previous table. The second row of this table follows from remark \ref{remark: pair of pairs}.
		
		\subsection{Cohomological results}
		Consider $M$ a connected and oriented $n$-manifold with boundary $\partial M$ (possibly empty).
		\[\begin{array}{|l|l|l|}\hline
		\rule{0ex}{4ex}H^0(M)\cong\R&\rule{0ex}{4ex}H^0_c(M)\cong\left\{\!\!\begin{array}{ll}\R&\text{if }M\text{ is compact}\\0&\text{if }M\text{ is non-compact}\end{array}\right.&\rule{0ex}{4ex}H^n_c(M)\cong\left\{\!\!\begin{array}{ll}\R&\text{if }\partial M=\varnothing\\0&\text{if }\partial M\neq\varnothing\end{array}\right.\\[2ex]\hline
		\rule{0ex}{3ex}H^k(M\times\R)\cong H^k(M)&\rule{0ex}{3ex}H_c^k(M\times\R)\cong H^{k-1}_c(M)&\rule{0ex}{3ex}H_c^n(M,\partial M)\cong\R\\[.5ex]\hline
		\end{array}\]
		
		Recall that $(M,\varnothing)=M$\!. Moreover, if $M$ is compact, then $H^k_c(M)=H^k(M)$. All these results can be found in \cite{weintraub2014differential} except the last one, which follows from the Lefschetz duality \cite{maunder1996algebraic}. We also have the following important result
		\begin{equation}\label{eq: ismorphism relative contractible}
		H^k(M\times\R,\partial M\times\R)\cong H^k(M,\partial M)
		\end{equation}
		which is the relative analog to $H^k(M\times\R)\cong H^k(M)$. The isomorphism holds if we restrict to those $n$-pair of forms which are integrable over $(M,\partial M)$ but not necessarily over $(M\times\R,\partial M\times\R)$.
		
		\subsection{Other results}\label{section: other results}
		We consider the space of null Lagrangians i.e.~those with no Euler-Lagrange equations
		\[\Lag_{\mathrm{null}}(M)=\Big\{(L,\overline{\ell})\in\Omega^n(J^\infty\!E,J^\infty(\partial E))\ / \ [\underline{\d}_V(L,\overline{\ell})]=0\Big\}\] 
		We denote $H^n_{\mathrm{null}}(M)$ the space formed by the cohomological (horizontal) classes of null Lagrangians.
		
		\begin{theorem}\label{theorem: null lagrangians}\leavevmode
			\begin{itemize}
				\item If $E\to M$ is a bundle over the $n$-manifold $M$, then $H^n_{\mathrm{null}}(M)\cong H^n(E,\partial E)$, where $\partial E\to\partial M$ is the induced bundle.
				\item If $E\to M$ is a contractible bundle over the $n$-manifold $M$, then $H^n_{\mathrm{null}}(M)\cong H^n(M,\partial M)$.
			\end{itemize}
		\end{theorem}
		\begin{proof}\mbox{}\\
			The first result follows adapting the proof of \cite[theorem 5.9]{anderson1989variational} to the relative case. Everything works out as the proof relies on cohomological techniques (see also \cite{takens1979global}). The second point follows from \eqref{eq: ismorphism relative contractible} and the fact that the fibers are contractible.
		\end{proof}

		\begin{theorem}\label{theorem: trivial action trivial lagrangians}\mbox{}\\
			Let $\F$ be the space of sections of a contractible bundle and $(L,\overline{\ell})\in\Lag(M)$. If $(L,\overline{\ell})\equivint 0$, then $[(L,\overline{\ell})]=0$
		\end{theorem}
		\begin{proof}\mbox{}\\
			The action obtained from $(L,\overline{\ell})$ is identically zero, so it provides no EL equation. That means that $[(L,\overline{\ell})]\in H^n_{\mathrm{null}}(M)$. From the previous theorem we have $H^n_{\mathrm{null}}(M)\cong H^n(M,\partial M)$ and the latter space is generated, according to the table in section \ref{section: other results}, just by one element which will be of the form $[(\vol_M,\vol_{\partial M})]\neq0$. Now notice that $(\pi^\infty_M)^*(\vol_M,\vol_{\partial M})\in\Omega^n(J^\infty\!E,J^\infty(\partial E))$ is independent of the fields, so its vertical derivative is zero i.e.~it has no EL equation either. This means that
			\[[(\pi^\infty_M)^*(\vol_M,\vol_{\partial M})]\in H^n_{\mathrm{null}}(M)\]
			Let us prove that it is non-zero. Assume that $(\pi^\infty_M)^*(\vol_M,\vol_{\partial M})$ were exact, then its pullback through $j^\infty\!\phi$ had to be also exact. But this would be a contradiction because
			\[(j^\infty_M)^*(\pi^\infty_M)^*(\vol_M,\vol_{\partial M})=(\pi^\infty_M\circ j^\infty\!\phi)^*(\vol_M,\vol_{\partial M})\overset{\dagger}{=}\mathrm{Id}^*(\vol_M,\vol_{\partial M})=(\vol_M,\vol_{\partial M})\]
			which is non-exact. Notice that in $\dagger$ we have used that $j^\infty\!\phi$ is a section of the bundle $\pi^\infty_M:J^\infty\!E\to M$\!.\vspace*{2ex}
			
			As $H^n_{\mathrm{nul}}(M)$ has dimension $1$, then $(\pi^\infty_M)^*(\vol_M,\vol_{\partial M})$ alone forms a generating system. In particular, there exists some $\alpha\in\R$ such that $[(L,\overline{\ell})]=\alpha[(\pi^\infty_M)^*(\vol_M,\vol_{\partial M})]$. Integrating this expression, which only depends on the cohomology, and using that by hypothesis $(L,\overline{\ell})\equivint0$, leads to
			\begin{align*}
			0&=\int_{(M,\partial M)}(L,\overline{\ell})(\phi)=\int_{(M,\partial M)}\alpha\Big((\pi^\infty_M)^*(\vol_M,\vol_{\partial M})\Big)(\phi)\overset{\eqref{eq: S(phi)=int j^*phi L}}{=}\\
			&=\alpha\int_{(M,\partial M)}(j^\infty\!\phi)^*(\pi^\infty_M)^*(\vol_M,\vol_{\partial M})=\alpha\int_{(M,\partial M)}\Big(\pi^\infty_M\circ j^\infty\!\phi\Big)^*(\vol_M,\vol_{\partial M})=\\
			&=\alpha\int_{(M,\partial M)}\mathrm{Id}^*(\vol_M,\vol_{\partial M})=\alpha\int_{(M,\partial M)}(\vol_M,\vol_{\partial M})
			\end{align*}
			The last integral is non-zero applying the inverse relative Stokes' theorem, stated in section \ref{section stokes theorems}, to  $[(\vol_M,\vol_{\partial M})]\neq0$. Therefore, $\alpha=0$ which shows that $[(L,\overline{\ell})]=0$. 
		\end{proof}
	\end{appendices}

	\bibliographystyle{plainnat}
	
	\small
	\bibliography{bibligraphy}
\end{document}